\documentclass[11pt,a4paper,reqno]{amsart}
\usepackage[hmargin={2.7cm,2.7cm},vmargin={2.5cm,2.5cm},centering]{geometry}
\usepackage{amssymb,amsmath,amsthm,amsfonts,amscd}

\let\a=\alpha \let\b=\beta  \let\g=\gamma  \let\d=\delta
\let\e=\varepsilon
\let\z=\zeta  \let\h=\eta   \let\th=\theta \let\k=\kappa \let\l=\lambda
\let\m=\mu    \let\n=\nu    \let\x=\xi         \let\r=\rho
\let\s=\sigma \let\t=\tau    \let\ph=\varphi
   \let\o=\omega
\let\G=\Gamma   \let\L=\Lambda \let\X=\Xi
         
\let\O=\Omega

\def\EE{{\cal E}}

 \def\pp{{\bf p}}
 \def\xx{{\bf x}} \def\yy{{\bf y}} 
\def\kk{{\bf k}}

\def\nn{\nonumber}

\newcommand{\SVert}{\Vert\!|}

\def\EE{\mathbb{E}}
\def\\{\hfill\break}
\def\={:=}
\let\io=\infty
\let\0=\noindent

\def\tende#1{\,\vtop{\ialign{##\crcr\rightarrowfill\crcr\noalign{\kern-1pt
    \nointerlineskip} \hskip3.pt${\scriptstyle #1}$\hskip3.pt\crcr}}\,}
\def\otto{\,{\kern-1.truept\leftarrow\kern-5.truept\to\kern-1.truept}\,}

\def\to{\rightarrow}
\def\la{\left\langle}
\def\ra{\right\rangle}
\def\qed{\hfill\raise1pt\hbox{\vrule height5pt width5pt depth0pt}}

\def\lis{\overline}

\def\be{\begin{equation}}
\def\ee{\end{equation}}
\def\bp{\begin{pmatrix}}
\def\ep{\end{pmatrix}}
\def\bea{\begin{eqnarray}}
\def\eea{\end{eqnarray}}
\def\nn{\nonumber}
\def\pref#1{(\ref{#1})}

\usepackage{latexsym}
\usepackage[english]{babel}
\usepackage{fancyhdr}
\usepackage[mathscr]{eucal}
\usepackage{amsmath}
\usepackage{mathrsfs}
\usepackage{amsthm}
\usepackage{amsfonts}
\usepackage{amssymb}
\usepackage{amscd}
\usepackage{bbm}
\usepackage{graphicx}
\usepackage{graphics}
\usepackage{latexsym}
\usepackage{color}
\usepackage{pifont}
\usepackage{color}
\usepackage{tikz}
\usetikzlibrary{patterns, decorations.markings,decorations.pathmorphing,arrows, calc}

\tikzset{snake it/.style={decorate, decoration=snake}}
\newcommand{\ud}{d}

\newcommand{\ii}{\mathrm{i}}

\newcommand{\bx}{{\mathbf{x}}} 
\newcommand{\bn}{{\mathbf{n}}}
\newcommand{\bk}{{\mathbf{k}}}
\newcommand{\bq}{{\mathbf{q}}}
\newcommand{\balpha}{{\boldsymbol{\alpha}}}
\newcommand{\bchi}{{\boldsymbol{\chi}}}
\newcommand{\bpsi}{{\boldsymbol{\psi}}}
\newcommand{\bxi}{{\boldsymbol{\xi}}}
\newcommand{\by}{{\mathbf{y}}}

\newcommand{\ZZ}{\mathbb{Z}}

\newcommand{\btheta}{\boldsymbol{\vartheta}}
\newcommand{\bfzeta}{\boldsymbol{\zeta}}

\theoremstyle{plain}
\newtheorem{theorem}{Theorem}[section]
\newtheorem{lemma}[theorem]{Lemma}
\newtheorem{corollary}[theorem]{Corollary}

\theoremstyle{definition}
\newtheorem{definition}[theorem]{Definition}

\newtheorem{remark}[theorem]{Remark}
\newtheorem*{remark*}{Remark}

\numberwithin{equation}{section}

\usepackage{amsmath,amsthm,amsfonts,amssymb, graphicx,mathabx,epsfig}
\usepackage{bbm}
\usepackage{bm}
\usepackage{graphicx,color}
\usepackage{footmisc}
\usepackage{epstopdf}
\usepackage{dsfont}
\usepackage{pgfplots}
\usepackage{subfig}
\usetikzlibrary{fit}
\usepackage{slashed}
\usepackage{amsaddr}

\makeatletter
\renewcommand{\email}[2][]{%
  \ifx\emails\@empty\relax\else{\g@addto@macro\emails{,\space}}\fi%
  \@ifnotempty{#1}{\g@addto@macro\emails{\textrm{(#1)}\space}}%
  \g@addto@macro\emails{#2}%
}
\makeatother

\renewcommand{\Im}{\mathfrak{Im}}

\title{Universality in the 2d quasi-periodic Ising model and Harris-Luck irrelevance}
\author{Matteo Gallone}
\address{International School for Advanced Studies - SISSA \\ via Bonomea 265 - 34136 Trieste (Italy). \vspace{-0.2cm}}
\email{matteo.gallone@sissa.it}
\author{Vieri Mastropietro}
\address{Dipartimento di Matematica F. Enriques - 
Universit\`a di Milano \\ via C. Saldini 50 - 20129 Milano (Italy).}
\email{vieri.mastropietro@unimi.it}

\date{\today}

\begin{document}

\vspace{-0.4cm}
\maketitle

\vspace{-0.75cm}

\begin{abstract} 
We prove that in the 2D Ising Model with a weak bidimensional quasi-periodic disorder in the interaction, the critical behavior is the same as in the non-disordered case; that is, the critical exponents for the specific heat and energy-energy correlations are identical, and no logarithmic corrections are present. The disorder produces a quasi-periodic modulation of the amplitude of the correlations and a renormalization of the velocities, that is, the coefficients of the rescaling of positions, and of the critical temperature. The result establishes the validity of the prediction based on the Harris-Luck criterion, and it provides the first rigorous proof of universality in the Ising model in the presence of quasi-periodic disorder in both directions and for any angle. Small divisors are controlled assuming a Diophantine condition on the frequencies, and the convergence of the series is proved by Renormalization Group analysis.
\end{abstract}

\vspace{-0.3cm}
\tableofcontents

\section{Introduction }

\subsection{Universality and Harris-Luck criterion}
A certain number of macroscopic properties close to phase transitions show a remarkable independence from microscopic details. In particular, it is both predicted theoretically and observed experimentally that the critical exponents, describing the singularities of thermodynamic functions, are the same in systems sharing only a few general features but having different inter-molecular forces, atomic weights, or lattice structures. This phenomenon is known as \emph{universality}, and the Renormalization Group, introduced by Kadanoff \cite{K} and Wilson \cite{W}, provides an explanation by introducing the concepts of scaling dimension, dimensionally relevant, marginal, or irrelevant interactions, and universality classes. The fact that interactions are dimensionally relevant or marginal \textcolor{black}{does not by itself imply} that they can change the critical behavior; the precise effect on critical exponents is governed by an effective dimension, which can be different from the scaling dimension due to cancellations or other mechanisms.

A paradigmatic model where universality can be investigated is the Ising model, which describes a system of spins with nearest-neighbor interactions and shows a phase transition in dimensions $d \ge 2$ \textcolor{black}{characterized} by certain values of the critical exponents. One can perturb this model with finite-ranged or higher spin interactions, or consider it on different lattices, and ask what happens to the critical behavior. In $d \ge 4$, universality is proven in the context of the closely related $\phi^4$ models (see, \textit{e.g.}, \cite{B} and references therein), where it has been rigorously shown that the values of the exponents are equal to the mean-field ones, \emph{e.g.}, the correlation length exponent is $\nu = 1/2$ and the specific heat exponent $\alpha = (4-d)/2$. We remark, however, that while in \(d \ge 5\) the behavior is exactly the same as in the mean-field theory, in $d = 4$ logarithmic corrections are present; the difference is that in the first case the interaction is irrelevant in the Renormalization Group sense, while in the second it is marginal (or, more precisely, marginally irrelevant).

In $d=2$, the Ising model with nearest-neighbor interaction on a square lattice was solved by Onsager \cite{O}. His solution proves that the value of the critical exponents ($\nu=1$, $\alpha=0$) is different from the \textcolor{black}{ones obtained by approximate methods, such as the mean-field}. With universality in mind, it is natural to ask whether these values are robust under perturbations. One can ask, for example, if the addition of a next-to-nearest neighbor interaction or a non-quadratic one leaves the system in the Onsager universality class or not. In this case, it is not convenient to use $\phi^4$ models, but one can use the representation in terms of Grassmann integrals, at the basis of the exact solution, and analyze it using Renormalization Group methods. This strategy was proposed in \cite{PS1} and applied to the computation of the specific heat and energy correlations in \cite{PS2} and in Appendix N of \cite{M1}. The Grassmann integral representation was then used in \cite{M4,M1} for the case of two Ising models coupled to each other by a quartic interaction, which can be mapped into models like the Eight-vertex, Six-vertex, or the Ashkin-Teller model.

Even if single or coupled Ising models have the same exponents in the absence of quartic interaction, when the interaction is present they belong to different universality classes. In the first case, the interaction is dimensionally irrelevant, implying that, when the strength of the interaction is small enough, the exponents are the same as in the pure Ising model (\textit{e.g.}\ $\nu=1$, $\alpha=0$) and no logarithmic corrections are present. In the second case, the interaction is marginal, and its flow is controlled thanks to the complicated cancellations related to emergent symmetries. The exponents are continuous functions of the strength of the coupling \cite{M1}, verifying suitable Kadanoff extended scaling relations \cite{BFM1,BFM2}. Continuous exponents also appear in the transition between the two universality classes in the Ashkin-Teller model \cite{GM11,Micm}.

Subsequently, the Renormalization Group approach to interacting Ising models was used in the proof of the universality of the central charge \cite{M2}, the scaling limit of all the energy correlations \cite{M3}, and to analyze the role of non-periodic boundary conditions \cite{FG}. Interacting dimer models, which are in the same universality class as coupled Ising models in some parameter regions, were studied in \cite{GMT}. This approach typically requires a small value of the coupling.

Other approaches, different from the Renormalization Group, lead to universality results for the Ising model, like those in \cite{C1,C2}, with nearest-neighbor interactions on different planar graphs. In \cite{C4}, the Ising model with non-planar, or alternatively some non-nearest-neighbor pair interactions, was considered, proving the Gaussianity of correlations without a smallness condition but without providing information on exponents.

Another situation where the issue of universality can be posed in the Ising model is when \textit{disorder} is considered. Disorder can be introduced either in the magnetic field \cite{Az1,Az2,Az3,J} or in the interaction, and we focus here on this second case, for which much less is known at a rigorous level. Typically, one can consider two kinds of disorder in the interaction: random or quasi-periodic. The first describes the effect of impurities, while the second is realized in quasi-crystals or cold atoms experiments. Early investigations were done in the 2D random Ising model; in particular, the Ising model with a layered disorder (that is, constant in one direction) was considered in \cite{MW} (see also \cite{H} and \cite{F36}), and the specific heat was found continuous (instead of logarithmically divergent), while with a bidimensional random disorder, a double logarithmic behavior in the specific heat \cite{DD2} was found.

In more general cases, Harris \cite{H} proposed a criterion to predict when random disorder is irrelevant or not; if $\xi$ is the correlation length and $\Delta^2$ is the covariance of the disorder, the condition for irrelevance is $\sqrt{\Delta^2 / \xi^d} \ll |\beta - \beta_c|$, where the left-hand side is (roughly) the ratio between typical fluctuation of the sum of disorder terms within a distance given by the correlation length $\xi$ and the mean ($\beta_c$ is the critical inverse temperature). As close to criticality $\xi \sim |\beta - \beta_c|^{-\nu}$, with $\nu$ being the critical exponent, irrelevance is predicted for $\nu d/2 > 1$, see \cite{H}, while relevance is expected for $\nu d/2 < 1$. According to this criterion, irrelevance is predicted for $d \ge 5$ ($\nu=1/2 > 2/d$) and relevance for $d=3$ (conformal bootstrap predicts $\nu=0.627 \dots < 2/3$, see \cite{R}). In the marginal cases $d=4$ ($\nu=1/2$) and $d=2$ ($\nu=1$), Harris's criterion gives no predictions in general.

On the rigorous side, a generalization of Harris's result was proved in \cite{Ch}, where it was shown that in all systems with continuous transitions $\tilde{\nu} \ge 2 / d$, with $\tilde{\nu}$ being the index of the disordered system. In the case of layered disorder in $d=2$, the system is effectively one-dimensional as far as the ratio between mean and fluctuations is concerned, so the relevance of disorder is predicted in agreement with \cite{MW}. A rigorous proof is still lacking, despite progress being made in this direction in \cite{Gree}. In addition, the Harris criterion has been verified in simplified models of a probabilistic nature \cite{G3}.

While the Harris criterion regards the case of random hopping, the case of \textit{quasi-periodic} disorder was considered by Luck \cite{L} (Harris-Luck criterion). In the case of the 2D Ising model with layered quasi-periodic disorder, the condition for irrelevance was generalized to ${1/\xi}\sum_{x=0}^\xi \delta_x \ll |\beta-\beta_c|$, where $\delta_x$ is a suitable function measuring the fluctuation of the quasi-periodic hopping, see \cite{L}. Since $\nu=1$, the condition for irrelevance requires that $\sum_{x=0}^\xi \delta_x$ is bounded and small uniformly in $\xi$, a condition verified in the case of weak quasi-periodic modulation, while it is violated for strong quasi-periodic disorder.

Such conjectures were checked in \cite{L} by a perturbative method, but the issue of convergence of the series was not addressed; they have also been confirmed by numerical investigations, see \emph{e.g.}\ \cite{GG,C}. In particular, in \cite{C} it was numerically found that the Ising model with weak quasi-periodic disorder remains in the Onsager class, while evidence of a new universality class is found at stronger disorder. Finite difference equations for the spin correlations have been derived in \cite{Pe1} from which low and high temperatures expansions are obtained.

In this paper, we finally prove that the critical exponents for the specific heat and energy-energy correlations in the weak quasi-periodic Ising model are identical to the Onsager ones, both for layered and non-layered disorder, in agreement with the Harris-Luck criterion. The result is based on convergent series expansions in the disorder, and the small-divisor problem is addressed via Renormalization Group analysis. This provides one of the very few cases in which a rigorous understanding of the critical behavior of the 2D Ising model with disorder is achieved and universality is proven.

\subsection{Main result}

The Hamiltonian of the 2D quasi-periodic Ising model is
\begin{equation} 
H=-\sum_{\bx\in \L_i} \big[
J^{(1)}_{\bx}
 \sigma_\bx \sigma_{\bx+{\bf e}_1}+J^{(0)}_{\bx}
 \sigma_\bx \sigma_{\bx+{\bf e}_0}\big] \,
\label{a1} 
\end{equation}
where ${\bf e}_0=(1,0)$, ${\bf e}_1=(0,1)$, $\bx=(x_0,x_1)$, $\s_\xx=\pm$ and:
\begin{enumerate} 
\item For $i \in \mathbb{N}$,
$\bx\in \Lambda_{ i}$, $\L_{ i}=(-L_{0,i}/2, L_{0,i}/2]\times (-L_{1,i}/2, L_{1,i}/2]\cap \ZZ^2$,
$\sigma_{\bx}=\pm$ and periodic boundary conditions are imposed.
\item
The interaction is given by
\begin{equation}
J^{(j)}_{\bx} =\Big(1+
\lambda \phi^{(j)}(2\pi \omega_{0,i } x_0+\theta_{j,0},
2\pi \omega_{1,i } x_1+\theta_{j,1})\Big) J^{(j)}\, ,\qquad j=0,1 \,
\label{b1}
\end{equation}
where $\phi^{(j)}(\by)$ is such that
\begin{equation}\label{eq:UnoPuntoTre}
\phi^{(j)}(\by)
=\sum_{n_0=-\lfloor L_{0,i}/2\rfloor}^{\lfloor(L_{0,i}-1)/2\rfloor}
\sum_{n_1=-\lfloor L_{1,i}/2\rfloor}^{\lfloor(L_{1,i}-1)/2\rfloor}\hat \phi^{(j)}_{\bn} 
e^{\ii  (n_0 y_0+n_1 y_1)} \, ,
\end{equation}
with $\hat\phi^{(j)}_{\bn}=(\hat\phi^{(j)}_{-\bn})^*$,
$\bn=(n_0,n_1)$ and $\mathbf{y}=(y_0,y_1)$; moreover, for suitable real constants $A, \eta>0$
\begin{equation}\label{eq:UnoPuntoQuattro}
|\hat \phi^{(j)}_{\bn}|\le A e^{- \h |\bn|} \, .
\end{equation}
\item $\{\o_{0,i}\}_{i \in \mathbb{N}},\{\o_{1,i}\}_{i \in \mathbb{N}}$ are the best approximants $\omega_{0,i }=p_{0,i }/q_{0,i }$ 
and $\omega_{1,i }=p_{1,i }/q_{1,i }$ of two irrational numbers
$\o_0,\o_1<1$. For $j=0,1$, the latter are obtained starting from the continuous fraction representation
$
\o_j=a_{j,0}+{1\over a_{j,1}+{1\over a_{j,2}+{1\over a_{j,3}+\cdots}}}$ from which, one has
${p_{j,1}\over q_{j,1}}=a_{j,0}+{1\over a_{j,1}}$, ${p_{j,2}\over q_{j,2}}=a_{j,0}+{1\over a_{j,1}+{1\over a_{j,2}}}$ with  $
\left|\o_j-{ p_{j,i}\over q_{j,i}}\right| \;\le\; { C\over q^2_{j,i}}$ (see \textit{e.g.} Section IV.7 in \cite{D}).
\item $\o_1,\o_0$ are irrational numbers verifying a
\emph{Diophantine condition}, that is, for $j=0,1$, 
\begin{equation}
	|2\pi \o_j n|_T\;{\ge}\; c_j |n|^{-\rho_j}\, \qquad \forall n \in \mathbb{Z} \setminus\{0\} \, ,\label{dd}
\end{equation}
where $|\cdot|_T:=\inf_{m \in \mathbb{Z}}|\cdot+2m \pi|$ and
$\r_j\ge 1$, $c_j>0$. 
\item The side lengths of the boxes are chosen so that
\begin{equation}
L_{1,i}=q_{1,i} \, ,\qquad L_{0,i}=q_{0,i} \, ,
\end{equation}
and $\lim_{i\to\io}L_{1,i}/L_{0,i}=c$ with $0<c<\infty$.
\end{enumerate}
\begin{remark}
\\
\vspace{-0.4cm}
\begin{enumerate}
\item The energy correlations of the quasi-periodic Ising model are obtained as the limit of the energy correlations of a sequence of Ising models in boxes with interactions periodic in space with a period equal to the side of the boxes. 
In the limit $i\to\io$ the modulation
becomes 
$\sum_{n_0,n_1=-\io}^\io
\hat \phi^{(j)}_{\bn} 
e^{\ii  (n_0 (2\pi \omega_{0} x_0+\theta_{j,0})
+n_1(2\pi  \omega_{1} x_1+\theta_{j,1}))} 
$, that is 
quasi-periodic in both directions.
While in principle other ways to define a quasi-periodic Ising model can be imagined, this is the one chosen in numerical simulations in the physical literature, see \textit{e.g.} \cite{C}.
\item The quasi-periodic Ising model has been considered up to now
only with layered disorder, corresponding \textit{e.g.}\ to $\phi^{(0)}=0$; for instance $J^{(0)}_\xx=J$  and 
$J^{(1)}_\xx=J(1+\l \cos (2\pi \o_1 x_1+\th))$. In contrast, we consider a rather more general situation including
interactions of the form, for instance,
$J^{(0)}_\xx=(1+\l \cos (2\pi \o_0 x_0+\th)
\cos (2\pi \o_1 x_1+\phi))J^{(0)}$,  $J^{(1)}_\xx=
(1+\l (\cos (3\pi \o_0 x_0+\psi)\cos (6\pi \o_0 x_0+2\psi)
\cos (2\pi \o_1 x_1+\xi)))J^{(1)}$, with $\th, \ph,\psi,\x$ phases: that is 
the interaction is different in any bond, and quasi periodically modulated in both directions. 
\item The form of disorder we are considering breaks essentially all the symmetries present in the
non-disordered case other than spin-flip symmetry; in particular translation invariance
and inversion 
symmetry $x_j\to -x_j$ in both directions. Less general forms of disorder preserve some symmetry; 
in particular, in the case of layered disorder, translation invariance
and inversion in one space direction is preserved.
\end{enumerate}
\end{remark}
The truncated \emph{energy correlations} are defined for $\bx_1,\bx_2 \in \Lambda_i$ and $j_1,j_2 \in \{\pm\}$ as
\begin{equation}
S_i(\xx_1,j_1;\xx_2,j_2)=\langle 
\s_{\xx_1} \s_{\xx_1+{\bf e}_{j_1}}\s_{\xx_2} \s_{\xx_2+{\bf e}_{j_2}}\rangle_i-\langle \s_{\xx_1} \s_{\xx_1+{\bf e}_{j_1}}
\rangle_{i}\langle
\s_{\xx_2} \s_{\xx_2+{\bf e}_{j_2}}\rangle_{i} \, ,
\end{equation}
with
\begin{equation}
\la O\ra_{i}={1\over Z}  \sum_{ \{\sigma_\bx\} \in \{\pm \}^{\Lambda_i} } e^{- \beta H} O \, ,
\quad\quad Z=
\sum_{ \{\sigma_\bx\} \in \{\pm \}^{\Lambda_i} } e^{- \beta H} \, ,\end{equation}
where $Z$ is the partition function at inverse temperature $\beta>0$.

If $\l=0$, for $\b\neq\b_c$, with $\b_c$ given by
\begin{equation}
	\sinh (2\b_c J^{(0)}) \sinh (2\b_c J^{(1)})=1\label{is} \, ,
\end{equation}
the thermodynamic limit $i \to +\infty$ of the truncated energy correlations exists and is denoted by
$S(\xx_1,j_1;\xx_2,j_2)$. Such limit decays exponentially
for large distances with
correlation length $\x$ diverging at $\b_c$ as
$\x=O(|\b-\b_c|^{-1})$; $\b_c$ is therefore the
{\it critical temperature}. Moreover, 
in the limit $\b\to\b_c$ one has
\begin{equation}
S(\xx_1,j_1;\xx_2,j_2)=
Z_{j_1} Z_{j_2}g^0_+(\xx_1-\xx_2)g^0_-(\xx_2-\xx_1)+R_{j_1,j_2}(\xx_1,\xx_2)\label{bo}
\end{equation}
with $g^0_\pm(\xx-\xx_2)=  (v_1 (x_{1,1}-x_{2,1})\pm \ii (v_0 (x_{1,0}-x_{2,0}))^{-1}$,
$Z_{j}, v_1, v_0$ real constants,
${|R_{j_1,j_2}(\xx_1,\xx_2)|\over |\xx_1-\xx_2|^{2+\th}}\to 0$ for $|\bx_1-\bx_2|\to \io$ and 
$\th= \frac{1}{4}$. $\b_c$ is therefore the critical temperature, defined as the temperature
at which 
the correlation length diverges. 
Note that one is taking the $|\L_i|\to \io$
limit at  $\b\neq\b_c$, so that terms 
$O(e^{- L_i c |\b-\b_c|})$ vanishes in the limit, see Section \ref{sec:5} below, if $c$ is a constant and $L_i=\min\{L_{0,i},L_{1,i}\}$ is the shorter side of $\L_i$.
Note that $v_1,v_0$ are
the coefficients of the anisotropic rescaling of positions 
$g_+(x)=\bar g(v_1 x_{1},  v_0 x_0)$
with $\bar g(x_1,x_0)=
{1\over x_{1} + \ii  x_0}$ (and similar for $g_-$); they will be also called velocities.
Our main result describes the long-distance decay of correlations in the interacting case.
\vskip.3cm
\begin{theorem}\label{mainthm} Consider the Hamiltonian \pref{a1} and assume (1)-(5). There exist $\l_0,C,\kappa>0$, functions $b:(-\l_0,\l_0) \to \mathbb{R}$, $\xi_j:(-\l_0,\l_0) \times \mathbb{T}^2 \to \mathbb{R}$ and $\alpha_j:(-\l_0,\l_0) \to \mathbb{C}$ for $j=0,1$, with  $\sup_{\lambda} |b(\l)|,\,\sup_\lambda |\alpha_j(\l)|,\sup_{\lambda,\boldsymbol{\vartheta}} |\xi_j(\l,\boldsymbol{\vartheta})| < C$ such that the following holds. 
For any $|\lambda| < \lambda_0$ there exists $\beta_c(\lambda)=\beta_c+b(\lambda)$ such that
\begin{enumerate}
\item for $\b\neq\b_c(\l)$ the limit
$\lim_{i \to \infty}
S_{i}(\xx_1,j_1;\xx_2,j_2)=S(\xx_1,j_1;\xx_2,j_2)$
exists and is finite.
\item 
For $\b\neq \b_c(\l)$
\begin{equation}
|S(\xx_1,j_1;\xx_2,j_2)|\le C e^{-\k (|\b-\b_c(\l)||\xx_1-\xx_2|)^{1\over 2}} \, .
\end{equation}
\item
For $\b\to \b_c(\l)$ 
\begin{equation} \label{eq:SdelTeorema}
\lim_{\beta\to \beta_c(\lambda)}S(\xx_1,j_1;\xx_2,j_2)=
Z_{j_1,\bx_1}(\l) Z_{j_2,\bx_2}(\l) g_+(\xx_1-\xx_2)g_-(\xx_2-\xx_1)
+R_{j_1,j_2}(\xx_1,\xx_2)
\end{equation}
with 
\begin{equation}
g_+(\xx)=  {1\over v_1(\l) x_{1} + \ii  v_0(\l) x_0} \, ,\quad g_-(\xx)=  {1\over (v_1(\l))^* x_{1} -\ii  (v_0(\l))^* x_0} \, ,
\end{equation}
 and 
$|R(\xx_1,j_1;\xx_2,j_2)|/|\xx_1-\xx_2|^{2+\th}\to 0$ for $|\bx_1-\bx_2|\to \io$,
$\th=
 1/4$
and 
\begin{equation}
 Z_{j,\bx}(\l)=Z_{j}+\l \x_j(\l,2\pi \o_0 x_0,2\pi \o_1 x_1)\quad 
v_j(\l)=v_j+\l \a_j(\l)
\end{equation}
with $Z_{j}, v_j$ defined in \pref{bo}.
 \end{enumerate}  
\end{theorem} 
\vskip.3cm
\begin{remark}
\\
\vspace{-0.4cm}
\begin{enumerate}
\item
The asymptotic behavior of the 2-point correlation \eqref{eq:SdelTeorema} at criticality
is similar to the one of the unperturbed case, with the main difference that the
amplitude is the product of two quasi-periodic functions 
$Z_{j_1,\bx_1}(\lambda)$ and $Z_{j_2,\bx_2}(\lambda)$. 
The velocities and the critical temperature are also modified. In contrast, the
exponents are universal and no logarithmic corrections are present; this provides a rigorous confirmation 
of the Harris-Luck criterion.
Outside the critical temperature a stretched exponential decay is found, but this is
just for technical reasons and exponential decay is expected. 
The analysis could  be easily extended to the $n$-point energy correlations.
\item The proof is based on the convergence of the series for the correlations,
showing a small-divisor problem similar to the one appearing in perturbation of
integrable Hamiltonian systems, see \textit{e.g.}\ \cite{Ga}. Convergence is shown
assuming
only a Diophantine condition on the frequencies, the smallness of the coupling and a fast decay property of the harmonics;
without such assumptions a different behavior is expected.
\item 
The result holds for any angle $\th_j$, including cases where inversion or translation invariance is broken
in both directions. 
This is a peculiar fact since in many similar models with small-divisor problems, extra conditions are usually required.
\end{enumerate}
\end{remark}

\subsection{Sketch of the proof}

The starting point of the analysis is the exact representation of the quasi-periodic Ising model as a Grassmann integral, which is an immediate consequence of the dimer representation, see \textit{e.g.}\ \cite{MW}, and the fact that Pfaffians can be expressed as Gaussian Grassmann integrals, see \textit{e.g.}\ \cite{MastropietroNonPerturbative}. The energy correlations can be written as the sum of terms of the form (the exact expressions are in Section \ref{sec:Grassmann})

\begin{equation}
\frac{\int P_\psi(d\psi) P_\xi(d\xi) e^{V} O}{\int P_\psi(d\psi) P_\xi(d\xi) e^{V}} \label{int}
\end{equation}
where $P_\psi(d\psi), P_\x(d\xi)$ are Grassmann Gaussian integrations,
$O$ is a quartic monomial in the Grassmann variables, 
and $V$ is a sum of monomials in $\psi, \xi$ and vanishes for $\l=0$.
The propagator (or covariance) of $P_\xi(d\xi)$ is $\hat{g}_\xi(\kk)$, given by
\begin{equation}\label{fo}
\hat{g}_\xi(\kk) := 
    \begin{pmatrix}
        -\ii t^{(1)} \sin k_1 + t^{(0)} \sin k_0 & \ii m_{\xi}(\bk) \\
        -\ii m_{\xi}(\bk) & -\ii t^{(1)} \sin k_1 - t^{(0)} \sin k_0
    \end{pmatrix}^{-1} \, ,
\end{equation}
with {\color{black}$t^{(j)}=\frac{1}{|\Lambda|} \sum_{\bx \in \Lambda} \tanh (\b J^{(j)}_\bx)$} and $m_\xi = m_\chi = O(1)$. From the explicit expression given below in \eqref{eq:MassaChi2}, $m_{\chi}(\kk) = m_{\chi}(0) + F(\kk)$ with $m_{\chi}(0) = O(1)$ and $F(\kk) = 0$ at $\kk = 0$, and bounded away from zero uniformly in $\beta$ in the other three poles of the diagonal elements of $\hat{g}_\xi(\kk)$. \textcolor{black}{One recognizes in \eqref{fo} the propagator of a lattice Dirac fermion with a mass $m_{\chi}(0)$ and Wilson term $F(\kk)$.}

The propagator $\hat{g}_\psi(\bk)$ of $P_\psi(d\psi)$ has a similar expression with a mass that can vanish as a function of temperature. The variables $\xi$, being associated with a bounded propagator (called non-critical variables for this reason), can be integrated out (see Section \ref{sec:Xi}), expressing the energy correlations as Grassmann integrals of the form
\begin{equation}
\frac{\int P_\psi(d\psi) e^{\tilde{V}} \tilde{O}}{\int P_\psi(d\psi) e^{\tilde{V}}} \label{ss}
\end{equation}
with $\tilde{V} = \frac{1}{| \Lambda_i |}\sum_\bn \sum_\kk \psi_{-\bk} \widehat{W}_\bn(\bk) \psi_{\bk - 2\pi \Omega \bn}$, where $\widehat{W}_\bn(\bk)$ is a matrix with elements exponentially decaying in $\bn$ and analytic in $\lambda$. Here, $\psi = (\psi_+, \psi_-)$, $\Omega = \begin{pmatrix} \omega_0 & 0 \\ 0 & \omega_1 \end{pmatrix}$, and $\tilde{O}$ is still quartic in $\psi$. This representation \textcolor{black}{is an immediate consequence of Wick's theorem}, allowing us to represent $W_\bn(\bk)$ as a sum of chain graphs, that is, products of propagators of the form $\hat{g}_\xi(\bk) \hat{g}_\xi(\bk - 2\pi \Omega \bn_1) \hat{g}_\xi(\bk - 2\pi \Omega \bn_2) \cdots$. Convergence follows from the exponential decay of $\hat{\phi}_\bn$ and the boundedness of $\hat{g}_\xi$.

One could perform the integration in $\psi$ (critical variables) in a similar way, obtaining an expansion for the correlations still expressed in terms of graphs. In this case, however, the propagator of the $\psi$-variables is unbounded, and at criticality there are graphs that are naively bounded by $O(n!^\alpha)$ if $n$ is the order and $\alpha$ a constant, due to the presence of small divisors. To achieve convergence, one needs to improve the bounds, showing that such factorials are indeed not present.

To show this, a multiscale analysis is required, as described in Section \ref{sec:Critical}. One decomposes the propagator as a sum of propagators supported at different momentum shells with scale $h$, that is $|\bk| \sim \gamma^{h}$, $\gamma>1$, with $h=1,0,-1,-2,\dots$. In other words, {\color{black}$\hat{g}_\psi(\bk)=\sum_{h=-\infty}^1 \hat{g}^{(h)}(\bk)$} with $\hat{g}^{(h)}(\bk)=O(\gamma^{-h})$. Integrating the higher momentum scales, we obtain
\begin{equation}
\frac{\int P_\psi^{(\leq h)}(d\psi^{(\le h)} ) e^{V^{(h)} (\psi^{(\le h)})} \tilde{O}^{(h)}}{\int P_\psi^{(\leq h)}(d\psi^{(\le h)} ) e^{V^{(h)} (\psi^{(\le h)})}}
\end{equation}
with $\textcolor{black}{P^{(\leq h)}_\psi(d\psi^{(\le h)} )}$ a Gaussian Grassmann integration corresponding to scales $\leq h$ and again
\[
V^{(h)} = \frac{1}{|\Lambda_i|} \sum_\bn \sum_\kk \psi_{-\bk}^{(\leq h)} \widehat{W}^{(h)}_\bn(\bk) \psi^{(\le h)}_{\bk-2\pi\Omega \bn}
\]
with $\widehat{W}^{(h)}_\bn$ depending on the scale $h$. Using that, for Gaussian Grassmann integrals, $P_\psi^{(\leq h)}(d\psi^{(\le h)} ) =P_\psi^{(\leq h-1)}(d\psi^{(\le h-1)} )P_\psi^{(h)}(d\psi^{(h)} )$, we can integrate the $\psi^{(h)}$ variable iteratively; this again produces chain graphs as a product of propagators of arbitrarily large size $O(\gamma^{-h})$ times products of the $\widehat{W}^{(h)}_\bn$. In Renormalization Group terminology, the terms in $V^{(h)}$ are relevant perturbations that could alter the critical behavior.

To show that this is not the case, one needs to distinguish between the case $\bn=0$, which are called resonant terms or resonances, and the non-resonant case $\bn\neq0$. In the first case, one gets an accumulation of identical small divisors in the perturbative expansion, ending with a non-summable behavior. Such a phenomenon is avoided by modifying the expansion, introducing a counterterm to account for the modification of the critical temperature, and by modifying the velocities at each iteration step, \textcolor{black}{see Section \ref{subsec:PolentaECarciofini}}. That is, the propagator of the $\psi^{(\le h)}$ close to $\kk=0$ acquires the form
$
\sim \chi_h(\kk) \begin{pmatrix}
    -\ii v_{1,h} k_1 + v_{0,h} k_0 & -\ii m \\
    \ii m & -\ii v_{1,h}^* k_1 - v_{0,h}^* k_0
\end{pmatrix}^{-1}
$
where $\chi_h(\kk) \neq 0$ for $|\kk|\le \gamma^{h}$. Note that reabsorbing certain terms in the propagator is possible only if the $\widehat{W}^{(h)}_0$ have a suitable form that does not change the qualitative structure of the propagator; this is indeed what happens. When the angles $\theta_{j}$ are generic, the breaking of symmetries does not allow us to conclude the reality of velocities (which turn out to be real in the layered case).

One has then to deal with the terms in $V^{(h)}$ with $\bn\neq0$; in that case, the repeated small divisors are not identical and they cannot be reabsorbed into the propagator. If the disorder was periodic, that is, $\Omega$ is rational so that $2\pi \Omega \bn$ mod $2\pi$ is bounded, this would mean that there is a scale $\bar{h}$ so that such terms are not present for $h\leq \bar{h}$; hence, they could be easily bounded. In contrast, if $\Omega$ is irrational, that is in the quasi-periodic case, such terms appear at any scale $h$, and the propagators associated with fields multiplying $\widehat{W}^{(h)}_\bn$ are as large as $O(\gamma^{-h})$. One needs therefore, to achieve convergence, to prove that $\widehat{W}^{(h)}_\bn(\bk)$ has a fast decay in $h$ compensating for the small divisor $\gamma^{-h}$. This follows from the Diophantine condition, as it implies that if $\bk$ and $\bk-2 \pi \Omega \bn$ are $O(\gamma^{h})$, then $\bn$ is large, that is $|\bn| \geq \gamma^{- \frac{h}{\tau} }$ for a suitable constant $\tau$. The decay in $\bn$ of $\widehat{W}^{(h)}_\bn(\bk)$ can therefore be converted into a decay in $\gamma^{-h}$ compensating for the $\gamma^{-h}$ of the propagator. 

However, the gain must be obtained at every iteration step and one has to check that no non-summable combinatorial factors are present; this is done \textcolor{black}{using the cluster structure of graphs (see Section \ref{sec:Critical} and in particular Lemma \ref{prop:EstimateResonantClusters} where the convergence of the series expansion is proved)}. The series obtained is in $\lambda$ and in the running coupling constants (corresponding to the renormalizations of the temperature and of the velocities); one has to show that it is possible to fine-tune a parameter, corresponding to the shift of the critical temperature, to prove that they remain small at any iteration, as proved in Section \ref{sec:Critical}. 

Finally, in Section \ref{sec:5}, the full expansion for the energy correlations is considered. In this case, after the integration of the fields of scales $1,0,-1,-2,\dots,h$, one gets source terms of the form $
\frac{1}{|\Lambda_i|^2} \sum_{\bn,j, \bk,{\bf p}} Z_{h,\bn}^{(j)} \psi_{-\bk}^{(\leq h)} \sigma_2 \psi^{(\leq h)}_{\bk+{\bf p}-2\pi\Omega \bn}\widehat{A}_\pp^{(j)}
$
where $Z_{h,\bn}^{(j)}$ are running coupling constants associated with the source terms in the generating function for correlations and $\widehat{A}_\pp^{(j)}$ is the Fourier coefficient of an external field (see \eqref{en} below). In this case, there are running coupling constants corresponding to $\bn\neq0$ as there is no gain due to the Diophantine condition. They have a finite limit as $h\to-\infty$, and this implies that the critical exponents are the same as in the unperturbed case, and they produce the quasi-periodic amplitude of the energy correlations.

\subsection{Comparison with previous results}
The paper uses a fermionic Renormalization Group approach to the Grassmann representation of the Ising model, previously used in the case of non nearest neighbor perturbations, see \cite{M1,M2,M3}, or for coupled Ising and related models like Six-vertex, Ashkin-Teller or dimer models \cite{M1, BFM1,BFM2,GMT}. In such cases, the starting point is a Grassmann integral similar to \eqref{int} but with $V$ a quartic or higher order translation invariant interaction. 

In the case of the quasi-periodic Ising model, the situation is different: the interaction in the Grassmann integral is quadratic but the modulation of the potential breaks translation invariance and it requires the use of KAM methods to solve the small-divisor problem.

The relation with KAM appears from \eqref{int}; as the exponent of the integrand is quadratic in the Grassmann variables, the energy correlations could, in principle, be deduced by a suitable lattice Dirac equation in a quasi-periodic potential, essentially given by
\[
\sigma_2 (\psi_{\mathbf{x}+\mathbf{e}_0}-\psi_\mathbf{x}+\lambda \phi^{(0)}_\mathbf{x}\psi_\mathbf{x})+\sigma_1(\psi_{\mathbf{x}+\mathbf{e}_1}-\psi_\mathbf{x}+\lambda \phi^{(1)}_\mathbf{x}\psi_\mathbf{x})+\ii m  \sigma_3 \psi_\mathbf{x}=E\sigma_1 \psi_\mathbf{x} \, ,
\]
with $\sigma_1, \sigma_2, \sigma_3$ being the Pauli matrices. Indeed, such an equation has not been studied, but an extensive literature has been instead devoted to the related problem of the lattice Schr\"odinger equation with a quasi-periodic potential (which is strictly related to a KAM problem), like
\[
\psi_{x+1}+\psi_{x-1}+\lambda \phi_x \psi_x=E \psi_x\label{allp}
\]
where $x \in \mathbb{Z}$ and $\phi_x=\bar{\phi}(2\pi \omega x+\theta)$ with $\bar{\phi}$ $2\pi$-periodic. For small $\lambda$, the eigenvalues and eigenfunctions of the above equation were studied in \cite{P1} where two Diophantine conditions are assumed, one over the frequency and the other over the energy, using KAM methods. In particular, it was required that $|2\pi \omega n|_T \geq C |n|^{-\tau}$ and $|2\pi \omega n \pm 2 \rho|_T \geq C |n|^{-\tau}$, with $E=\cos \rho$ (first and second Melnikov condition). In \cite{P2} instead, the case $\rho= n \pi \omega$ was studied, corresponding to the gaps in the spectrum. Several attempts were made to improve such conditions, culminating in \cite{P3}, where the second Melnikov condition was removed, and in \cite{P4} where $\omega$ was assumed to be any irrational. In the bidimensional case, more complicated Diophantine assumptions are required \cite{B1} and less detailed knowledge is available.

An important related issue is the computation of the correlations of a system of several particles (fermions in particular) in a quasi-periodic potential, with a single-body interaction described by \eqref{allp}. In the absence of a many-body interaction, the knowledge of the single particle properties of \eqref{allp} could be sufficient to determine the properties of the ground state correlations. If $\phi_x$ in \eqref{allp} is random, this was indeed done in \cite{AG}, and with a periodic potential (in the continuum) it was done in \cite{Bom}, where indeed the asymptotic properties of correlations were determined only by a very precise knowledge of the singularities of the eigenvalues (branch points) in the complex plane.

In the quasi-periodic case, a derivation of the asymptotic behavior of fermionic correlations directly from the Schrödinger equation \eqref{allp} has never been attempted. However, such asymptotic decay has been derived by writing the fermionic correlations as Grassmann integrals similar to \eqref{int}, with interacting measure $P(d\psi) e^V$, propagator $(\ii k_0+\cos (k_1+n \omega)-E)^{-1}$ and $V$ sum of monomials $\psi^+_{k_0,k}\psi^-_{k_0,k+2\pi n\omega}$. The long-distance behavior of the non-interacting ground state correlations in $d=1$ has been determined using a multiscale analysis in \cite{BGM} via fermionic Renormalization Group methods, inspired by the ones used in KAM Lindstedt series \cite{G,GM}. The result was valid for $E=\cos \bar{m}\pi \omega$, $\bar{m}\in \mathbb{N}$, that is assuming a gap condition like the one in \cite{P2}; the ground state correlations decay exponentially both in space and Euclidean time. Note that there are infinitely many gaps with size $O(\lambda \hat{\phi}_{\bar{m}})$, the spectrum being a Cantor set.

Later on, the RG methods were extended to include the presence of a weak many-body interaction (and weak quasi-periodic potential): it was shown in \cite{M11} that the gaps are not closed by the interaction (if the corresponding harmonic is present in the potential), but are strongly modified via the presence of a critical interaction-dependent exponent; the gaps become $O((\lambda \hat{\phi}_{\bar{m}})^{1+\eta})$, $\eta=a U+O(U^2)$, where $U$ is the coupling of the many-body interaction and $\eta$ is a critical exponent. A similar phenomenon was also shown to happen in the interacting Aubry-André model where only one harmonic is present in the initial potential \cite{M11cc} and in the interacting Hofstadter model \cite{M11a} for the Hall effect. In higher dimensions, a class of fermionic systems in $d=2,3$ known as Weyl semimetals have been considered \cite{M11f} in presence of a quasi-periodic disorder and interaction in the weak coupling regime; by assuming a first and second Melnikov condition restricting densities, it was shown the stability of the Weyl phase, that is the absence of localization.

While the above-mentioned results regard the case of fermions on a lattice with a weak quasi-periodic potential and a many-body interaction, the case of strong potential has a different behavior, manifesting the phenomenon of Anderson localization. In this case, one considers the kinetic energy as a perturbation of the quasi-periodic potential, and not the opposite as in the previous case. In \cite{GM1}, localization without many-body interaction was shown, and later the proof of $T=0$ many-body localization of interacting fermions \cite{M11b,M11c,M11d,M11e} was established. It should be remarked that at the moment, such RG methods are the only ones allowing us to take into account rigorously the interaction in the thermodynamic limit.

At the mathematical level, the Renormalization Group methods used to analyze the above fermion systems in the weakly disordered regime are related to the ones used here for the quasi-periodic Ising model, but there are important differences. First of all, in fermionic systems one has to restrict the values of the chemical potential either to ensure the validity of a gap condition, as in \cite{M11,M11cc,M11a}, or a second Melnikov condition \cite{M11f}. There is no analogue of chemical potential in the Ising model, but we can solve the small-divisor problem without imposing any condition. In addition, in fermionic models considered so far, the 2-point fermionic correlation was studied, while here the energy correlations are considered, quartic in the fermions, a fact producing new (infinitely many) marginal operators and the quasi-periodic modulation of the amplitude. Moreover, the quasi-periodic disorder is bidimensional in space and Euclidean time and all possible choices of angles are considered, while previously the only layered or bidimensional cases with angles chosen equal to zero were treated \cite{M11e}. The general form of the disorder considered here breaks the inversion symmetries, an important property to prove the reality of the velocities.

In addition to such technical improvements, it should be also remarked that the application of direct methods, previously developed for apparently unrelated problems like KAM series or non-relativistic fermions, to the quasi-periodic Ising model is a major novelty of this paper and it produces the first rigorous proof of the Harris-Luck criterion, and a natural starting point for the inclusion of next to nearest neighbor interactions.

\section{Grassmann representation}
\label{sec:Grassmann}
\noindent
From the dimer representation of the Ising model, see \textit{e.g.}\ \cite{MW2}, one can express the energy correlations, which are expressed in terms of four Pfaffians, using Grassmann integrals; see \textit{e.g.}\ \cite{MastropietroNonPerturbative}. The energy correlations can therefore be written as
\begin{equation}
S(\mathbf{x}_1,j_1;\mathbf{x}_2,j_2) = \left. \frac{\partial^2}{\partial A_{\mathbf{x}_1,j_1} \partial A_{\mathbf{x}_2,j_2}} \log Z(A) \right|_{A=0}, \label{en}
\end{equation}
with
\begin{equation}
Z(A) = \frac{1}{2} \sum_{\boldsymbol{\alpha} \in \{\pm\}^2} \tau_{\boldsymbol{\alpha}} Z_{\boldsymbol{\alpha}}(A),
\end{equation}
where $\tau_{+,-} = \tau_{-,+} = \tau_{-,-} = -\tau_{+,+} = 1$ and
\begin{equation}
Z_{\boldsymbol{\alpha}}(A) = \left[ \prod_{\mathbf{x} \in \Lambda_i} \prod_{j=0}^1 \cosh(\beta J^{(j)}_{\mathbf{x}} + A_{\mathbf{x},j}) \right] \int \mathcal{D}^{\Lambda_i} \Phi\, e^{S_{\Lambda_i}(\Phi,A)},
\end{equation}
with
\begin{equation} \label{eq:StartingAction}
\begin{split}
S_{\Lambda_i}(\Phi, A) &:= \sum_{\mathbf{x} \in \Lambda_i} \left[ \tanh(\beta J^{(1)}_{\mathbf{x}} + A_{\mathbf{x},1}) \overline{H}_{\mathbf{x}} H_{\mathbf{x} + \mathbf{e}_1} + \tanh(\beta J^{(0)}_{\mathbf{x}} + A_{\mathbf{x},0}) \overline{V}_{\mathbf{x}} V_{\mathbf{x} + \mathbf{e}_0} \right] \\
&\quad + \sum_{\mathbf{x} \in \Lambda_i} \left[ \overline{H}_{\mathbf{x}} H_{\mathbf{x}} + \overline{V}_{\mathbf{x}} V_{\mathbf{x}} + \overline{V}_{\mathbf{x}} \overline{H}_{\mathbf{x}} + V_{\mathbf{x}} \overline{H}_{\mathbf{x}} + H_{\mathbf{x}} \overline{V}_{\mathbf{x}} + V_{\mathbf{x}} H_{\mathbf{x}} \right].
\end{split}
\end{equation}
Here, $\overline{H}_{\mathbf{x}}, H_{\mathbf{x}}, \overline{V}_{\mathbf{x}}, V_{\mathbf{x}}$ are independent Grassmann variables, four for each lattice site, and $E_{\mathbf{x},1} := \overline{H}_{\mathbf{x}} H_{\mathbf{x} + \mathbf{e}_1}$, while $E_{\mathbf{x},0} := \overline{V}_{\mathbf{x}} V_{\mathbf{x} + \mathbf{e}_0}$. Moreover, $\Phi := \{\overline{H}_{\mathbf{x}}, H_{\mathbf{x}}, \overline{V}_{\mathbf{x}}, V_{\mathbf{x}}\}_{\mathbf{x} \in \Lambda_i}$ denotes the collection of all these Grassmann variables, and $\mathcal{D}^{\Lambda_i} \Phi$ is a shorthand for $\prod_{\mathbf{x} \in \Lambda_i} d\overline{H}_{\mathbf{x}} dH_{\mathbf{x}} d\overline{V}_{\mathbf{x}} dV_{\mathbf{x}}$. The Grassmann integration is defined so that, for all $\mathbf{x} \in \Lambda_i$,
\begin{equation}
\int d\overline{H}_{\mathbf{x}} dH_{\mathbf{x}} d\overline{V}_{\mathbf{x}} dV_{\mathbf{x}} = 0 \, , \quad \quad 
\int d\overline{H}_{\mathbf{x}} dH_{\mathbf{x}} d\overline{V}_{\mathbf{x}} dV_{\mathbf{x}} (V_{\mathbf{x}} \overline{V}_{\mathbf{x}} H_{\mathbf{x}} \overline{H}_{\mathbf{x}}) = 1 \, .
\end{equation}
The label $\boldsymbol{\alpha} = (\alpha_1, \alpha_2)$, with $\alpha_1, \alpha_2 \in \{\pm\}$, refers to the boundary conditions, which are periodic or antiperiodic in the horizontal (resp. vertical) direction. Letting $Z = \sum_{\boldsymbol{\alpha} \in \{\pm\}^2} \tau_{\boldsymbol{\alpha}} Z_{\boldsymbol{\alpha}}$ with $Z_{\boldsymbol{\alpha}} = Z_{\boldsymbol{\alpha}}(0)$, the truncated energy correlation \eqref{en} can be written as
\begin{equation}
S(\mathbf{x}_1,j_1;\mathbf{x}_2,j_2) = \sum_{\boldsymbol{\alpha} \in \{\pm\}^2} \frac{\tau_{\boldsymbol{\alpha}} Z_{\boldsymbol{\alpha}}}{2Z} \left\langle E_{\mathbf{x}_1,j_1}; E_{\mathbf{x}_2,j_2} \right\rangle^T_{\boldsymbol{\alpha},i}, \label{LeEquazioniTutteUguali}
\end{equation}
where $\left\langle \cdot \right\rangle_{\boldsymbol{\alpha},i}$ is the average with respect to the Grassmann ``measure'' $\mathcal{D}^{\Lambda_i} \Phi e^{S_{\Lambda_i}(\Phi, \mathbf{0})} / Z_{\boldsymbol{\alpha}}$ with $\boldsymbol{\alpha}$ boundary conditions.

Let us consider first the case $A=0$.

We perform the (well-known) change of variables 
\begin{equation}\label{eq:CambioDiVarHPsi}
\begin{array}{lcl}
\lis H_\xx+\ii H_\xx=e^{\ii \pi/4}\psi_{+,\xx}-e^{\ii \pi/4}\chi_{+,\xx} \, ,&\quad\quad& \lis H_\xx-\ii H_\xx=e^{-\ii \pi/4}\psi_{-,\xx}-e^{-\ii \pi/4}\chi_{-,\xx} \, ,\\
\lis V_\xx+\ii V_\xx=\psi_{+,\xx}+\chi_{+,\xx} \, ,& &\lis V_\xx-\ii V_\xx=\psi_{-,\xx}-\chi_{-,\xx} \, .\\
\end{array}	
\end{equation}
We set $\X_{\boldsymbol\a}=\int \mathcal{D}^{\Lambda_i} \Phi\, e^{S_{\Lambda_i}(\Phi,0)}$, and, for $j=0,1$,
$ t^{(j)} =\frac{1}{|\Lambda_i|}\sum_{\bx \in \Lambda_i} \tanh( \beta J^{(j)}_\bx)$, we define $V_\bx^{(j)}$ as
\begin{equation}\label{eq:defVx}
	\begin{split}
	t^{(j)}_{\bx}&=\tanh( \beta J^{(j)}_{\xx})=\tanh\left( \beta J^{(j)}\Big(1+\lambda \phi^{(j)}(2\pi \omega_{0,i } x_0+\theta_{j,0},2\pi \omega_{1,i } x_1+\theta_{j,1})\Big)\right)\\&\equiv t^{(j)} +V^{(j)}_\bx 
	\end{split}
\end{equation}
so that $\sum_{\bx \in \Lambda_i}V^{(j)}_\bx=0$. We can write
\begin{equation}
\X_{\boldsymbol\a}=\int \prod_{\xx\in\L_i} d\psi_{+,\xx} d\psi_{-,\xx} d\chi_{+,\xx} d\bar\chi_{-,\xx} e^{S^{(\chi)}(\chi)+S^{(\psi)}(\psi)+Q(\psi,\chi)}
\end{equation}
where{\color{black}, denoting with $\cdot$ the Euclidean scalar product,}
\begin{equation}\label{eq:schi}
	\begin{split}
	S^{(\chi)}(\chi) \;:=&\;-\frac{1}{4} \sum_{\bx \in \Lambda_i} t^{(1)}_{\bx} \begin{pmatrix}
		\chi_{+,\bx}\\
		\chi_{-,\bx}
	\end{pmatrix}
	\cdot\begin{pmatrix}
		-1 & +\ii \\
		-\ii & -1 
	\end{pmatrix}
	\begin{pmatrix}
		\chi_{+,\bx+\mathbf{e}_1}\\
		 \chi_{-,\bx+\mathbf{e}_1}		
	\end{pmatrix} + \\
	&-\frac{1}{4 } \sum_{\bx \in \Lambda_i} t^{(0)}_{\bx} \begin{pmatrix}
		\chi_{+,\bx}\\
		\chi_{-,\bx}
	\end{pmatrix}
	\cdot\begin{pmatrix}
		-\ii & +\ii \\
		-\ii & +\ii 
	\end{pmatrix}
	\begin{pmatrix}
		\chi_{+,\bx+\mathbf{e}_0}\\
		 \chi_{-,\bx+\mathbf{e}_0}		
	\end{pmatrix} + \\
	&-\frac{1}{4 } \sum_{\bx \in \Lambda_i}  
	2\ii(\sqrt{2}+1) \big(\chi_{+,\bx} \chi_{-,\bx}-\chi_{-,\bx} \chi_{+,\bx} \big)\, .
	\end{split}
\end{equation}
\begin{equation}
	\begin{split}
	S^{(\psi)}(\psi) \;:=&\;-\frac{1}{4 } \sum_{\bx \in \Lambda_i} t^{(1)}_{\bx} \begin{pmatrix}
		\psi_{+,\bx}\\
		\psi_{-,\bx}
	\end{pmatrix}
	\cdot\begin{pmatrix}
		-1 & +\ii \\
		-\ii & -1 
	\end{pmatrix}
	\begin{pmatrix}
		\psi_{+,\bx+\mathbf{e}_1}\\
		 \psi_{-,\bx+\mathbf{e}_1}		
	\end{pmatrix} + \\
	&-\frac{1}{4 } \sum_{\bx \in \Lambda_i} t^{(0)}_{\bx} \begin{pmatrix}
		\psi_{+,\bx}\\
		\psi_{-,\bx}
	\end{pmatrix}
	\cdot\begin{pmatrix}
		-\ii & +\ii \\
		-\ii & +\ii 
	\end{pmatrix}
	\begin{pmatrix}
		\psi_{+,\bx+\mathbf{e}_0}\\
		\psi_{-,\bx+\mathbf{e}_0}		
	\end{pmatrix} +\\
	&-\frac{1}{4} \sum_{\bx \in \Lambda_i}  \big[-2\ii(\sqrt{2}-1) \big] \big(\psi_{+,\bx} \psi_{-,\bx}-\psi_{-,\bx} \psi_{+,\bx} \big)\, .
	\end{split}
\end{equation}
\begin{equation}
	\begin{split}
	Q(\psi,\chi) \;:=\;& \, {1\over 4} \sum_{\bx \in \Lambda_i} t_{\bx}^{(1)} \begin{pmatrix}
		\psi_{+,\bx} \\
		\psi_{-,\bx}
	\end{pmatrix}
	\cdot
	\begin{pmatrix}
		-1 & \ii \\
		-\ii & -1
	\end{pmatrix}
	\begin{pmatrix}
		\chi_{+,\bx+\mathbf{e}_1} \\
		\chi_{-,\bx+\mathbf{e}_1}
	\end{pmatrix} + \\
	& \frac{1}{4} \sum_{\bx \in \Lambda_i} t_{\bx}^{(0)} \begin{pmatrix}
		\psi_{+,\bx} \\
		\psi_{-,\bx}
	\end{pmatrix}
	\cdot
	\begin{pmatrix}
		\ii & -\ii \\
		\ii & -	\ii
	\end{pmatrix}
	\begin{pmatrix}
		\chi_{+,\bx+\mathbf{e}_0} \\
		\chi_{-,\bx+\mathbf{e}_0}
	\end{pmatrix}+(\psi \leftrightarrow \chi) \, .
	\end{split}
\end{equation}
Note that, by \eqref{eq:defVx}, $V^{(j)}_{\xx}$ is a $2\pi$-periodic function in $2\pi\Omega\bx + \boldsymbol{\vartheta}_j$ with $\Omega = \begin{pmatrix} \omega_{0,i} & 0 \\ 0 & \omega_{1,i} \end{pmatrix}$, and with zero mean so that we can write
\begin{equation}
V^{(j)}_{\xx} = \sum_{\bn} \widehat{V}^{(j)}_\bn e^{\ii \bn \cdot \boldsymbol{\vartheta}_j} e^{ \ii 2 \pi \Omega \bn \cdot \bx}, \quad \text{with} \quad \widehat{V}_{\bn}^{(j)} := \frac{1}{|\Lambda_i|} \sum_{\bx \in \Lambda_i} V_\bx^{(j)} e^{- \ii \bn \cdot (2 \pi \Omega \bx + \boldsymbol{\vartheta}_j)},
\end{equation}
where $V^{(j)}_{\xx}$ is defined in \eqref{eq:defVx}, $\bn$ takes values as in \eqref{eq:UnoPuntoTre},
\begin{equation}
|\widehat{V}^{(j)}_\bn| \leq C |\lambda| e^{-\h |\bn|} \label{app}
\end{equation}
with $C$ and $\eta$ independent of $i$, and $\widehat{V}^{(j)}_\bn = (\widehat{V}^{(j)}_{-\bn})^*$, $\widehat{V}^{(j)}_0 = 0$; these properties follow from \eqref{eq:UnoPuntoQuattro}, \eqref{eq:defVx} and by analyticity of $V_\bx^{(j)}$ as a function of the coordinates.

Denoting by $\zeta_\pm=\psi_\pm,\chi_\pm$,
\begin{equation}
\textcolor{black}{\zeta_{\pm,\bx}} \;:=\; \frac{1}{|\Lambda_i|} \sum_{\bk \in \mathcal{D}_\balpha} \textcolor{black}{\widehat{\zeta}_{\pm,\bk}} e^{\ii \bk\cdot\bx} \, ,
\end{equation}
with
\begin{equation}\label{eq:DuePuntoQuindici}
	\begin{split}
	\mathcal{D}_\balpha \;=\; \left\{ \bk=(k_0,k_1) \in \mathbb{R}^2 \, \Bigg| \begin{array}{c}\, k_j=\frac{\pi}{L_j}(2 \kappa_j+1-\alpha_j1) \\ \kappa_j\in\big\{-\big\lfloor \frac{L_j+1}{2} \big\rfloor,\dots,0,1,\dots,\big\lfloor \frac{L_j}{2} \big\rfloor\big\} \end{array}\right\} \, .
	\end{split}
\end{equation}
Note that 
\begin{equation}
	\begin{split}
		\sum_{\bx \in \Lambda_i} &V_\bx^{(1)} \begin{pmatrix}
			\widehat{\chi}_{+,\bx} \\ \widehat{\chi}_{-,\bx} 
		\end{pmatrix} \cdot \begin{pmatrix}
			-1 & \ii \\
			-\ii & -1
		\end{pmatrix} \begin{pmatrix}
			\widehat{\chi}_{+,\bx+\mathbf{e}_1} \\ \widehat{\chi}_{-,\bx+\mathbf{e}_1}
		\end{pmatrix} \\
	&=\frac{1}{2 |\Lambda_i|}\sum_{\substack{\bk \in \mathcal{D}_\balpha \\ \bn \in \mathbb{Z}^2}}  \widehat{V}_\bn^{{(1)}}e^{\ii \bn \cdot \btheta_1} \begin{pmatrix}
			\widehat{\chi}_{+,-\bk} \\
			\widehat{\chi}_{-,-\bk}
		\end{pmatrix}
		\cdot
		e^{\ii (k_1-2 \pi \omega_1 n_1)}
		\begin{pmatrix}
			-1 & \ii \\
			-\ii & -1
		\end{pmatrix}
		\begin{pmatrix}
			 \widehat{\chi}_{+,\bk- 2 \pi \Omega \bn} \\
			 \widehat{\chi}_{-,\bk- 2 \pi \Omega \bn}
		\end{pmatrix} \\
		&\quad +\frac{1}{2 |\Lambda_i|} \sum_{\substack{\bk \in \mathcal{D}_\balpha \\\bn \in \mathbb{Z}^2}}  \widehat{V}_\bn^{{(1)}}e^{\ii \bn \cdot \btheta_1} \begin{pmatrix}
			\widehat{\chi}_{+,\bk- 2 \pi \Omega \bn} \\
			\widehat{\chi}_{-,\bk- 2 \pi \Omega \bn}
		\end{pmatrix}
		\cdot
		e^{-\ii k_1}
		\begin{pmatrix}
			-1 & \ii \\
			-\ii & -1
		\end{pmatrix}
		\begin{pmatrix}
			 \widehat{\chi}_{+,-\bk} \\
			 \widehat{\chi}_{-,-\bk}
		\end{pmatrix} \\
		&=\frac{1}{2 |\Lambda_i|}\sum_{\substack{\bk \in \mathcal{D}_\balpha \\ \bn \in \mathbb{Z}^2}}  \widehat{V}_\bn^{{(1)}}e^{\ii \bn \cdot \btheta_1} \begin{pmatrix}
			\widehat{\chi}_{+,-\bk} \\
			\widehat{\chi}_{-,-\bk}
		\end{pmatrix}
		\cdot
		\left[
		e^{\ii (k_1-2 \pi \omega_1 n_1)}
		\begin{pmatrix}
			-1 & \ii \\
			-\ii & -1
		\end{pmatrix} -e^{-\ii k_1} \begin{pmatrix}
			-1 & -\ii \\
			\ii & -1
		\end{pmatrix}\right]
		\begin{pmatrix}
			 \widehat{\chi}_{+,\bk- 2 \pi \Omega \bn} \\
			 \widehat{\chi}_{-,\bk- 2 \pi \Omega \bn}
		\end{pmatrix} \\
		&=\frac{1}{|\Lambda_i|} \sum_{\substack{\bk \in \mathcal{D}_\balpha \\ \bn \in \mathbb{Z}^2}} \widehat{V}_{\bn}^{{(1)}} e^{-{ \pi \ii \omega_1 n_1}}e^{\ii \bn \cdot \btheta_1} \begin{pmatrix}
			\widehat{\chi}_{+,-\bk} \\
			\widehat{\chi}_{-,-\bk}
		\end{pmatrix}
		\cdot
		\begin{pmatrix}
			- \ii \sin\big({\textstyle k_1-\pi \omega_1 n_1}\big) & \ii \cos\big({\textstyle k_1-\pi \omega_1 n_1}\big) \\
			-\ii \cos \big({\textstyle k_1-\pi \omega_1 n_1}\big) & -\ii \sin \big({\textstyle k_1-\pi \omega_1 n_1}\big)
		\end{pmatrix}
		\begin{pmatrix}
			 \widehat{\chi}_{+,\bk- 2 \pi \Omega \bn} \\
			 \widehat{\chi}_{-,\bk- 2 \pi \Omega \bn}
		\end{pmatrix} \, .
	\end{split}
\end{equation}
and similar expressions hold for the other quadratic expressions. By setting
\begin{equation}\label{eq:2.17}
	\widehat{A}^{(j)}_\bn =\; \widehat{V}^{(j)}_\bn e^{-{ \ii \pi \omega_j n_j}} e^{ \ii \bn \cdot \boldsymbol{\vartheta}_j} \, ,
\end{equation}
we finally obtain
\begin{equation}
\X_{\boldsymbol\a}=\int \prod_{\kk \in \mathcal{D}_\balpha} \ud\hat\psi_{+,\kk}
\ud\hat\psi_{-,\kk} \ud\hat \chi_{+,\kk} \ud\hat\chi_{-,\kk} e^{S^{(\chi)}_{\mathrm{free}}(\chi)+S^{(\psi)}_{\mathrm{free}}(\psi)+Q_{\mathrm{free}}(\psi,\chi)+S^{(\chi)}_{\mathrm{int}}(\chi)+S^{(\psi)}_{\mathrm{int}}(\psi)+Q_{\mathrm{int}}(\psi,\chi)}
\end{equation}
where, if $\widehat{\bpsi}_{\bk}=(\widehat{\psi}_{\bk,+},\widehat{\psi}_{\bk,-})$ and $\widehat{\bchi}_{\bk}=(\widehat{\chi}_{\bk,+},\widehat{\chi}_{\bk,-})$.
\begin{equation}
	S^{(\textcolor{black}{\zeta})}_{\mathrm{free}}(\textcolor{black}{\zeta}) \;=\; -\frac{1}{4 |\Lambda_i|} \sum_{\bk \in \mathcal{D}_{\balpha}} \textcolor{black}{\widehat{\bfzeta}}_{-\bk} \cdot C_\zeta (\bk) \textcolor{black}{\widehat{\bfzeta}_{\bk}} \, ,
\end{equation}
\begin{equation}
		C_{\textcolor{black}{\zeta}}(\bk) \;:=\;
		\begin{pmatrix}
			-\ii t^{(1)} \sin k_1- t^{(0)} \sin k_0 & -\ii m_{\textcolor{black}{\zeta}}(\bk) \\
			\ii m_{\textcolor{black}{\zeta}}(\bk) & -\ii t^{(1)} \sin k_1 +t^{(0)} \sin k_0
		\end{pmatrix} \, ,
\end{equation}
\begin{equation}\label{eq:MassaChi2}
	m_{\chi}(\bk) \;:=\; t^{(1)} \cos k_1 +t^{(0)} \cos k_0 + 2(\sqrt{2}+1) \, ,
\end{equation}
\begin{equation}\label{eq:MassaPsiIniziale}
	m^0_{\psi}(\bk) \;:=\; t^{(1)}\cos k_1 +t^{(0)} \cos k_0 - 2(\sqrt{2}-1) \, .
\end{equation}
and
\begin{equation}
	Q_{\mathrm{free}}(\psi,\chi) \;=\; \frac{1}{4 |\Lambda_i|} \sum_{\bk \in \mathcal{D}_{\balpha}}  \big[\widehat{\bpsi}_{-\bk} \cdot Q(\bk) \widehat{\bchi}_{\bk}+\widehat{\bchi}_{-\bk}\cdot Q(\bk) \widehat{\bpsi}_{\bk} \big] \, ,
\end{equation}
with
\begin{equation}
Q(\bk)\;:=\;
	\begin{pmatrix}
		\ii t^{(1)}\sin k_1 - t^{(0)}\sin k_0 & \ii \big(t^{(1)}\cos k_1- t^{(0)}\cos k_0 \big) \\
		-\ii \big(t^{(1)} \cos k_1 - t^{(0)} \cos k_0 \big) & \ii t^{(1)}\sin k_1 + t^{(0)} \sin k_0
	\end{pmatrix} \, .
\end{equation}

Moreover,
\begin{equation}\label{eq:InteractingS1}
		S^{(\textcolor{black}{\zeta})}_{\mathrm{int}} \;=\; -\frac{1}{4 |\Lambda_i|} \sum_{\substack{\bk \in \mathcal{D}_\balpha \\ \bn \in \mathbb{Z}^2}} \sum_{j=0,1}  \widehat{A}^{(j)}_\bn \textcolor{black}{\widehat{\bfzeta}}_{-\bk} \cdot P^{(j)}(\bk,\bn) \textcolor{black}{\widehat{\bfzeta}}_{\bk-2\pi\Omega \bn} \, ,
\end{equation}
\begin{equation}\label{eq:InteractingS2}
	Q_{\mathrm{int}}(\psi,\chi) \;=\; \frac{1}{4  |\Lambda_i|} \sum_{\substack{\bk \in \mathcal{D}_\balpha \\ \bn \in \mathbb{Z}^2}} \sum_{j=0,1}  \widehat{A}^{(j)}_\bn  \widehat{\bpsi}_{-\bk}\cdot Q^{(j)}(\bk,\bn) \widehat{\bchi}_{\bk-2\pi\Omega \bn} + (\psi \leftrightarrow \chi) \, ,
\end{equation}
with
\begin{equation}\label{eq:P1}
P^{(1)}(\bk,\bn) =\; \begin{pmatrix}
		-\ii \sin \big({\textstyle k_1 -\pi \omega_1 n_1}\big) & \ii \cos \big({\textstyle k_1-\pi \omega_1 n_1 }\big) \\
		-\ii \cos \big({\textstyle k_1-\pi \omega_1 n_1 }\big) & -\ii \sin \big({\textstyle k_1 -\pi \omega_1 n_1 }\big)
	\end{pmatrix}(1-\delta_{n_1,0}) \, ,
\end{equation} 
\begin{equation}\label{eq:P0}
P^{(0)}(\bk,\bn) \;=\; \begin{pmatrix}
		\sin \big({\textstyle k_0-\pi \omega_0 n_0 }\big) & \ii \cos \big({ k_0-\pi \omega_0 n_0}\big) \\
		-\ii \cos \big({ k_0-\pi \omega_0 n_0}\big) & -\sin \big({\textstyle k_0-\pi \omega_0 n_0 }\big) 
	\end{pmatrix}(1-\delta_{n_0,0}) \, ,
\end{equation}
and 
\begin{equation}\label{eq:Q}
Q^{(1)}(\bk,\bn) \;=\; P^{(1)}(\bk,\bn) \, , \qquad Q^{(0)}(\bk,\bn) \;=\; -P^{(0)}(\bk,\bn) \, .
\end{equation} 
Finally, we introduce new Grassmann variables $\widehat{\bxi}_\bk$
\begin{equation}\label{eq:ChangeOfVariablesXi}
\widehat{\bchi}_{\bk} = 
\widehat{\bxi}_\bk+C^{-1}_\chi(\bk) Q(\bk) \widehat{\bpsi}_\bk
\end{equation} 
and with a straightforward computation yields
\begin{equation}
	S_{\mathrm{free}} \;=\; S^{(\xi)}_{\mathrm{free}}+S^{(\psi)}_{\mathrm{free}} \, , \qquad
	S_{\mathrm{int}} \;=\; S^{(\xi)}_{\mathrm{int}}+S^{(\psi)}_{\mathrm{int}}+Q_{\mathrm{int}}^{(\psi,\xi)} \, .
\end{equation}
Explicitly, we obtain 
$S^{(\xi)}_{\mathrm{free}}(\xi)=
S^{(\chi)}_{\mathrm{free}}(\xi)$ and
\begin{equation}
	S^{(\psi)}_{\mathrm{free}}(\psi) \;=\;-\frac{1}{4 |\Lambda_i|} \sum_{\bk \in \mathcal{D}_\balpha} \widehat{\bpsi}_{-\bk} \cdot
(g_\psi(\bk))^{-1} 
 \widehat{\bpsi}_{\bk}
\end{equation}
with
\begin{equation}\label{eq:CtildePsi}
	(\hat{g}_\psi(\bk))^{-1}  =\; C_\psi(\bk)-Q(\bk) C^{-1}_{\chi}(\bk) Q(\bk) \, ,
\end{equation}
\begin{equation}
Q(\bk) C_\chi^{-1}(\bk) Q(\bk) \;=  M_\psi+R(\bk)
\end{equation}
where, if we denote with $|M|:=\sum_{a,b} |M_{a,b}|$ the chosen norm on the space of matrices, we have
$|R(\bk)|\le C |\bk|$,
$M_\psi=-((t^{(0)}-t^{(1)})^2/m_\chi) \s_2$
and, if $m_\chi(0)=:m_\chi$ and $m^0_\psi(0)=:m^0_\psi$,
\bea
m_\psi\!\!\!\!&=&\!\!\!\!m^0_\psi- (t^{(0)}-t^{(1)})^2/m_\chi \nn\\
\!\!\!\!&=&\!\!\!\!{1\over m^0_\chi}[(t^{(0)}+t^{(1)})-2(\sqrt{2}-1))(t^{(0)}+t^{(0)})+2(\sqrt{2}+1))-(t^{(0)}-t^{(0)})^2)\nn \\
&=&\!\!\!\!{1\over m_\chi}((t^{(0)}+t^{(1)})^2-4+4(t^{(0)}+t^{(0)})-(t^{(0)}-t^{(0)})^2)\nn \\
&=&\!\!\!\!{4\over m_\chi}(t^{(0)} t^{(1)}+t^{(0)}+t^{(1)}-1) \, .\label{eq:MassaPsiInt}
\eea
In conclusion,
\begin{equation}
\Xi_{\boldsymbol\a}={\mathcal{N}}\int P_\xi(d\xi) \int  P_\psi(d\psi) 
 e^{V(\psi,\xi)}
\label{Gra}
\end{equation}
where ${\mathcal{N}}$ is a normalization constant and
$P_\x(d\x)$ is the Gaussian Grassmann integration, see \textit{e.g.}\ Section 4.1 of \cite{GM}, with propagator
$\hat g_\xi(\bk)\equiv C^{-1}_\x(\kk) $
\be
g_\xi(\xx-\yy)={2\over |\L_i|}\sum_{\kk\in\mathcal{D}_\balpha } e^{\ii\kk \cdot (\xx-\yy)} \hat g_\xi(\bk) \, ,\ee			
$P_\psi(d\psi)$ is the Grassmann integration with propagator 
$g_\psi(\bx-\by)={2\over |\L_i|}\sum_\kk e^{\ii\kk \cdot (\xx-\yy)}
\hat g_\psi(\bk)
$ and $V(\psi,\xi)= S^{(\xi)}_{\mathrm{int}}(\xi)+
S^{(\psi)}_{\mathrm{int}}(\psi)+Q_{\mathrm{int}}(\psi,\xi)$
where
\begin{equation}\label{eq:2.39}
S^{(\xi)}_{\mathrm{int}}(\xi) \;=\; -\frac{1}{4 |\Lambda_i|} \sum_{\substack{\bk \in \mathcal{D}_\balpha \\ \bn \in \mathbb{Z}^2}} \sum_{j=0,1}  \widehat{A}^{(j)}_\bn \widehat{\xi}_{-\bk} \cdot P^{(j)}(\bk,\bn) \widehat{\xi}_{\bk-2 \pi\Omega \bn} \, ,
\end{equation}
\begin{equation}\label{eq:2.40}
	S^{(\psi)}_{\mathrm{int}}(\psi) \;=\; -\frac{1}{4  |\Lambda_i|} \sum_{\substack{\bk \in \mathcal{D}_\balpha \\ \bn \in \mathbb{Z}^2}} \widehat{\bpsi}_{-\bk}\cdot \left(\sum_{j=0,1} \widehat{A}_\bn^{(j)}  P_\psi^{(j)}(\bk,\bn)\right) \widehat{\bpsi}_{\bk-2 \pi\Omega \bn} \, ,
\end{equation}
and
\begin{equation}\label{eq:2.40VERA}
	Q_{\mathrm{int}}(\psi,\xi) \;=\; \frac{1}{4 |\Lambda_i|} \sum_{\substack{\bk \in \mathcal{D}_\balpha \\ \bn \in \mathbb{Z}^2}} \sum_{j=0,1}  \widehat{A}^{(j)}_\bn  \widehat{\bpsi}_{-\bk}\cdot 
Q^{(j)}_\psi(\bk,\bn) \widehat{\bxi}_{\bk-2 \pi \Omega \bn} + (\psi \leftrightarrow \chi) \, ,
\end{equation} 
with 
\begin{equation} \label{eq:2.41}
Q_\psi^{(j)}(\bk,\bn) \;=\; Q^{(j)}(\bk,\bn)-Q(\bk) C^{-1}_\x(\bk) P^{(j)}(\bk,\bn) \, ,
\end{equation}
and
\begin{equation}\label{eq:2.42}
	\begin{split}
P^{(j)}_\psi(\bk,\bn)&=\; P^{(j)}(\bk,\bn)-Q^{(j)}(\bk,\bn) 
C^{-1}_\x(\bk- 2 \pi \Omega \bn) Q(\bk- 2 \pi \Omega \bn)\\
&-Q(\bk) C^{-1}_\x(\bk) Q^{(j)}(\bk,\bn)+Q(\bk) C^{-1}_\x(\bk) P^{(j)}(\bk,\bn) C^{-1}_\x(\bk- 2 \pi \Omega \bn) Q(\bk- 2 \pi \Omega \bn) \, .
	\end{split}
\end{equation}
%

\begin{remark}
The partition function is written in terms of Grassmann integrals (see \eqref{Gra}), and a similar representation holds for the energy correlations (see Section \ref{sec:5} below). The propagator $g_{\zeta}(\bx - \by)$ decays exponentially with a rate proportional to $m_{\zeta}(0)$, where $\zeta = \psi, \chi$ and $m_\chi(0) = O(1)$. We call $\psi$ and $\chi$ (or $\xi$) respectively critical and non-critical, or massless and massive variables.

If there is no disorder (i.e., $\lambda = 0$ and $t^{(j)}=\tanh(\beta J^{(j)})$), the critical temperature $\beta_c$, which is the temperature at which the correlation length diverges, is given by the condition $m_\psi = 0$. Indeed, one finds that this happens when $\sinh 2\beta_c J^{(1)} \sinh 2\beta_c J^{(0)} = 1$, noting that $ 4 \frac{t^{(1)}}{1 - (t^{(1)})^2} \frac{t^{(0)}}{1 - (t^{(0)})^2} = 1 $ is true for $t^{(0)} = \frac{1 - t^{(1)}}{1 + t^{(1)}}$. As we will see below, the critical temperature when $\lambda \neq 0$ is different. 
\end{remark}

\section{Integration of non-critical variables}
\label{sec:Xi}
\subsection{Series expansion}

We define
\begin{equation}\label{eq:DefFirstIntegration}
\begin{split}
e^{E^\xi+{\mathcal{V}}(\psi)}&=
\int P_\xi(d\x)   e^{V(\psi,\xi)}=e^{\sum_{q=1}^{\infty}{1\over q!}\EE^T_\xi(V(\psi,\cdot);q) } \\
&=\exp \Bigg({E^\xi+\frac{1}{4 |\Lambda_i|} \sum_{\substack{\bk \in \mathcal{D}_\balpha \\ \bn \in \mathbb{Z}^2}} \bpsi_{-\bk} \cdot \widehat{\mathcal{V}}_\bn(\bk) \bpsi_{\bk-2\pi\Omega \bn}} \Bigg)
\end{split}
\end{equation}
where $\EE^T_\xi(V(\psi,\cdot);q)$ are the truncated expectations with respect to $P_\xi(\ud \xi)$ defined as
\begin{equation} \label{eq:truncatedExpectations}
\EE^T_\xi(V;q)={\partial^q \over \partial \a^q}\log 
\int P_\xi(d\xi) e^{ \a V(\psi,\xi)}\Big|_{\a=0}
\end{equation}
where $\alpha \in \mathbb{R}$ and $E^\xi$ is constant. $\widehat{\mathcal{V}}_\bn(\bk) $ is a $2\times 2$ matrix which can be expressed as sum of connected graphs defined as follows.

\begin{definition}\label{def:Grafici111}
A graph with $q$ vertices and index $\bn$ is defined, see Fig.\ \ref{fig:FeynmanChi}, as a chain of $q$ lines $\ell_1, \dots, \ell_{q+1}$ connecting points (vertices) $v_1, \dots, v_{q}$, so that $\ell_i$ enters $v_i$ and $\ell_{i+1}$ exits from $v_i$; $\ell_1$ and $\ell_{q+1}$ are external lines of the graph and both have a free extreme, while the others are the internal lines. A labeled graph $\Gamma$ is defined from the graph defined above by associating the following labels:

\begin{enumerate}
\item To each point $v$ is associated a label $j_v \in \{0, 1\}$ and a momentum label $\bn_v \in \mathbb{Z}^2$ with the constraint that $\sum_{v=1}^q \bn_{v_i} = \bn$.
\item To each line $\ell$ is associated a momentum $\bk_\ell$ with the constraint that $\bk_{\ell_{i+1}} - \bk_{\ell_i} = -2\pi\Omega \bn_{v_i}$; moreover, $\bk_{\ell_1} = \bk$ and $\bk_{\ell_{q+1}} = \bk - 2\pi\Omega\bn$.
\item $\mathcal{G}_{\bn, q}$ is the set of all possible graphs with $q$ vertices and index $\bn$.
\end{enumerate}
\end{definition}

The value of the labeled graph $\G$ is defined as
\begin{equation}
W_\G(\bk):=F_{v_1}(\bk) \left(\prod_{i=2}^{q}\hat{g}_\xi(\bk_{\ell_i})
F_{v_ i}(\bk_{\ell_{i}}) \right)\label{eq}
\end{equation}
where 
\begin{equation}\label{eq:FaVoloso}
	F_{v}(\bk_\ell) \;=\; \widehat{A}_{\bn_v}^{(j_v)} \times \begin{cases}
		Q_\psi^{(j_v)}(\bk_\ell,\bn_v) \, , \qquad \text{if} \, v=1,q \\
		P^{(j_v)}(\bk_\ell,\bn_v) \, , \qquad \text{if} \, v =2,3,\dots,q-1 \, 
	\end{cases}
\end{equation}
with the definitions in \eqref{eq:defVx}, \eqref{eq:2.17}, \eqref{eq:2.41} and \eqref{eq:2.42}. 

\begin{lemma}\label{lem:LemmaImpestato}
	The effective potential $\mathcal{V}(\psi)$ admits the representation $\frac{1}{4 |\Lambda_i|} \sum_{\bk \in \mathcal{D}_\balpha} \sum_{\bn \in \mathbb{Z}^2} \bpsi_{-\bk} \cdot \widehat{\mathcal{V}}_{\bn}(\bk) \bpsi_{\bk-2 \pi \Omega \bn}$ with
\begin{equation}\label{eq:Wstortoq1}	
	\widehat{\mathcal{V}}_{\bn}(\bk) =\sum_{q=1}^\infty \sum_{\substack{\Gamma \in \mathcal{G}_{\bn,q}}}  W_\G(\bk) \, .
\end{equation}	
\end{lemma}
For the proof, see Appendix \ref{app:Grafici}.

\begin{figure}[h!]
	\begin{center}
\begin{tikzpicture}[thick,scale=0.5]
	\draw[dashed] (0,0) -- (5,0);
	\draw[fill] (5,0) circle[radius=0.09]; 
	\draw[snake it] (5,0) -- (5,4);
	\draw[->] (5,0) -- (7.5,0);
	\draw[-] (7.5,0) -- (10,0);
	\node at(5,-1) {$v_1$};
	\node at(10,-1) {$v_2$};
	\node at (2.5,1) {$\widehat{\psi}_{-\bk}$};
	\node at (7.5,1) {$\hat{g}_\x(\bk_{\ell_2})$};
	\draw[snake it] (10,0) -- (10,4);
	\draw[->] (10,0) -- (12.5,0);
	\draw[-] (12.5,0) -- (15,0);
	\node at (12.5,1) {$\hat{g}_\x(\bk_{\ell_3})$};
	\draw[fill] (10,0) circle[radius=0.09]; 
	\draw[->] (15,0) -- (17.5,0);
	\draw[-] (17.5,0) -- (20,0);
	\node at (17.5,1) {$\hat{g}_\x(\bk_{\ell_4})$};
	\draw[fill] (15,0) circle[radius=0.09]; 
	\draw[snake it] (15,0) -- (15,4);
	\draw[fill] (20,0) circle[radius=0.09]; 
	\draw[snake it] (20,0) -- (20,4);
	\draw[dashed] (20,0) -- (25,0);
	\node at (15,-1) {$v_3$};
	\node at (20,-1) {$v_4$};
	\node at (22.5,1) {$\widehat{\psi}_{\bk-2 \pi\Omega \bn}$};
\end{tikzpicture}
\end{center}
	\caption{A graph $\G$ with $q=4$.}\label{fig:FeynmanChi}
\end{figure}
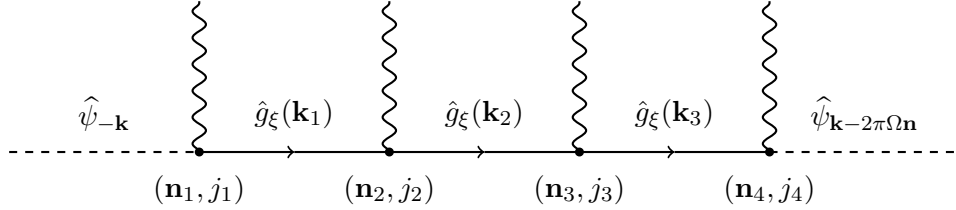

We denote by $|A|:=\sum_{i,j} |A_{i,j}|$, if $A$ is a square matrix. Note that $W_\G(\bk)$ depends on $\bn$.

\begin{lemma}\label{lem:MassiveIntegration2}
	There exist $C,\l_0>0$ independent of $i$ such that for $|\l|\le \l_0$, $\widehat{\mathcal{V}}_\bn(\bk)$ and its derivatives satisfy, for $s\le 2$,
\begin{equation}\label{eq:BoundOnEffectivePotential1}
		|\partial_{\bk}^s\widehat{\mathcal{V}}_\bn(\bk) | \;
 \leq \; C  |\lambda| e^{-\frac{\eta}{2} |\bn|} \, .
	\end{equation}
Moreover,
\begin{equation}
		\widehat{\mathcal{V}}_0(\bk) \;=\; \begin{pmatrix}
			a(\bk) & \ii b(\bk) \\
			-\ii b(\bk) & -a^*(\bk)
		\end{pmatrix}\label{gra}
	\end{equation}
with $a(\bk)=-a(-\bk) \in \mathbb{C}$ and 
$b(\bk)=b(-\bk) \in \mathbb{R}$.  
\end{lemma}
\begin{proof}
Using that $|\partial^s_{k_j} g_{\x}(\bk)| \; \leq \; G_\x $
and \textcolor{black}{recalling that by \eqref{eq:defVx} and \eqref{app} one has $|F_v(\bk_{\mathrm{in}})| \; \leq \; |\lambda| C_1 \, e^{- \eta |\bn_v|} \, $,} and by \eqref{eq:2.17}, \eqref{eq} and \eqref{eq:FaVoloso} we get, for suitable constants $G_\xi, C_1>0$ independent of $i$, 
\begin{equation}
|\partial_\kk^s W_\G(\kk)|
\le  9^q |\lambda|^q G_\x^{q-1} C_1^q \prod_{v }  e^{-\eta |\bn_v|} \; \leq \; 9^q |\l|^q G_\x^{q-1} C_1^q e^{-\frac{\eta}{2} |\bn|} \prod_{v }
 e^{-\frac{\eta}{2} |\bn_v|}
\end{equation}
where $9$ is an upper bound for the number of derivatives on the propagators and on the 
$F_v$'s. The sum over graphs consists simply in the sums over all possible $j_v$ and $\bn_v$ so that, using that 
$\sum_{\bn_v}  e^{-\frac{\eta}{2} |\bn_v|}\le  {4\over 
(1-e^{-\frac{\eta}{2}})^2 }$ and the sum over $j_v$ is bounded by $2$, one gets
\begin{equation}
	|\partial_\kk^s \mathcal{V}_{\bn}(\bk)|\le \; \sum_{q=1}^\io  |\lambda|^q \frac{1}{G_\x}\Bigg(\frac{ 72 C_1 G_\x}{(1-e^{-\frac{\eta}{2}})^2} \Bigg)^q e^{-\frac{\eta}{2} |\bn|}
\end{equation}
and the sum over $q\ge 1$ is convergent for $|\lambda| <  \frac{(1-e^{-\frac{\eta}{2}})^2}{144  C_1 G_\x}$. The proof of \pref{gra} is in Appendix \ref{app:Sym}.
\end{proof}

\section{Integration of critical modes}
\label{sec:Critical}

\subsection{Multiscale decomposition}\label{subsec:PolentaECarciofini}

We write
\bea
\X_{\boldsymbol\a}\!\!\!\!&=&\!\!\!\!\mathcal{N}\int \!\! P_\psi(\ud\psi)  \exp\Big\{ \frac{1}{4 |\Lambda_i|} \sum_{\substack{\bk \in \mathcal{D}_\balpha \\ \bn \in \mathbb{Z}^2}} \hat\bpsi_{-\bk} \cdot \widehat{\mathcal{V}}_\bn(\bk) \hat\bpsi_{\bk-2\pi\Omega \bn}  \Big\}\nn \\
\!\!\!\!&=&\!\!\!\!
\mathcal{N}_1 \!\int \!\! P^{(\le 1)}(\ud\psi) \exp\Big\{ 
 \frac{1}{4  |\Lambda_i|} \sum_{\bk \in \mathcal{D}_\balpha}\!\!\! \hat\bpsi_{-\bk} \cdot \gamma^2 \nu \sigma_2 \hat\bpsi_{\bk}
+
\frac{1}{4  |\Lambda_i|} \sum_{\substack{\bk \in \mathcal{D}_\balpha \\ \bn \in \mathbb{Z}^2}}\!\! \hat\bpsi_{-\bk} \cdot \widehat{\mathcal{V}}_\bn(\bk) \hat\bpsi_{\bk-2\pi\Omega \bn}  \Big\}
\label{Gra2}
\eea
where $P^{(\le 1)}(\ud\psi):= \frac{\mathcal{N}}{\mathcal{N}_1} P(\ud\psi)\exp\{ 
-\frac{1}{4  |\Lambda_i|} \sum_{\bk \in \mathcal{D}_\balpha} \widehat{\bpsi}_{-\bk} \cdot \gamma^2 \nu \sigma_2 \widehat{\bpsi}_{\bk}\} $. 
Note that, in writing the above expression  we have added and subtracted a counterterm
proportional to $\n$, which will be suitably chosen below.

As we noticed, in the integration over the $\psi$ we cannot repeat the analysis
done for the $\x$ because the propagator is unbounded.
The integration of $\psi$ in \pref{Gra2}
is done via a multiscale analysis.
\textcolor{black}{We introduce a Gevrey class 2 function $\chi$ (see e.g.\ \cite[Appendix A]{GMR})} such that $\chi'(|\bk|_{\mathbb{T}}) \leq 0$ and 
\begin{equation}\label{ChiLHaVisto}
	\chi(\bk)\;=\;\chi(|\bk|_{\mathbb{T}})\;=\; \begin{cases}
		1 \, , \qquad \text{if $|\bk|_{\mathbb{T}} <\gamma^{-1}{\pi\over 2}$} \, \\
		0 \, , \qquad \text{if $|\bk|_{\mathbb{T}} \geq {\pi\over 2}$}
	\end{cases}
\end{equation}
with $\mathbb{T}$ denoting the two dimensional torus of length $2 \pi$, $| \bk |_{\mathbb{T}}:=
\sqrt{|k_0|^2_T+|k_1|_T^2}$
with $|k|_T:=\inf_{m \in \mathbb{Z}}|k+2m \pi|$.
We also define, if $\g>1$, $h\le 0$
\begin{equation}\label{eq:chih}
	\chi_h(\bk) \;:=\; \chi(\gamma^{-h} \bk) \, ,
\end{equation}
and $\chi_1(\bk)=1$.
The functions $f_h(\bk):=\chi_h(\bk)-\chi_{h-1}(\bk)$ and $\tilde f_h(\bk):=\chi_{h}(\bk)(1-\chi_{h-1}(\bk))$
\textcolor{black}{are Gevrey class 2 compact support functions} with support {\color{black}$\frac{\pi}{2}\g^{h-2}\le |\bk|_{\mathbb{T}}\le \frac{\pi}{2}\g^{h}$}, see Fig.\ \ref{fig:ChiF}.

\begin{figure}[h!]
\begin{center}
	\begin{tikzpicture}
		\draw[->] (-0.2,0) -- (4.5,0);
		\draw[->] (0,-0.2) -- (0,4);
		\draw (2.0,-0.1) -- (2.0,0.1);

		\draw (3.5,-0.1) -- (3.5,0.1);
		\node at (2.0,- 0.3) {${\textstyle \frac{\pi}{2}}\gamma^{-1}$};
		\node at (3.5,- 0.3) {$\textstyle{\frac\pi2}$};
		\node at (4.5,-0.3) {$|\bk|_{\mathbb{T}}$};
		\node at (0.7,4) {$\chi(|\bk|_{\mathbb{T}})$};
		\draw[color=gray,dashed] (2,0.1) -- (2,3.5);
		\draw[thick, color=blue] (3.5,0) to[out=180,in=0] (2.0,3);
		\draw[thick, color=blue] (0,3) -- (2.0,3);
		\draw[thick, color=blue] (3.5,0) -- (4,0);
	\end{tikzpicture} 
	\begin{tikzpicture}
		\draw[->] (-0.6,0) -- (9,0);
		\draw[->] (-0,-0.3) -- (0,4);
		\draw (4.0,-0.1) -- (4.0,0.1);
		\draw (2.4,-0.1) -- (2.4,0.1);
		\draw (1.4,-0.1) -- (1.4,0.1);
		\draw (0.8,-0.1) -- (0.8,0.1);
		\draw (0.5,-0.1) -- (0.5,0.1);
		\draw (6,-0.1) -- (6,0.1);
		\node at (4.0,- 0.3) {${\textstyle \frac\pi2}\gamma^{-1}$};
		\node at (2.4,- 0.3) {${\textstyle \frac\pi2}\gamma^{-2}$};
		\node at (6,- 0.3) {${\textstyle\frac{\pi}{2}}$};
		\node at (9,-0.3) {$ |\bk|_{\mathbb{T}}$};
		\node at (1.4,4) {$f_h(|\bk|_{\mathbb{T}})$};
		\draw[color=gray,dashed] (4,0.1) -- (4,3.75);
		\draw[color=gray,dashed] (2.4,0.1) -- (2.4,3.75);
		\draw[color=gray,dashed] (1.4,0.1) -- (1.4,3.75);
		\draw[color=gray,dashed] (0.8,0.1) -- (0.8,3.75);
		\draw[color=gray,dashed] (0.5,0.1) -- (0.5,3.75);
		\draw[color=gray, dashed] (6,0.1) -- (6,3.75);
		\draw[thick, color=cyan] (6,0.5) to[out=180, in=0] (4.0,0);
		\draw[thick, color=cyan] (6,0.5) -- (7,0.5);
		\draw[thick, color=cyan] (0,0) -- (4,0);
		\node at (8.6,0.25) {\textcolor{cyan}{$f_1(\bk)$}};
		
		\draw[->, color=gray, dotted] (-0.4,1) -- (9,1); 
		\draw[thick, color= blue] (6,1) to[out=180, in=0] (4.0,1.5);
		\draw[thick, color=blue] (6,1) to (7,1);
		\draw[thick, color=blue] (4.0,1.5) to[out=180, in =0] (2.4,1);
		\draw[thick, color=blue] (2.4,1) to (0,1);
		\node at (8.6,1.25) {\textcolor{black}{$f_0(\bk)$}};
		
		\draw[->, color=gray, dotted] (-0.4,2) -- (9,2);
		\draw[thick, color=violet] (7,2) -- (4,2);
		\draw[thick, color= violet] (4.0,2) to[out=180, in=0] (2.4,2.5);
		\draw[thick, color=violet] (2.4,2.5) to[out=180, in =0] (1.4,2);
		\draw[thick, color=violet] (1.4,2) to (0,2);
		\node at (8.4,2.25) {\textcolor{violet}{$f_{-1}(\bk)$}};
		
		\draw[->, color=gray, dotted] (-0.4,3) -- (9,3);
		\draw[thick, color=magenta] (7,3) -- (2.4,3);
		\draw[thick, color= magenta] (2.4,3) to[out=180, in=0] (1.4,3.5);
		\draw[thick, color=magenta] (1.4,3.5) to[out=180, in =0] (0.8,3);
		\draw[thick, color=magenta] (0.8,3) to (0,3);
		\node at (8.4,3.25) {\textcolor{magenta}{$f_{-2}(\bk)$}};		
		
	\end{tikzpicture}
\end{center}
\caption{Plot of the function $\chi$ and some of the $f_h$.}\label{fig:ChiF}
\end{figure}
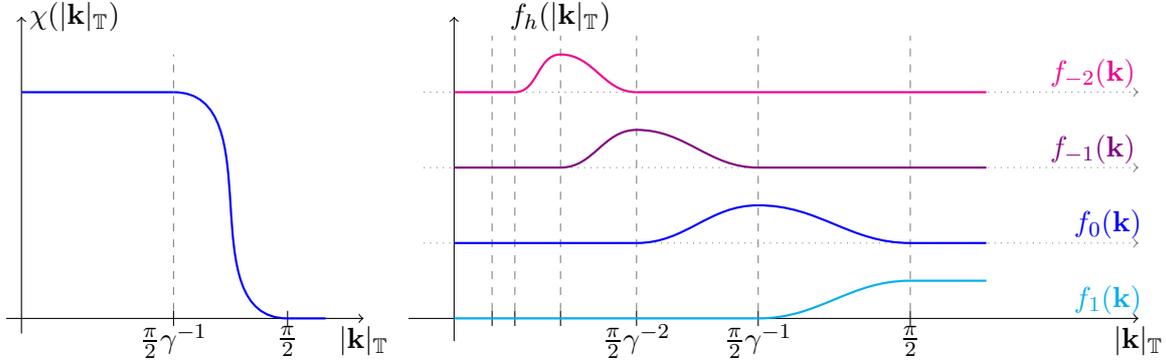

The integration is defined recursively in the following way.
Suppose we have just integrated the field on scale $h$, $h=1,0,-1,-2,\dots$
obtaining
\begin{equation}
	\X_{\boldsymbol\a}=\mathcal{N}_h\int P^{(\leq h)}(\ud \psi^{(\leq h)}) e^{V^{(h)}(\psi^{(\leq h)})}\label{ste1}  \, ,
\end{equation}
with $\mathcal{N}_h$ constant in $\psi$ and
 $P^{(\leq h)}(\ud \psi^{(\leq h)})$ a Grassmann Gaussian integration with propagator 
\begin{equation}\label{eq:PropagatorLeq}
g^{(\le h)}(\bk) = \chi_{h}(\bk) A_{h+1}(\bk)
\end{equation}
with 
\begin{equation}
	 A_h(\bk)=\begin{pmatrix}
		-\ii a_1^{(h)} k_1 - a_0^{(h)}k_0-b_{1}(\bk)  & -\ii \mu-\ii b_2(\bk) \\
		\ii \mu+\ii b_2(\bk) & -\ii(a_1^{(h)})^* k_1 + (a_0^{(h)})^* k_0+b^*_{1}(\bk)
	\end{pmatrix}^{-1} \, .\label{acco}
\end{equation}
and $|b_{1}(\bk)|, |b_{2}(\bk)|\le C |\bk|^2$; moreover
\begin{equation}\label{eq:StellaStellina}
V^{(h)}(\psi^{(\leq h)}) \;=\; \frac{1}{4  |\L_i|} \sum_{\substack{\bk \in \mathcal{D}_\balpha \\ \bn \in \mathbb{Z}^2}} \bpsi_{-\bk}^{(\leq h)} \cdot \widehat{\mathcal{V}}_{\bn}^{(h)}(\bk) \bpsi_{\bk- 2 \pi \Omega \bn}^{(\leq h)} \, .
\end{equation}
If $h=1$, \pref{eq:PropagatorLeq} holds with $\chi_1(\bk)=1$; moreover
$\m=m_\psi+\gamma^2\n$ and $V^{(1)}$ given by the exponent of the second line of
\pref{Gra2}.
\textcolor{black}{
\begin{remark} We will show in the following that $\n$ has to be chosen as a suitable non trivial function
of $\l,\m,\b$; the condition for criticality, that is so that the correlation length diverges, is given by $\m=0$ 
and not by $m_\psi=0$ as in the non disordered case.
\end{remark}
}
%
We define a localization operation as
\begin{equation}
	\mathcal{L} V^{(h)}(\psi^{(\leq h)}) \;:=\; \frac{1}{4 |\Lambda_i|} \sum_{\bk \in \mathcal{D}_{\balpha}} \widehat{\bpsi}^{(\leq h)}_{-\bk}\cdot \bigg(\widehat{\mathcal{V}}^{(h)}_0(0)+\sum_{j=0}^1  k_j \partial_j \widehat{\mathcal{V}}^{(h)}_0(0) \bigg)  \widehat{\bpsi}^{(\leq h)}_{\bk} \, ,\label{loco}
\end{equation}
and
\be \label{eq:RenOp}
\mathcal{R} V^{(h)}(\psi^{(\leq h)}) =
V^{(h)}(\psi^{(\leq h)})-\mathcal{L} V^{(h)}(\psi^{(\leq h)}) \, .
\ee
We move the second term of $\mathcal{L} V^{(h)}(\psi^{(\leq h)})$ in the Gaussian integration and by the \emph{change of integration} property of Gaussian Grassmann Integrals {\color{black}\cite[Eq. 2.24]{GM}}, we have for suitable $\bar{\mathcal{N}}_h \in \mathbb{R}$,
\begin{equation}
\begin{split}
\mathcal{N}_h\int P^{(\leq h)}(\ud \psi^{(\leq h)})& e^{ \mathcal{L}  V^{(h)}(\psi^{(\leq h)})+  \mathcal{R}  V^{(h)}(\psi^{(\leq h)})  }=\\
&=\mathcal{\bar N}_h\int \bar P^{(\leq h)}(\ud \psi^{(\leq h)}) e^{
 \frac{1}{4  |\Lambda_i|} \sum_{\bk \in \mathcal{D}_\balpha} \hat\bpsi^{(\le h)}_{-\bk} \cdot \gamma^{h} \nu_h \sigma_2 \hat\bpsi^{(\le h)}_{\bk}+
\mathcal{R}
V^{(h)}(\psi^{(\leq h)}) }
\label{ste2}  \, ,
\end{split}
\end{equation}
where $\bar P^{(\leq h)}(\ud \psi^{(\leq h)}) $
has propagator
\begin{equation}\label{eq:bargleq}
\bar g^{(\le h)}(\bk) = \chi_{h}(\bk) \bar A_{h}(\bk)
\end{equation}
with
	\begin{equation}
\!\bar{A}_{h}(\bk)\!:=\!\begin{pmatrix}
		-\ii a_1^{(h)}(\bk) k_1 - a_0^{(h)}(\bk)k_0-b_{1}(\bk)  & -\ii \mu-\ii b_2(\bk) \\
		\ii \mu+\ii b_2(\bk) & -\ii(a_1^{(h)}(\bk))^* k_1 + (a_0^{(h)}(\bk))^* k_0+b^*_{1}(\bk)
	\end{pmatrix}^{-1}\!\!\!  \label{accobar}
\end{equation}
and 
\begin{equation}
a_1^{(h)}(\bk)=a_1^{(h+1)}+\ii \chi_h(\bk)\big[\partial_1 
\widehat{\mathcal{V}}_0^{(h)}(0)\big]_{1,1} \, , \qquad a_0^{(h)}(\bk)=a_0^{(h+1)}
-\chi_h(\bk) \big[\partial_0 \widehat{\mathcal{V}}_0^{(h)}(0)\big]_{1,1}
\end{equation}
where $a^{(h+1)}_j:=a^{(h+1)}_j(0)$ for any $j=0,1$ and for any $h$, and with $\nu_h \sigma_2 = 
\gamma^{-h}\widehat{\mathcal{V}}_0^{(h)}(0)$.
To begin the iteration, one can define 
%
\begin{equation}
	a_0^{(2)}:=-\big[ \partial_0 (\hat{g}^{(\leq 1)})^{-1}(0) \big]_{1,1} \, ,\qquad a_1^{(2)}:= \ii \big[ \partial_1 (\hat{g}^{(\leq 1)})^{-1}(0) \big]_{1,1} \, .
\end{equation}
We can write 
\be
\bar P^{(\leq h)}(\ud \psi^{(\leq h)}) = P^{(\leq h-1)}(\ud \psi^{(\leq h-1)})P^{(h)}(\ud \psi^{(h)})
\ee
where $P^{(\leq h-1)}(\ud \psi^{(\leq h-1)})$ has propagator
\begin{equation}\label{propo}
g^{(\le h-1)}(\bk) = \chi_{h-1}(\bk) A_{h}(\bk)
\end{equation}
with $A_{h}(\bk)$ being defined in \pref{acco}. $P^{(h)}(\ud \psi^{(h)})$ has propagator
\begin{equation}\label{eq:gacca}
 g^{(h)}(\bk)= \bar g^{(\le h)}(\bk)-
g^{(\le h-1)}(\bk)
\end{equation}
where the analogous of \pref{gra} has been used.
We can integrate $P^{(h)}(\ud \psi^{(h)})$ and the procedure can be iterated.

\subsection{The single scale propagator}
Inserting \eqref{propo} and \eqref{eq:bargleq} in \eqref{eq:gacca} one obtains
\begin{equation}
	g^{(h)}(\bk)=f_h(\bk) A_{h}(\bk)+ \tilde{f}_h(\bk) \big(\bar A_{h}(\bk)-A_{h}(\bk) \big) \, ,
\end{equation}
where {\color{black}$f_h$ and $\tilde f_h$} are defined after \eqref{eq:chih}.
It is important to notice that  $\mathrm{supp} \,\chi_h(\bk)(\bar{A}_h(\bk)-A_h(\bk)) \subseteq [\frac{\pi}{2} \gamma^{h-1},\frac{\pi}{2} \gamma^h]$ (therefore we can multiply for free with $(1-\chi_{h-1}(\bk)$ to obtain $\tilde{f}_h$) and therefore $g^{(h)}(\bk)$ is a Gevrey compact support function, with $\mathrm{supp} \, g^{(h)} \subseteq [\frac{\pi}{2} \gamma^{h-2},\frac{\pi}{2} \gamma^h]$.
Note also that in the expression of $g^{(1)}$ the second term is not present because $\chi_1(\bk)=1$. 

%
	Assuming iteratively (what will be proved inductively below in Lemma \ref{prop:EstimateResonantClusters} for $|\lambda|$ small enough) that $\frac{7}{8}a_j^{(2)} \leq a_j^{(h)} \leq \frac{9}{8} a_j^{(2)}$, we can show that for $s=0,1,2$,
		\begin{equation}
		| \partial_\bk^s g^{(h)}(\bk) | \leq C_1 \gamma^{-h(1+s)} \, . \label{ello}
	\end{equation}
Indeed,
	\begin{equation}
		|\det \textcolor{black}{A_h^{-1}}(\bk)| = |\ii a_1^{(h)} k_1+a_0^{(h)} k_0 + b_1(\bk) |^2 + |\mu+b_2(\bk)|^2 
	\end{equation}
	with $b_1(\bk),b_2(\bk)=O(|\bk|^2)$ as $\bk \to 0$. Then, by algebraic manipulations, one obtains
	\begin{equation}\label{eq:Determinant}
		|\det \textcolor{black}{A_h^{-1}}(\bk)| \geq |a_1^{(h)}|^2 k_1^2 + |a_0^{(h)}|^2 k_0^2 + 2 \Im (a_1^{(h)} {a_0^{(h)}}^*)k_0k_1+F(\bk)
	\end{equation}
	with $F(\bk)=O(|\bk|^3)$ as $\bk \to 0$. Using now that $a_1^{(2)},a_0^{(2)} \in \mathbb{R}$ and the iterative hypothesis on $a_j^{(h)}$, one has
	\begin{equation}
		\begin{split}
			|\Im (a_1^{(h)} {a_0^{(h)}}^*)|&=|\Im(a_1^{(h)} {a_0^{(h)}}^*-a_1^{(2)} a_0^{(2)})|\\
			&=|\Im((a_1^{(h)}-a_1^{(2)}) {a_0^{(h)}}^*-a_1^{(2)} ({a_0^{(h)}}^*-a_0^{(2)}))| \\
			&\leq |a_1^{(h)}-a_1^{(2)}| |a_0^{(h)}|+|a_1^{(2)}| |a_0^{(h)}-a_0^{(2)}| \leq \left(\frac{1}{8} \cdot \frac{9}{8}+\textcolor{black}{\frac{1}{8}} \right) |a_0^{(2)}| |a_1^{(2)}|\\
			&\leq \textcolor{black}{\frac{17}{64}} |a_0^{(2)}| | a_1^{(2)}| \, .
		\end{split}
	\end{equation}
	Thus, \eqref{eq:Determinant} can be estimated as
	\begin{equation}
		\begin{split}
		|\det \textcolor{black}{A_h^{-1}}(\bk) | &\geq |a_1^{(h)}|^2 k_1^2 + |a_0^{(h)}|^2 k_0^2 - 2 |\Im (a_1^{(h)} {a_0^{(h)}}^*)k_0k_1|-|F(\bk)| \\
		& \geq \frac{49}{64}|a_1^{(2)}|^2 k_1^2 + \frac{49}{64}|a_0^{(2)}|^2 k_0^2 - \textcolor{black}{\frac{17}{32}} |a_0^{(2)}| |a_1^{(2)}| |k_0| |k_1| -|F(\bk)| \\
		&\geq \textcolor{black}{\frac{1}{2}} \Big((a_1^{(2)})^2 k_1^2 + (a_0^{(2)})^2 k_0^2\Big) -|F(\bk)|\, ,
		\end{split}
	\end{equation}
	where in the last step we used $|(a_0^{(2)} k_0) (a_1^{(2)} k_1)| \leq \frac{1}{2}\Big((a_0^{(2)})^2 k_0^2+(a_1^{(2)})^2 k_1^2\Big)$.

\subsection{Graphs and clusters}\label{subsec:RenormalizedGraphAnalysis}
The outcome of the multiscale integration described above is again a representation of the effective potential in terms of graphs, which are called renormalized graphs.

\begin{definition}\label{def:RenGraph}
$\mathcal{G}_{\bn,q}^{R,h}$ 
is the set of {\it renormalized graphs} $\G$, 
which are defined starting from the graphs defined in Definition \ref{def:Grafici111} by associating the following labels
\begin{enumerate}
\item To each point $v$ is associated a label 
$\bn_{v}$ and a label $i_{v}\in \{\n,V\}$,
with the constraint that
$\sum_{i=1}^q \bn_{v_i}=\bn$.
\item
To each line $\ell$ is associated a momentum $\bk_\ell$
with the constraint that $\kk_{\ell_{i+1}}-\kk_{\ell_i}=-2 \pi \O \bn_{v_i}$;
moreover $\kk_{\ell_{1}}=\kk$ and  $\kk_{\ell_{q+1}}=\kk-2\pi\O\bn$.
\item
To each line $\ell$ is associated  a 
{\it scale index} $h_\ell=1,0,\dots,-\infty$; if $\ell$ is an internal line
$h_\ell\ge h+1$; the minimal scale of the internal lines is $h_\G$. To each external line is associated  a scale and $h^{ext}\le h$ is the greatest of such scales.
\end{enumerate}
\end{definition}

Given a renormalized graph, we associate a set of clusters defined in the following way.
\begin{definition}\label{def:ClustersSec4}
Given a renormalized Graph $\G$
\begin{enumerate}
\item A non-trivial cluster $T$ is defined as 
a nonempty connected subset of internal lines
and points attached to them such that if $h_T$ is the minimum
of the scales of the lines of $T$, then $h_T>h^{ext}_T$,
where $h^{ext}_T$ 
is the maximal of the scales of the external lines of $T$
(the lines $\not\in T$ attached to a single point of $T$). The points are trivial clusters and $\G$ is also a 
cluster. 
\item
The difference of the momenta of the external lines of $T$ is given by $2\pi \O \bn_T$ with $\bn_T=\sum_{v\in T}
\bn_v$. 
If $\bn_T=0$ then $T$ is a \emph{resonant cluster} (or resonance), otherwise is a 
\emph{non-resonant cluster}. An inclusion relation is established between clusters
and we say that $\tilde T\subset T$ if all the elements of $\tilde T$ belong also to $T$. 
$\tilde T$ is a maximal cluster (trivial or not trivial) contained in $T$
if $\tilde T\subsetneq T$ and there is no other cluster $\bar T$ such that  $\tilde T\subsetneq \bar T\subsetneq T$.
\item
$Q_T$ is the number of maximal clusters in $T$,
$M_T$ is the number of the maximal non-resonant clusters contained in $T$;
$R_T$ is the number of the maximal resonant clusters contained $T$; $Q_T=M_T+R_T$;
$M_T^\n$  ($M_T^I$ ) is the set of resonant (non-resonant) maximal trivial clusters (\textit{i.e.} points) in $T$.
\end{enumerate}
\end{definition}

Given a cluster, we can associate a value in the following way.
\begin{definition}
The value of a cluster $T$ with maximal clusters $\tilde T_w$, $w=1,\dots,Q_T$
is given by
\begin{equation}\label{assoGatto}
W_T(\kk)
=\;\left[\prod_{w=1}^{Q_T-1} \overline{W}_{\tilde T_{w}}(\bk_w) g^{(h_T)}(\bk_{w+1}) \right] \overline{W}_{\tilde T_{Q_T}}(\bk_{Q_T}) \, ,
\end{equation}
where $\bk_w-\bk_{w-1}=2 \pi\Omega  \bn_{\tilde T_{w-1}} $, $\bk_1=\bk$ and $\overline{W}_{\tilde T_{w}}(\bk_w)$ is defined as
\begin{enumerate}
\item If $\tilde T_w$ is a trivial cluster, then by Definition \ref{def:RenGraph} (item (1)) it has two labels $i_w$ and $\bn_w$. If $i_w=\nu$, then $\bn_w=0$ and $W_{\tilde{T}_w}(\bk_w)=\gamma^{h_T} \nu_{h_T} \sigma_2$; if $i_w=V$ then either $\bn_w=0$ and then $W_{\tilde{T}_w}(\bk_w)=\mathcal{R} \widehat{\mathcal{V}}_0(\bk_w)$ or $\bn_w \neq 0$ and then 
$W_{\tilde{T}_w}(\bk_w)=\widehat{\mathcal{V}}_{\bn_w}(\bk_w)$  defined in \eqref{eq:DefFirstIntegration}.
\item 
If $\tilde T_w$ is a non-trivial cluster
then 
$\bar W_{\tilde{T}_w}=\mathcal{R} W_{\tilde{T}_w}$
with $\mathcal{R}=1-\mathcal{L}$
defined in
\pref{loco}.
\end{enumerate}
\end{definition}

{\color{black}
\begin{remark}\label{rem:CostumeDaBagno}
	Let $v$ be a maximal trivial cluster $v \in T$. If $v$ is a resonant $V$-point (i.e.\ $\bn_v=0$), then by \eqref{eq:RenOp} and Lemma \ref{lem:MassiveIntegration2}, we have
		\begin{equation}
			|\chi_h(\bk) \chi_h(\bk-2 \pi \Omega \bn) \mathcal{R} \hat{V}_0(\bk) | \leq \gamma^{2h_T} C |\lambda| \, .
		\end{equation}
	\end{remark}}

With the above definitions, the following lemma holds.

\begin{lemma}\label{lem:CaronDimonioConOcchiDiBragia}
$\widehat{\mathcal{V}}_\bn^{(h)} (\bk)$ in \pref{eq:StellaStellina} can be written as
\begin{equation}\label{eq:Wstortoq}
	\widehat{\mathcal{V}}_{\bn}^{(h)}(\bk) =\sum_{q=1}^\infty
\sum_{\substack{\Gamma \in \mathcal{G}^{R,h}_{\bn,q}}} 
W_\G(\bk) \, .
\end{equation}
Similarly, the running coupling constants verify
\begin{equation}
\n_{h-1}=\g \n_h+\b_{\n,h}\quad\quad a_j^{(h-1)}=
a_j^{(h)}+\b_{a_j,h}\label{rc}
\end{equation}
with 
\begin{equation}\label{Betadefinizioni}
\begin{split}
\b_{\n,h}&=\ii \g^{-h+1}\sum_{q=2}^\infty
\sum_{\substack{\Gamma \in \mathcal{G}^{R,h-1}_{0,q}, h_\G=h}} 
 \big[W_\G(0)\big]_{1,2} \, , \\
  \b_{a_1,h}
&=-\ii \sum_{q=2}^\infty
\sum_{\substack{\Gamma \in \mathcal{G}^{R,h-1}_{0,q}, h_\G=h}} 
\big[\partial_{k_1}  W_\G(0)\big]_{1,1} \, , \quad 
\b_{a_0,h}= \sum_{q=2}^\infty
\sum_{\substack{\Gamma \in \mathcal{G}^{R,h-1}_{0,q}, h_\G=h}} 
\big[\partial_{k_0}  W_\G(0)\big]_{1,1} \, .
\end{split}
\end{equation}
\end{lemma}
The proof is an immediate consequence of Appendix \ref{app:Grafici} and Section \ref{sec:Xi}.

An example of a renormalized graph with its clusters is given in Fig.\ \ref{fig:ClusterRep}:
in Fig.\ \ref{fig:Structure} is represented the same graph with only its maximal clusters.
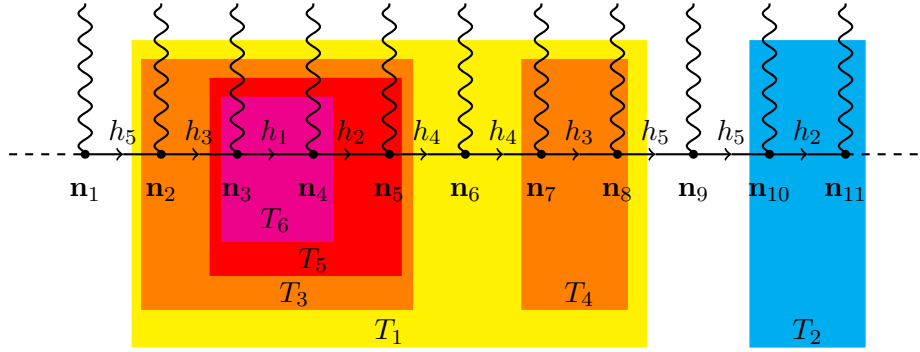
\begin{figure}[h!]
	\begin{center}
\begin{tikzpicture}[thick,scale=0.5]
	\draw[-,color=yellow,fill=yellow] (3.25,-5.1) rectangle (16.75,3);
	\draw[-,color=orange,fill=orange] (3.5,-4.1) rectangle(10.6,2.5);
	\draw[-,color=orange,fill=orange] (13.5,-4.1) rectangle(16.25,2.5);	
	\draw[-,color=red, fill=red] (5.3,-3.2) rectangle(10.3,2.0);
	\draw[-,color=magenta, fill=magenta] (5.6,-2.3) rectangle(8.5,1.5);
	\draw[-,color=cyan, fill=cyan] (19.5,-5.1) rectangle(22.5,3);
	
	\node at(7,-1.75) {$T_6$};
	\node at(8,-2.8) {$T_5$};
	\node at(7.5,-3.7) {$T_3$};
	\node at(15,-3.7) {$T_4$};
	\node at(10,-4.7) {$T_1$};
	\node at(21,-4.7) {$T_2$};

	\draw[dashed] (0,0) -- (2,0);
	\draw[fill] (2,0) circle[radius=0.09]; 
	\draw[snake it] (2,0) -- (2,4);
	\draw[->] (2,0) -- (3,0);
	\draw[-] (3,0) -- (4,0);
	\node at(2,-1) {$\bn_1$};
	\node at(4,-1) {$\bn_2$};
	\draw[fill] (4,0) circle[radius=0.09];
	\draw[snake it] (4,0) -- (4,4);
	\draw[->] (4,0) -- (5,0); 
	\draw[-] (5,0) -- (6,0);
	\node at(6,-1) {$\bn_3$};
	\draw[fill] (6,0) circle[radius=0.09];
	\draw[snake it] (6,0) -- (6,4);
	\draw[->] (6,0) -- (7,0);
	\draw[-] (7,0) -- (8,0);
	
	\node at(8,-1) {$\bn_4$};
	\draw[fill] (8,0) circle[radius=0.09];
	\draw[snake it] (8,0) -- (8,4);
	\draw[->] (8,0) -- (9,0);
	\draw[-] (9,0) -- (10,0);
	
	\node at(10,-1) {$\bn_5$};
	\draw[fill] (10,0) circle[radius=0.09];
	\draw[snake it] (10,0) -- (10,4);
	\draw[->] (10,0) -- (11,0);
	\draw[-] (11,0) -- (12,0);	
	
	\node at(12,-1) {$\bn_6$};
	\draw[fill] (12,0) circle[radius=0.09];
	\draw[snake it] (12,0) -- (12,4);
	\draw[->] (12,0) -- (13,0);
	\draw[-] (13,0) -- (14,0);

	\node at(14,-1) {$\bn_7$};
	\draw[fill] (14,0) circle[radius=0.09];
	\draw[snake it] (14,0) -- (14,4);
	\draw[->] (14,0) -- (15,0);
	\draw[-] (15,0) -- (16,0);

	\node at(16,-1) {$\bn_8$};
	\draw[fill] (16,0) circle[radius=0.09];
	\draw[snake it] (16,0) -- (16,4);
	\draw[->] (16,0) -- (17,0);
	\draw[-] (17,0) -- (18,0);

	\node at(18,-1) {$\bn_9$};
	\draw[fill] (18,0) circle[radius=0.09];
	\draw[snake it] (18,0) -- (18,4);
	\draw[->] (18,0) -- (19,0);
	\draw[-] (19,0) -- (20,0);

	\node at(20,-1) {$\bn_{10}$};
	\draw[fill] (20,0) circle[radius=0.09];
	\draw[snake it] (20,0) -- (20,4);
	\draw[->] (20,0) -- (21,0);
	\draw[-] (21,0) -- (22,0);

	\draw[fill] (22,0) circle[radius=0.09];
	\draw[dashed] (22,0) -- (24,0);
	\draw[snake it] (22,0) -- (22,4);
	\node at(22,-1) {$\bn_{11}$};
	
	\node at(3,0.6) {$h_5$};
	\node at(5,0.6) {$h_3$};
	\node at(7,0.6) {$h_1$};	
	\node at(9,0.6) {$h_2$};
	\node at(11,0.6) {$h_4$};
	\node at(13,0.6) {$h_4$};	
	\node at(15,0.6) {$h_3$};
	\node at(17,0.6) {$h_5$};
	\node at(19,0.6) {$h_5$};
	\node at(21,0.6) {$h_2$};

\end{tikzpicture}
\end{center}

\caption{Graphical representation of a Renormalized graph $\G$: $q=11$, $h_5<h_4<h_3<h_2<h_1$, $h_\G=h_5$. $Q_\Gamma=4$ (with $2$ non-trivial clusters, \textit{i.e.} $T_1$ and $T_2$, and two trivial ones, \textit{i.e.} the points $1$ and $9$). $Q_{T_1}=3$ (with 2 non-trivial clusters $T_3$ and $T_4$ and a trivial one, $v=6$). $T_3$ has two maximal clusters, a trivial one $v=2$ and a non-trivial one $T_5$. $T_5$ has two maximal clusters, a non-trivial one $T_6$ and a trivial one, $n_5$. $T_6$, $T_4$ and $T_2$ have two maximal trivial clusters each.}
\label{fig:ClusterRep}
\end{figure}
Note that a set of clusters can be equivalently represented as a 
\emph{Gallavotti-Nicol\`o} tree, see \textit{e.g.}\ \cite{MastropietroNonPerturbative}.

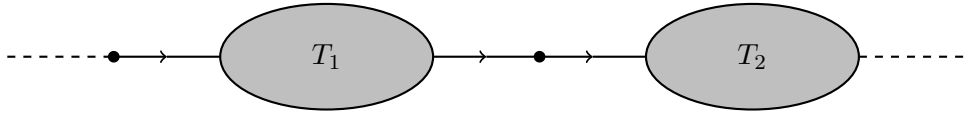
\begin{figure}[!ht]
\begin{center}
\begin{tikzpicture}[thick,scale=0.7]
	\draw[dashed] (0,0)--(2,0);
	\draw[fill]  (2,0) circle[radius=0.09];
	\draw[->] (2,0) --(3,0);
	\draw[-] (3,0)--(4,0);
	\draw[fill=lightgray] (6,0) ellipse (2 and 1); 
	\draw[->] (8,0) -- (9,0);
	\draw[-] (9,0)--(10,0);
	\draw[fill] (10,0) circle[radius=0.09];
	\draw[->] (10,0) --(11,0);
	\draw[-] (11,0) -- (12,0);
	\draw[fill=lightgray] (14,0) ellipse (2 and 1);
	\draw[dashed] (16,0) -- (18,0);
	
	\node at(6,0) {$T_1$};	
	\node at(14,0) {$T_2$};
	
\end{tikzpicture}
\end{center}
\caption{The same graph as in Fig.\ \ref{fig:ClusterRep} with only its maximal clusters represented. Trivial clusters are represented by dots, non-trivial clusters by ellipses.} \label{fig:Structure}
\end{figure}

%
\textcolor{black}{If we consider as first non-trivial cluster $T=\Gamma$ and} we use the above definition we get 
an expression similar to the graphs defined in Section \ref{sec:Xi} with the difference that a)the propagators associated
to the lines $\ell$ are $g^{(h_\ell)}$; b) to each resonant cluster is associated the $\mathcal{R}$
operation; c) the vertices are of type $\nu$ 
or $V$; {\color{black}d) the vertices do not have a $j_v$ index}. In contrast with the expansion in $\l$ seen in Section \ref{sec:Xi}, the renormalized expansion is in $\l$ and
in the running coupling constants $\n_h$.


In the following we denote by $\prod\limits_{T \, \text{n.t.}}=\prod\limits_{\substack{T \in \Gamma \\ T \, \text{non-trivial}}}$.

\subsection{Bounds}
We define
\begin{equation}
\Vert \widehat{\mathcal{V}}_{\bn}^{(h)}\Vert:=\sup_{\kk \in \mathcal{D}_\balpha} \chi_h(\kk) \chi_h(\kk-2 \pi \Omega \bn) |\mathcal{V}_{\bn}^{(h)}(\bk)| \, .
\end{equation}
The following lemma holds. We denote with subscript $l$ the infinite volume limit of a quantity.

\begin{lemma}\label{prop:EstimateResonantClusters}
	Let $\tau:=\min\{\rho_1,\rho_0\}$, take $\gamma > 4^\tau$ and assume that for $h'>h$ one has $|\nu_{h'} |\leq |\lambda|$. Then, there exist $\lambda_0,C>0$ independent of $i$ and $h$ such that, for any $|\lambda| < \lambda_0$ one has
	\begin{itemize}
		\item[(i)] the limit $\widehat{\mathcal{V}}_{\bn,l}^{(h)}(\bk):=\lim_{i \to +\infty} \widehat{\mathcal{V}}_\bn^{(h)}(\bk)$ exists;
		\item[(ii)] for $s=0,1,2$, the following estimates hold
		\begin{eqnarray}\label{eq:ResonantClustersW}
\Vert\partial^s_\kk\mathcal{R}\widehat{\mathcal{V}}_{\bn,l}^{(h)} \Vert \; \hspace{-0.4cm} 
&\leq& \; \hspace{-0.4cm}   \gamma^{ h(1-s)}  C|\lambda|  e^{-\frac{\eta}{4} |\bn|} \, , \\
		|\b_{\n,h}| \hspace{-0.4cm}\; &\le& \hspace{-0.4cm}\; (C\l)^2  \g^{h} \, ,
		\qquad |\b_{a_j,h}|\le  (C\l)^2 \g^{h}	\, .		\label{eq:ResonantClustersW1}
	\end{eqnarray}
	\end{itemize}
\end{lemma}

{\color{black} The bounds \eqref{eq:ResonantClustersW} and \eqref{eq:ResonantClustersW1} are obtained by estimating the value of the graphs in \eqref{eq:Wstortoq} and \eqref{Betadefinizioni}. For clarity, we write in Remark \ref{rem:ConcreteExample} below and example: we show how the general procedure works in the particular case of the graph in Figure \ref{fig:ClusterRep}.}

\begin{proof}
$W_{\G,l}$ is obtained by 
$W_\G$ replacing $\O_i$ with $\O$, considering $\bn_v\in \mathbb{Z}^2$
and $\kk\in [-\pi,\pi)^2$. First, we show that we can multiply \textcolor{black}{by} $\chi_\Gamma$, \textit{i.e.}\ we show that
\begin{equation}
 \chi_h(\kk) \chi_h(\kk-2 \pi \Omega \bn)
W_{\G,l}(\bk)= \chi_h(\kk) \chi_h(\kk-2 \pi \Omega \bn)
\chi_\G W_{\G,l}(\bk) \label{appop}
\end{equation}
where $\chi_\G=1$ if, for any non-resonant cluster $T$ in $\Gamma$, it is true that
\begin{equation}
|\bn_T| \; \geq \; C_0 \gamma^{-\frac{h^{\text{ext}}_T}{\tau}}\label{lem}
\end{equation}
and $\chi_\Gamma=0$ otherwise.
Indeed if $\bk_{\text{in}}$ and $\bk_{\text{out}}$ are the momenta associated
to the external lines of $T$, then by the compact support properties of $g^{(h)}$'s or $\chi_h$,  
$|\bk_{\text{in}}|_{\mathbb{T}}\le {\pi\over 2} \gamma^{h^{ext}_T} $ and $|\bk_{\text{out}}|_{\mathbb{T}}
\le {\pi\over 2} \gamma^{h^{ext}_T}$ (note that $h\le h^{ext}_T$). Therefore
\begin{equation}
|\bk_{\text{in}}-\bk_{\text{out}}|_{\mathbb{T}} \; \le |\bk_{\text{in}}|_{\mathbb{T}} +|\bk_{\text{out}}|_{\mathbb{T}} \le 
2 {\pi\over 2}\gamma^{h^{ext}_T}
\end{equation}
and by the Diophantine condition \pref{dd} we get{\color{black}
\begin{equation}	
\begin{split}
2 {\pi\over 2} \gamma^{h^{ext}_T} &\geq 
|\bk_{\text{in}}-\bk_{\text{out}}|_{\mathbb{T}}=
2 \pi \min_{m_0,m_1 \in \mathbb{Z}}\sqrt{
 (\omega_1 n_1 -m_1)^2+ (\omega_0 n_0-m_0)^2}
\\
&\geq \max_{j=0,1} c_j |n_j|^{-\rho_j} \geq \; \max(c_1,c_0) |\bn_T|^{-\min(\rho_1,\rho_0)}
\end{split}
\end{equation}}
hence the l.h.s.\ of \pref{appop} is vanishing if for \textcolor{black}{at least one non-resonant $T$, \pref{lem} is not true.}

The proof proceeds then 
by induction. First, notice that the first step is a straightforward consequence of Lemma \ref{lem:MassiveIntegration2}. By the inductive step, let us assume that \pref{eq:ResonantClustersW} and 
\pref{eq:ResonantClustersW1} hold for any scale $2,\dots,h+1$ and we prove that they hold at scale $h$. First of all 
by \pref{eq:ResonantClustersW1} we get
\begin{equation}
|a_{j}^{(h)}-a_{j}^{(2)}|\le C^2 \l^2 \sum_{k=h+1}^{2} \g^k\le C^2 \l^2 {\gamma^3\over \gamma-1}
\end{equation}
hence for $C^2 \l^2 \frac{\gamma^3}{\gamma-1}<\min_j \frac{1}{8}a_{j}^{(2)} $ we get  $\frac{7}{8} a_{j}^{(2)}\le a_{j}^{(h)}\le \frac{9}{8} a_{j}^{(2)}$;
this implies, for $s=0,1,2$ \pref{ello}.
To estimate the quantities appearing in \eqref{assoGatto}, we recall that from Lemma \ref{lem:MassiveIntegration2}, there exists a constant $C_2$ independent of $i$, such that
 $\Vert\partial^s_\bk \widehat{\mathcal{V}}_\bn\Vert\le C_2|\l| e^{-\h/2 |\bn|}$, and from \pref{ello} there exists a constant $C_1$ independent from $i$ and $h'$ such that $\Vert\partial^s_\bk g^{(h')} \Vert \leq C_1 \gamma^{-h'(1+s)}$. {\color{black}Moreover, by Remark \ref{rem:CostumeDaBagno}, we can estimate resonant $V$ vertices as $|\mathcal{R} \hat{V}_0 | \leq |\lambda|  	\gamma^{h_T}$ (see also Remark \ref{rem:Caligola} below).} Thus,
\begin{equation}\begin{split}\label{asso1QQQQQ}
\Vert\partial^s_\kk\mathcal{R}&\chi_\G W_{\G,l}(\bk) \Vert	
 \leq (c C_1 C_2)^q \g^{-s h}
|\l|^q\times\\
&\left(\prod_{v } e^{-\frac{\eta}{2}|\bn_v|}\right)
\left( \prod_{T\, \text{n.t.}} \g^{-h_T (M_T+R_T-1)}
\right) 
\Bigg(\prod_{\substack{T\, \text{n.t.} \\\bn_T=0}} \gamma^{2(h_{T}^{\text{ext}}-h_{T})}\Bigg)\prod_{T  \, \text{n.t.}} \gamma^{h_{ T} M^\n_{T}} \, .
\end{split}
\end{equation}
where $c=9$ counts the number of derivatives produced by $\mathcal{R}$ or $\partial^s$,
the factor $\gamma^{2(h_{T}^{\text{ext}}-h_{T})} $ is the result of the application of the $\mathcal{R}$ operation described in Appendix \ref{appendix:LemmaGG1}
and $\g^{-h_T (M_T+R_T-1)}$ comes from the product of propagators.
We can write
\begin{equation}\label{Samurai}
\prod_{v } e^{-\frac{\eta}{2} |\bn_v|}\le  e^{-{\eta\over 4} |\bn|}
\left(\prod_{v } e^{- {\eta\over 8}|\bn_v|}\right)\left(\prod_{v } e^{-{\eta\over 8} |\bn_v|}\right) \, ,
\end{equation}
and
$e^{- {\eta\over 8}|\bn_v|}=\prod\limits_{h=-\io}^0 e^{- 2^h{\eta\over 16}|\bn_v|}$ so that 
\begin{equation}\label{eq:436}
\prod_{v } e^{-{\eta\over 8} |\bn_v|}\le \prod_{
T\, \text{n.t.}}
e^{-{\eta\over 16 } 2^{h_T^{ext}}|\bn_{T}|} \, .
\end{equation}
The presence of $\chi_\Gamma$ guarantees that when $\bn_T \neq 0$, the estimate \eqref{lem} holds and the assumption $\gamma > 4^\tau$ ensures that $\tilde{\gamma}:=\frac{\gamma^{1/\tau}}{2}>1$. Therefore,
\begin{equation}\label{eq:Sambaaa}
\prod_{v } e^{-{\eta\over 8} |\bn_v|}\le\left\{\begin{array}{lcl} e^{-\zeta \tilde{\gamma}^{-h}} \prod\limits_{T\, \text{n.t.}}
e^{-\zeta M_T \tilde\gamma^{-h_{T}} }& \quad & \text{if $\bn_\G\neq 0$} \\
& & \\
\prod\limits_{T\, \text{n.t.}}
e^{-\zeta {M}_{T} \tilde\gamma^{-h_{T}}} &\,&  \text{if $\bn_\G=0$}
\end{array}\right. \, .
\end{equation}
with $\z={\h\over 16}C_0$ a constant independent of $i$ and $h$.
We get therefore
\begin{equation}\label{eq:Intermediate_Before}
\begin{split}
&\Vert\partial^s_\kk\mathcal{R}\chi_\G W_{\G,l}\Vert	
 \leq (c C_1C_2)^q\g^{-s h} |\lambda|^q f_{\text{ext}} e^{-{\eta\over 4} |\bn|}
\left(\prod_{v } e^{- {\eta\over 8}|\bn_v|}\right)
 \left(\prod_{T\, \text{n.t.}} \gamma^{-{h_{T}}(M_{T}+R_{T}-1)} \right)\times
\\
&\qquad\quad\times\Bigg(\prod_{\substack{T\, \text{n.t.}\\\bn_T=0}} \gamma^{2(h_{T}^{\text{ext}}-h_{T})} \Bigg)
\left(\prod_{T\,\text{n.t.}}
e^{-\zeta {M}_{T} \tilde\gamma^{-h_{T}}}\right)
\prod_{T\,\text{n.t.}}  \g^{h_{T} M_{T}^\n} \, ,
\end{split}
\end{equation}
where 
\begin{equation}
	f_{\text{ext}}:=\left\{\begin{array}{lcl} e^{-\zeta \tilde{\gamma}^{-h_T^{ext}}}  & \quad & \text{if $\bn_T\neq 0$} \\
1 & & \text{if $\bn_T=0$}
\end{array}\right. \, .
\end{equation}
\textcolor{black}{Using that for any $M \in \mathbb{N}$, one has $e^{-\zeta \tilde{\gamma}^{-h_T}} \leq \gamma^{-M} \left(\frac{M \ln \gamma}{\zeta} \right)^{M \ln \gamma} \gamma^{M h_T}$ (this is a consequence of the bound  $e^{-\a x} x^M\le ({M\over \a})^M e^{-M}$)}
and $\sum_{T\,\text{n.t.}} M_T \leq 4 q$, we can bound
\begin{equation}
\prod_{T\,\text{n.t.}}
e^{-\zeta M_{T} \tilde{\gamma}^{-h_{T}}} \le  C_3^q
\left(\prod_{T\,\text{n.t.}} \gamma^{2 h_{T} M_{T}} \right)
\prod_{T\,\text{n.t.}} \gamma^{h_{T} M^I_{T}}\label{asxq}
\end{equation}
\textcolor{black}{by setting $M=3$} and with $C_3=\gamma^{-12} \left(\frac{3 \ln \gamma}{\zeta} \right)^{12 \ln \gamma}$. We bound $M_T$ with $M_T^I$, that is the number of non resonant maximal trivial clusters.
Therefore
\begin{equation}
\left(\prod_{T\,\text{n.t.}} \gamma^{-h_{T}  M_{T} } \right)
\prod_{T\,\text{n.t.}}
e^{-\zeta M_{T} \tilde\gamma^{-h_{T}}} \le C_3^q\left(\prod_{T\,\text{n.t.}} \gamma^{ h_{T} M_{ T}} \right)\prod_{T\,\text{n.t.}} \gamma^{h_{ T} M_{T}^I} \, 
\end{equation}
and
\begin{equation}\label{eq:Intermediate_Proof}
	\begin{split}
\Vert\partial^s_\kk\mathcal{R}\chi_\G W_{\G,l} \Vert&\leq (c C_1 C_2 C_3)^q\g^{-s h} |\lambda|^q f_{\text{ext}} e^{-{\eta\over 4} |\bn|}
\left(\prod_{v } e^{- {\eta\over 8}|\bn_v|}\right) \times\\
& \hspace{-1.4cm}
\times \left(\prod_{T\,\text{n.t.}} 
\gamma^{-{h_{T}}(R_{ T}-1)
} \right) \Bigg(\prod_{\substack{T\, \text{n.t.} \\ \bn_T=0} }\gamma^{2(h^{\text{ext}}-h_{T})} \Bigg) \left( \prod_{T\,\text{n.t.}} \gamma^{h_T M_{T}} \right)
\prod_{T\,\text{n.t.}} \gamma^{h_{T}(M^I_{T}+M^\nu_{T})  }.
	\end{split}
\end{equation}
Finally, using that $R_{T}=M_T^\n+R_T^{\text{n.t.}}$ where $R_T^{\text{n.t.}}$
is the number of non-trivial resonant maximal clusters in $T$ we get
\begin{equation}\label{abracadabra}
\Bigg(\prod_{T\,\text{n.t.}} \gamma^{-{h_{T}}(R_{T}-1)
} \Bigg)
\Bigg(\prod_{\substack{T\, \text{n.t.} \\ \bn_T=0}} \gamma^{h_{T}^{\text{ext}}-h_{T}} \Bigg) 
\prod_{T\,\text{n.t.}} \gamma^{h_{T} M^\nu_{ T} }
\le \g^{\e_\G h} 
\end{equation}
with $\e_\G=1$ if $\bn_\G=0$ and $\e_\G=0$ otherwise.
Equation \eqref{abracadabra} follows from the fact that 
\begin{equation}
	\begin{split}
 \Bigg(\prod_{T\,\text{n.t.}} \gamma^{-h_{T}R_T}\Bigg) & \Bigg(
\prod_{\substack{T\, \text{n.t.} \\ \bn_T=0}} 
\gamma^{h_{T}^{\text{ext}}}\Bigg)\prod_{T\,\text{n.t.}} \gamma^{h_{T} M^\nu_{ T} }\\ &=\g^{\e_\G h} 
\Bigg(\prod_{\substack{T\, \text{n.t.} \\ \bn_T=0, \, T\neq\G}} 
\gamma^{h_{T}^{\text{ext}}}\Bigg)
\Bigg(\prod_{T\,\text{n.t.}} \gamma^{-{h_{T}}R_{T}
}\Bigg)\prod_{T\,\text{n.t.}} \gamma^{h_{T} M^\nu_{T} }=\g^{\e_\G h} 
\end{split}
\end{equation} 
and, moreover
\begin{equation}
 \prod_{T\,\text{n.t.}} \gamma^{h_{T}}
\prod_{\substack{T\, \text{n.t.}\\ \bn_T=0}} \gamma^{-h_{T}} \label{all1}  \le 1 \, .
\end{equation}
We define
\begin{equation}\label{defftildee}
	\tilde{f}_{ext}:=\left\{\begin{array}{lcl} e^{-\zeta \tilde{\gamma}^{-h}} & \qquad& \text{if $\bn_\G\neq 0$} \\
\gamma^{h} & & \text{if $\bn_\G=0$}
\end{array}\right. \, .
\end{equation} 
Inserting \eqref{abracadabra} and \eqref{defftildee} in \eqref{eq:Intermediate_Proof}, we get
\begin{equation}\label{eq:QuasiFin}
\begin{split}
\Vert\partial^s_\kk\mathcal{R}\chi_\G W_{\G,l}\Vert
 &\leq \g^{-s h} (c C_1C_2 C_3)^q |\lambda|^q \tilde{f}_{ext} e^{-{\eta\over 4} |\bn|}
\left(\prod_{v } e^{- {\eta\over 8}|\bn_v|}\right)\times\\
&
\times\Bigg(\prod_{\substack{T\, \text{n.t.} \\ \bn_T=0}} \gamma^{h_{T}^{\text{ext}}-h_{T}}\Bigg) \left(\prod_{T\,\text{n.t.}} \gamma^{h_T M_T}\right)
\prod_{T\,\text{n.t.}} \gamma^{h_T M^I_{T}  } \, .
\end{split}
\end{equation}
{\color{black}We use the inequality $e^{-\zeta \tilde{\gamma}^{-h}} \leq \gamma^{-1} \left(\ln \gamma \over \zeta \right)^{\ln \gamma} \gamma^h$, and we call $C_4:=\max\left\{1,\gamma^{-1} \left(\ln \gamma \over \zeta \right)^{\ln \gamma}\right\}$. If $\tilde{T}\subset T$ is maximal and $T$  is non-resonant we have $\g^{h_T}=\g^{h^{\text{ext}}_{\tilde T}}\le
\g^{h^{\text{ext}}_{\tilde T}-h_{\tilde T}}$ and therefore}
\begin{equation} \label{bibidibobidibu}
\tilde{f}_{ext}\Bigg(\prod_{\substack{T\, \text{n.t.}\\ \bn_T=0}} \gamma^{h_{T}^{\text{ext}}-h_T}\Bigg)\prod_{T\, \text{n.t.}} \gamma^{h_{T} M_{T}}
\le C_4 \g^{h} \prod_{T\, \text{n.t.}} \gamma^{h_{T}^{\text{ext}}-h_{T}} \, .
\end{equation}
Inserting \eqref{bibidibobidibu} in \eqref{eq:QuasiFin}, {\color{black}bounding the factor $\prod \g^{h_T M_T^I}$ in \eqref{eq:QuasiFin} by a constant,} we get
\begin{equation}\label{Magicabula}
\Vert\partial^s_\kk\mathcal{R}\chi_\G W_{\G,l} \Vert			
 \;  \leq \;  \gamma^{ h(1-s)} (c C_1C_2C_3C_4)^q|\lambda|^q e^{-{\eta\over 4} |\bn|} \left( \prod_{v } e^{-\frac{\eta}{8} |\bn_v|} \right)
		 \prod_{T\, \text{n.t.}} \gamma^{h_{T}^{\text{ext}}-h_{T}} \, .
\end{equation}

The sum over $\G$ consists in the sum over the label $\bn_v$
associated to the vertices and the sum over the scales. We use that
\begin{equation}\label{SommareAscensori}
\prod_v \sum_{\bn_v}  e^{- {\eta\over 8}|\bn_v|}\le\prod_v 4 \Bigg(\sum_{n \geq 0} e^{-\frac{\eta}{8} n}\Bigg)^2 \leq 4^{q}(1-e^{- {\eta\over 8}})^{-2q} \, .
\end{equation}
The sum over the scale labels of the lines, $h_\ell$ can be controlled by summing over the scales of non-trivial clusters and keeping only the constraint that, for each non-trivial cluster, $h_T^{\text{ext}}<h_T$:
\begin{equation}
\sum_{\{h_\ell\} } \prod_{T\,\text{n.t.}} 
\gamma^{h_{T}^{ext}-h_{T}}= \prod_{T\,\text{n.t.}}  \sum_{h_T> h_{T}^{ext} }\gamma^{h_{T}^{ext}-h_{T}} \leq \Big(\sum_{r>0} \gamma^{-r}\Big)^{\sum_{T} M_T} \le \Big(\frac{1}{\gamma-1}\Big)^{4q} \, ,
\label{apro} 
\end{equation}
where we used again $\sum_{T} M_T \leq 4 q$. Inserting \eqref{SommareAscensori} and \eqref{apro} in \eqref{Magicabula}, we get
\begin{equation}
\Vert\partial^s_\kk\mathcal{R}\chi_\G W_{\G,l} \Vert
\leq \g^{(1-s) h} e^{-{\eta\over 4} |\bn|} |\lambda|^q \bar C^q \, ,
\end{equation}
with $\bar{C}:=4 c C_1 C_2 C_3 C_4 \frac{(1-e^{-\frac{\eta}{8}})^{-2}}{(\gamma-1)^4}$ a constant independent on $i$ and $h$. The sum over $q$ is convergent if $|\lambda| <\bar{C}$, therefore, if $|\lambda| \leq \frac{\bar{C}}{2}$ one gets 
\begin{equation}\label{eq:PastaFagioliETonno}
\Vert \partial_\bk^s \mathcal{R}\widehat{\mathcal{V}}_{\bn,l}^{(h)} \Vert \leq \sum_{q=1}^{+\infty}\textcolor{black}{\sum_{\G \in \mathcal{G}_{\bn,q}}}
\Vert\partial^s_\kk\mathcal{R}\chi_\G W_{\G,l} \Vert
\leq \g^{(1-s) h} 2 \bar C |\lambda| e^{-{\eta\over 4} |\bn|} \, .
\end{equation}
To estimate $\beta_{\nu,h}$ and $\beta_{a_j,h}$,
we have to bound  $W_\G(0)$ for
$\Gamma \in \mathcal{G}^{R,h}_{0,q}$, with $h_\Gamma=h$ (see \eqref{Betadefinizioni}). In this case we have to consider only the case $q\ge 2$, since the sums in \eqref{Betadefinizioni} start from $q=2$. Moreover,
there must be at least two maximal non-resonant clusters in $\Gamma$, therefore $M_\Gamma \geq 1$. Indeed, if this was not the case, then there must be an internal line with $\bk_\ell=0$, implying $g^{(h_\ell)}(\bk_\ell)=0$ by the support properties of $g^{(h)}$'s, which yields $W_\Gamma(0)=0$. Thus, in particular, we must have $M_\Gamma \neq 0$ and $q \geq 2$.

One can repeat the same argument used to estimate $\partial^s_\bk \mathcal{R} \chi_\Gamma W_\Gamma$ with the following difference. By construction, $\Gamma$ is a resonant cluster on which no $\mathcal{R}$ operator acts. Therefore, analogously to \eqref{eq:Intermediate_Before}, one obtains
\begin{equation}
\begin{split}
	|\beta_{\nu,h}| &\leq \sum_{q=2}^\infty \sum_\G  (c C_1 C_2)^q |\l|^q \Bigg( \prod_v e^{-\frac{\eta}{8}|\bn_v|} \Bigg) \Bigg(\prod_{T \, \text{n.t.}} \gamma^{-h_T(M_T+R_T-1)} \Bigg) \\
	&\qquad\times\Bigg(\prod_{\substack{T \, \text{n.t.} \\ \bn_T=0 \\ T \neq \Gamma}} \gamma^{2(h_T^{\text{ext}}-h_T)} \Bigg)\Bigg(\prod_{T \, \text{n.t.}} e^{-\zeta M_T \tilde{\gamma}^{-h_T}} \Bigg)\prod_{T \, \text{n.t.}} \gamma^{h_T M_T^\n} \, .
\end{split}
\end{equation}
Using that $M_\Gamma \neq 0$ and $h_\Gamma=h$, one can replace \eqref{asxq} with
\begin{equation}
	\prod_{T \, \text{n.t.}} e^{-\zeta M_T \tilde{\gamma}^{-h_T}} \leq \gamma^{2h} \tilde{C}_3^q\left(\prod_{T\,\text{n.t.}} \gamma^{2 h_{T} M_{T}} \right)
\prod_{T\,\text{n.t.}} \gamma^{h_{T} M^I_{T}} \label{asx}
\end{equation}
where $\tilde{C}_3=\max\left\{C_3,\gamma^{-20} \left(\frac{5 \ln \gamma}{\zeta}\right)^{20 \ln \gamma}\right\}$.

Moreover,
\begin{equation}
\Bigg(\prod_{T\,\text{n.t.}} \gamma^{-{h_{T}}(R_{T}-1)
}\Bigg)
\Bigg(\prod_{\substack{T\, \text{n.t.}\\ \bn_T=0}} \gamma^{h_{T}^{\text{ext}}-h_{T}} \Bigg)
\prod_{T\,\text{n.t.}} \gamma^{h_{T} M^\nu_{T} }
\le 1 \, .
\end{equation}
The sum over the graphs is done in the same exact way, with regard to the effective potential. To sum over $q$, one first notices that we have no graphs with $q=1$ and therefore the sum starts from $q=2$, and then one proceeds obtaining, for $|\lambda| \leq \frac{\bar C'}{2}$,
\begin{equation}
|\b_{\nu,h}|\le \g^h \sum_{q=2}^\infty \l^q (\bar C')^q\le 2 (\bar C')^2 \l^2\g^h
\end{equation}
for a constant $\bar C'$ independent of $i$ and $h$. With the exactly same argument one proves
\begin{equation}
	|\beta_{a_j,h}| \leq 2 (\bar C')^2 \lambda^2 \gamma^h \, .
\end{equation}
We can therefore choose $C=\max\{2\bar C, 2\bar C'\}$ so that
\eqref{eq:ResonantClustersW} and \eqref{eq:ResonantClustersW1} hold.
Moreover,
\begin{equation}
	\l_0=\min\Bigg\{\frac{1}{C}, \frac{1}{C}\sqrt{\frac{\gamma-1}{8 \gamma^3} \min_j a_j^{(2)}}\Bigg\}
\end{equation}
 so that the inductive step is proved.

{\color{black}
It remains to prove the existence of the limit $i\to +\infty$ where $i$ is the index of the box side $L_i$ introduced in point (iii) after \eqref{eq:UnoPuntoTre}.
the expression obtained replacing $L_i$ with $\infty$ and $\o_i$ with $\o$ is finite.}

Let us denote with $\bar{L}_{i}:=\min\{L_{0,i},L_{1,i}\}$. Define for shortness of notation $\bk(t)=(\bk-\bk_i)t+\bk_i$ where $\bk \in [-\pi,\pi)^2$ and $\bk_i \in \mathcal{D}_{--}$ with $|\bk-\bk_i| \leq \frac{2 \pi}{\bar{L}_i}$, and $\Omega(t):=(\Omega-\Omega_i)t+\Omega_i$ and $\chi_h(\bk,\Omega)=\chi_h(\bk)\chi_h(\bk-2 \pi \Omega \bn)\chi_\Gamma$, then let us consider the term with $\bn\neq 0$ and $s=0$. One has
\begin{equation}\label{eq:Sardine2}
	\begin{split}
		\Vert  \widehat{\mathcal{V}}_{l,\bn}(\bk)&- \widehat{\mathcal{V}}_\bn(\bk_i)\Vert \leq \sum_{q=1}^{+\infty}\sum_{\Gamma \in \mathcal{G}_{\bn,q}}\Vert W_\Gamma(\bk,\Omega)-W_{\Gamma,l}(\bk_i,\Omega_i) \Vert \\
		&=\sum_{q=1}^{+\infty} \sum_{\Gamma \in \mathcal{G}_{\bn,q}} \left\Vert \int_0^1 \frac{d}{dt} \left[\chi_h(\bk(t),\Omega(t)) W_\Gamma(\bk(t),\Omega(t))\right] \, dt \right\Vert \, .\\
\end{split}
\end{equation}
By the Leibnitz rule there are three terms: one in which there is a difference $\bk-\bk_i$ which can be estimated using the same argument of eq.\ \eqref{asso1QQQQQ}-\eqref{eq:PastaFagioliETonno} with an additional term $\frac{2 \pi \gamma^{-h}}{\bar{L}_i}$. Therefore, at the end, this is bounded by $C |\lambda| e^{-\frac{\eta}{4}{|\bn|}} \frac{1}{\bar{L}_i}$. 

In the second term, the derivative can act either on a vertex or on a propagator producing terms that can be estimated as $|\Omega_i-\Omega||\bn_v|\gamma^{-h}$. Then, the procedure to estimate the sum is again similar to \eqref{asso1QQQQQ}-\eqref{eq:PastaFagioliETonno} but in the sum over $\bn$ \eqref{SommareAscensori} one sums $\sum_{n\geq 0} (|n|+1) e^{-\frac{\eta}{8}n}$ to absorb the term $|\bn_v| \leq \prod_{v'} |\bn_{v'}|$. One then uses that $|\Omega_i-\Omega| \leq \frac{C}{\bar{L}_i^2}$ because the sizes of the lattice are the best approximants of the Diophantine numbers $\omega_{0}$ and $\omega_1$ (see Section IV.7 in \cite{D}). Therefore, the second term can be estimated as $C |\lambda| e^{-\frac{\eta}{4}|\bn|} \frac{1}{\bar{L}_i^2}$.

To estimate the third term we use the same procedure from \eqref{asso1QQQQQ} to \eqref{eq:PastaFagioliETonno} and the fact that either $|\bk_i-\bk| \leq 2\pi/\bar{L}_i$ or $|\Omega_i-\Omega| \leq C/\bar{L}_i^2$.

The last term involves graphs with at least a vertex with $\bn_v \geq \bar{L}_i$ and this is $O(e^{-\bar{L}_i})$.

Therefore, there exists a $\lambda_0>0$ and $C>0$ independent of $h$ such that, for any $\l<\l_0$ one has
\begin{equation}
\Vert  \widehat{\mathcal{V}}_{l,\bn}(\bk)- \widehat{\mathcal{V}}_\bn(\bk_i)\Vert \leq \frac{C}{\bar{L}_i} \, .
\end{equation} 
This implies the existence of the limit.
\end{proof}

\begin{remark}\label{rem:ConcreteExample}
Take the Graph in Fig.\ \ref{fig:ClusterRep} and consider the case in which the only resonant cluster is $T_2$. We repeat the argument of Lemma \ref{prop:EstimateResonantClusters}, applied to this graph only, in order to clarify the procedure. One has,
\begin{equation}
	\begin{split}
		\prod_{\ell \in \Gamma} &|g^{(h_\ell)}(\bk_\ell)| \leq C_1^{10} \g^{-h_5} \g^{-h_3} \g^{-h_1} \g^{-h_2} \g^{-2 h_4} \g^{-h_3} \g^{-2h_5} \g^{-h_2} \\
		&=C_1^{10} \g^{-h_4(S_{T_1}-1)} \g^{-h_2(S_{T_2}-1)}\g^{-h_3(S_{T_3}-1)}\g^{-h_3(S_{T_4}-1)}\g^{-h_2(S_{T_ 5}-1)}
\g^{-h_5(S_{\G}-1)}
	\end{split}
\end{equation}
with $Q_{T_1}=3$, $Q_{T_2}=2$, $Q_{T_3}=2$, $Q_{T_4}=2$, $Q_{T_5}=2$, $Q_{T_\G}=3$;
moreover
the action of the $\mathcal{R}$ operator on $T_2$ produces a factor $\g^{2(h_5-h_2)}$, in agreement with \eqref{asso1QQQQQ}.
\end{remark}

{\color{black}
\begin{remark}\label{rem:Caligola}
Note that in \eqref{Magicabula} we have bounded the factor $\prod \g^{h_T M_T^I}$ in \eqref{eq:QuasiFin} with a constant,  and an extra $\g^{h_T}$ coming from the analysis of Remark \ref{rem:CostumeDaBagno} has been estimated by a constant before \eqref{asso1QQQQQ}. Such terms will be used after \eqref{eq:DaRiferirsiQua} in the proof of Lemma \ref{lem:Sandokan} and in the proof Corollary \ref{cor:Ktildone}.
\end{remark}
}

\subsection{The choice of the counterterm}
\label{subsec:Counterterms}

In Lemma \ref{prop:EstimateResonantClusters} we have proved the convergence of the expansion considering $\n_h$ as parameters and provided
that $\n_h$ are small enough. $\n_h$ are determined recursively by \pref{rc} starting from the initial value $\n$ which is a free parameter;
we show that there exists a unique choice of $\n$ so that 
$\n_h$ is bounded uniformly in $h$. We impose the condition $\n_{-\io}=0$ choosing $\n$ verifying
$\b_{\n,h}=\b_{\n,h} (\n_h,\nu_{h+1},\dots,\n;\l)$
\begin{equation}
\n=-\sum_{k=-\io}^{2} \g^k \b_{\n,k}(\n_k,\nu_{k+1},\dots,\n;\l)
\end{equation}
from which
\begin{equation}
\n_h=-\sum_{k=-\io}^h \g^{k-h}\b_{\n,k}(\n_k,\nu_{k+1},\dots,\n;\l)\label{all2}
\end{equation}
and we want to show that \pref{all2} has a solution.

We define the Banach space $\mathcal{M}$ of sequences $\underline\n=\{\n_k\}_{k\le 2}$ with 
norm $\Vert \n \Vert_{\mathcal{M}}:=\sum_{k\le 2}|\n_k|\g^{-k/4}\gamma^{1/2}$ and we consider the ball $\mathcal{B}\subset \mathcal{M}$ of sequences $\underline{\nu}$ such that $\Vert\n\Vert_{\mathcal{M}}\le |\l|$.
We define the 
map $T:\mathcal{M}\to \mathcal{M}$ as 
\begin{equation}
T(\underline\n)_h=-
\sum_{k=-\io}^h \g^{k-h}\b_{\n,k}(\n_k,\nu_{k+1},\dots,\n;\l) \, .
\end{equation}
Therefore, \pref{all2} can be rewritten as 
\begin{equation}\label{eq:ContractionNU}
(\underline\n)_h=T(\underline\n)_h \, .
\end{equation}
\begin{lemma}\label{lem:lemmaGigio} For $|\l|\le\l_0$, $T:\mathcal{B}\to \mathcal{B}$ is a contraction.
\end{lemma}
\begin{proof}
To prove that $T$ leaves $\mathcal{B}$ invariant, we prove a stronger statement: if $\underline\nu$ is such that $|\nu_h| \leq |\lambda|$, then $T(\underline{\nu}) \in \mathcal{B}$. Under these hypothesis, Lemma \ref{prop:EstimateResonantClusters} holds and, using \eqref{eq:ResonantClustersW1}, we get
\begin{equation}
	\Vert T(\underline{\nu}) \Vert_{\mathcal{M}} \leq \sum_{h \leq 2} \sum_{k=-\infty}^h \gamma^{k-h} \gamma^{-\frac{h}{4}}\gamma^{1/2} |\beta_{\nu,k} | \leq \frac{\gamma^{\frac{15}{4}}}{(\gamma-1)(\gamma^{\frac{3}{4}}-1)}(C\lambda)^2 \, ,
\end{equation}
where $C$ is the constant of Lemma \ref{prop:EstimateResonantClusters}.
Choosing now $\lambda_0 \leq \frac{(\gamma-1)(\gamma^{\frac{3}{4}}-1)}{\gamma^{\frac{15}{4}} C^2}$, we get $\Vert T(\underline{\nu}) \Vert_{\mathcal{M}} \leq |\lambda|$.
	If $\underline{\nu},\underline{\nu'} \in \mathcal{M}$, then
	\begin{equation}
		T(\underline\n)_h-T(\underline\n')_h=-\sum_{k=-\io}^h \g^{k-h}\big(\b_{\n,k}(\n_k,\nu_{k+1},\dots,\n;\l)-\b_{\n,k}(\n'_k,\nu'_{k+1},\dots,\n';\l)\big) \, .
	\end{equation}
	The r.h.s.\ can be expressed as a sum of graphs identical to $\G$ with the difference that in a vertex instead of $\n_k$ there is $\n_k-\n'_k$. Indeed, repeating the argument of Lemma \ref{prop:EstimateResonantClusters}, one gets
	\begin{equation}
		\begin{split}
			|\beta_{\nu,h}(\underline{\nu})-\beta_{\nu,h}(\underline{\nu'})| &\leq \sum_{q \geq 2} \sum_{\Gamma} (c C_1 C_2 \bar{C}_3)^q |\l|^{q-\sum_{T} M_T^\n} \gamma^h \Bigg(\prod_v e^{-\frac{\eta}{8}|\bn_v|} \Bigg)\\
			&\times \Bigg( \prod_{T \, \text{n.t.}} \gamma^{h_T^{\text{ext}}-h_T} \Bigg) \Bigg|\prod_{T \, \text{n.t.}}\nu_{h_T}^{M_T^\nu}-\prod_{T \, \text{n.t.}}\nu_{h_T}'^{M_T^\nu}\Bigg|
		\end{split}
	\end{equation}	
	Using now that $\underline{\nu} \in \mathcal{B}$, one has 
$|\prod_{T \, \text{n.t.}}\nu_{h_T}^{M_T^\nu}-\prod_{T \, \text{n.t.}}\nu_{h_T}'^{M_T^\nu}| \leq (2 |\lambda|)^{M_T^\nu-1} \Vert \underline{\nu}-\underline{\nu}'\Vert_{\mathcal{M}}$. Therefore, summing over $\Gamma$ as in Lemma \ref{prop:EstimateResonantClusters}, calling $C_4:=8 c C_1 C_2 \bar{C}_3 (1-e^{-\frac{\eta}{8}})^{-2} (\gamma-1)^{-4}$ one gets
\begin{equation}
	|\beta_{\nu,h}(\underline{\nu})-\beta_{\nu,h}(\underline{\nu'})| \leq \sum_{q \geq 2} C_4^q |\lambda|^{q-1} \gamma^h \Vert \underline{\nu}-\underline{\nu}'\Vert_{\mathcal{M}} \, .
\end{equation}
If $|\lambda| \leq \frac{1}{2 C_4}$, then
\begin{equation}\label{Dinosauri}
	|\beta_{\nu,h}(\underline{\nu})-\beta_{\nu,h}(\underline{\nu}')| \leq 2 C_4^2 |\lambda| \gamma^h \Vert \underline{\nu}-\underline{\nu}'\Vert_{\mathcal{M}} \, .
\end{equation}
Using \eqref{Dinosauri}, we now have
\begin{equation}
	\begin{split}
	\Vert T(\underline{\nu})-T(\underline{\nu}') \Vert_{\mathcal{M}} &\leq \sum_{h\leq 2} \gamma^{-\frac{h}{4}} \gamma^{1/2}|T(\underline{\nu})_h-T(\underline{\nu}')_h| \\
	& \leq \sum_{h \leq 2} \gamma^{-\frac{h}{4}} \sum_{k=-\infty}^h \gamma^{k-h} \gamma^{1/2}|\beta_{\nu,k}(\nu)-\beta_{\nu,k}(\nu')| \\
	&\leq 2\gamma^{1/2} C_4^2 |\lambda| \Vert \underline{\nu}-\underline{\nu}' \Vert_{\mathcal{M}} \sum_{h \leq 2} \gamma^{\frac{3h}{4}} \sum_{k=-\infty}^h \gamma^{2(k-h)} \, .
	\end{split}
\end{equation}
Thus, choosing $\lambda_0 \leq \frac{(\gamma^2-1)(\gamma^{3/4}-1)}{\gamma^{19/4} 2 C_4^2}$, $T$ is a contraction on $\mathcal{B}$.
\end{proof}

\begin{remark}\label{rem:GainNu}
Since, by construction, $\underline{\nu} \in \mathcal{B}$, we have the bound $\sum_{h \geq 2} |\nu_h| \gamma^{-h/4} \gamma^{1/2}\leq |\lambda|$, that implies
\begin{equation}
	|\nu_h| \leq \gamma^{(h-2)/4} |\lambda| \leq |\lambda|
\end{equation} 
which improves, and hence also justifies, the assumption in Lemma \ref{prop:EstimateResonantClusters}.
\end{remark}
\textcolor{black}{
\begin{remark}\label{rem:GainNu1} Equation
\eqref{eq:ContractionNU} and Lemma \ref{lem:lemmaGigio} determines uniquely $\n=\nu(\l,\m,\b)$ and proves the assumption $|\n_h|\le C |\l|$ used in Lemma \ref{prop:EstimateResonantClusters}.
From \eqref{acco},
$\m=m_\psi+\l \g^2 \nu(\l,\m,\b)$ with $m_\psi\equiv m_\psi(\b)$ given by \eqref{eq:MassaPsiInt}. 
The criticality condition is imposed setting $\m=0$; from $0=m_\psi(\b)+\l \g^2 \nu(\l,0,\b)$ we determine the value of $\b_c(\l)=\b_c(0)+O(\l)$ by the implicit function theorem as the derivative is non vanishing. In addition
$\m=O(|\b-\b_c(\l)|)$.
\end{remark}
}

{\color{black}
\section{Energy-energy correlations}
\label{sec:5}
\subsection{Integration of $\xi$ variables}

The energy correlation \pref{en} can be written as
\begin{equation}\label{eq:CinquePuntoUno}
S(\xx_1,j_1;\xx_2,j_2)=
\sum_{\boldsymbol\a \in \{\pm\}^2}\frac{\t_{\boldsymbol\a}Z_{\boldsymbol\a}}{2Z}
{\partial^2\over \partial A_{\xx_1,j_1} \partial A_{\xx_2,j_2}}\mathcal{W}_{\boldsymbol\a}(A)\big|_{A=0}
\end{equation}
with $Z$, $\tau_\balpha$ and $Z_\balpha$ defined as in \eqref{LeEquazioniTutteUguali},
\begin{equation}
\mathcal{W}_{\boldsymbol\a}(A):=\log \int \textcolor{black}{\mathcal{D}^{\Lambda_i}}\Phi \,  e^{S(\Phi,0)+B(\Phi,A) }
\end{equation}
where
\begin{equation}
B(\Phi, A)=\sum_{\xx\in \L_i}\big[
\bar t^{(1)}_\bx (A)  \lis H_{\xx} 
H_{\xx+{\mathbf{e}_1}}+ \bar t^{(0)}_\bx (A) \lis V_{\xx} 
V_{\xx+{\mathbf{e}_0}}\big] 
\end{equation}
and $\bar t^{(j)}_\bx (A)= \tanh( \beta J^{(j)}_{\xx}+A_{\xx,j})-\tanh( \beta J^{(j)}_{\xx})$. Proceeding as in Section \ref{sec:Grassmann} we perform the change of variables $\Phi=\Phi(\chi,\psi)$ defined in \eqref{eq:CambioDiVarHPsi} and then \eqref{eq:ChangeOfVariablesXi} to get 
\begin{equation}
e^{\mathcal{W}_{\boldsymbol\a}(A) }=\int P_\psi(d\psi) \int P_\xi(d\xi) e^{V(\psi,\xi)
+\hat B(\psi,\xi,A)}
\end{equation}
where
\begin{equation}
 \hat B(\psi,\xi,A):=\bar B(\psi, \xi+C_\chi^{-1} Q\psi, A)\, ,\qquad\bar B(\psi,\chi ,A):= B(\Phi(\psi,\chi),A) \, .
\end{equation}
Using the following representation in Fourier series for $A$
\begin{equation}
	A_{\bx,j}:=\frac{1}{|\Lambda_i|} \sum_{\pp \in \mathcal{D}_{++}} \widehat{A}_{\pp,j} e^{\ii \pp\cdot \bx} \, ,
\end{equation}
expanding in Taylor series $\bar{t}_\bx^{(j)}(A)$ around $A=0$ and denoting by $\boldsymbol{\zeta}=\bpsi,\bxi$, one has
\begin{equation}
\begin{split}
\!\!\!\!\hat B(&\psi,\xi,A)=\\
&=
\sum_{s=1}^{+\infty} 
\frac{1}{4|\Lambda_i|^{1+s} } \sum_{\boldsymbol\zeta_1,\boldsymbol\zeta_2=\bpsi,\bxi}
\sum_{\substack{\bk \in \mathcal{D}_{\balpha}, \\ \underline{{\bf p}} \in (\mathcal{D}_{++})^s, \\\bn \in \mathbb{Z}^2 \\ \underline{j} \in \{0,1\}^s}}  \hat{\boldsymbol\zeta}_{1,-\bk} \cdot \widehat{\mathcal{K}}_{\zeta_1,\zeta_2,\bn}(\bk,\underline{{\bf p}},\underline{j})
\hat{\boldsymbol \zeta}_{2,\bk-\sum_{r=1}^s {\bf p}_r-2\pi\Omega \bn} \prod_{r=1}^s \widehat{A}_{{\bf p}_r,j_r}.
\end{split}
\end{equation}
We can integrate over the $\x$ field obtaining
\be
e^{\mathcal{W}_{\boldsymbol\a}(A)}=e^{N_1(A)}
\int \textcolor{black}{P_\psi(d\psi)} e^{V(\psi)
+B^{(1)}(\psi,A)}\label{riva}
\ee
where
\begin{equation}
B^{(1)}(\psi,A)=
\sum_{s=1}^{+\infty} 
\frac{1}{4|\Lambda_i|^{1+s} } 
\sum_{\substack{\bk \in \mathcal{D}_{\balpha}, \\ \underline{{\bf p}} \in (\mathcal{D}_{++})^s, \\ \bn \in \mathbb{Z}^2, \\ \underline{j} \in \{0,1\}^s}} \hat\bpsi_{-\bk} \cdot \widehat{\mathcal{K}}^{2,s,1}_\bn(\bk,\underline{{\bf p}},\underline{j})\hat\bpsi_{\bk-\sum_{r=1}^s {\bf p}_r-2\pi\Omega \bn} \prod_{r=1}^s \widehat{A}_{{\bf p}_r,j_r}
\end{equation}
where $\widehat{\mathcal{K}}^{2,s,1}_\bn(\bk,\underline{{\bf p}},\underline{j})$
can be expressed as sum over graphs $\G$ similar to the ones in Definition 3.1 
with the following differences. 
To each point $v$ of the graph $\G$
is associated a label
$j_v \in \{0,1,2\}$
and momentum label $\bn_v \in \mathbb{Z}^2$, if
$j_v\in\{0,1\}$, or ${\bf p}_v$ if $j_v=2$,
with the constraint that
$\sum_v \bn_v=\bn$ and ${\bf p}_v$ is equal to one of the 
${\bf p}_1,\dots,{\bf p}_s$ or a linear combination of them; the number of points with $j_v=0,1$ is $q$.
To each line $\ell$ is associated a momentum $\bk_\ell$; if $\bk_i$ and 
 $\bk_o$ are two lines attached to the same point $v$, then
$\bk_i-\bk_o=2\pi \O \bn_v$ if $j_v=0,1$ and $\bk_i-\bk_o={\bf p}_v$ if $j_v=2$.
The proof of  Lemma \ref{lem:MassiveIntegration2} can be repeated up to some trivial modifications and 
we get, under the same conditions, the exponential decay of the kernels in $B^{(1)}(\psi,A)$:
\begin{equation}
|\hat{\mathcal{K}}^{2,s,1}_{\bn}(\bk,\underline{\bf p},\underline{j})|\le 
\bar C^s e^{-\frac{\h}{2} |\bn|}\label{pip}
\end{equation}
for a suitable constant $\bar C$.

\subsection{Multiscale analysis}
The integration of \pref{riva} is done inductively, by a generalization of the analysis in Sections \ref{sec:Xi} and \ref{sec:Critical}.
Suppose we have just integrated the scales $1,0,-1,-2,\dots,h+1$
obtaining
\begin{equation}
e^{\mathcal{W}_\balpha(A)}=e^{N_h(A)}\int P^{(\leq h)}(\ud \psi^{(\leq h)}) e^{V^{(h)}(\psi^{(\leq h)})+B^{(h)}(\psi^{(\leq h)},A)}\label{ste3}  \, ,
\end{equation}
with
\begin{equation}
\begin{split}
B^{(h)}(&\psi^{(\le h)}  ,A)=\\&=\sum_{s=1}^{+\infty}
\frac{1}{4|\Lambda_i|^{1+s} } 
\sum_{\substack{\bk \in \mathcal{D}_\balpha \\ \bn \in \mathbb{Z}^2}} \sum_{ \substack{\underline{\bf p} \in (\mathcal{D}_{++})^s \\ \underline{j} \in \{0,1\}^s}}  \hat\bpsi^{(\leq h)}_{-\bk} \cdot
 \widehat{\mathcal{K}}^{2,s,h}_\bn(\bk,\underline{\bf p},\underline{j})\widehat{\bpsi}^{(\leq h)}_{\bk-\sum_{r=1}^s {\bf p}_r-2\pi\Omega \bn}\prod_{r=1}^s \widehat{A}_{{\bf p}_r,j_r}
\end{split}
\end{equation}
and
\begin{equation}
N_h(A)=\sum_{s=0}^{+\infty}\frac{1}{4|\Lambda_i|^{s-1} } 
\sum_{\substack{\underline{{\bf p}}\in (\mathcal{D}_{++})^s \\ \underline{j} \in \{0,1\}^s}} \sum_{\bn \in \mathbb{Z}^2} 
\widehat{\mathcal{K}}^{0,s,h}_\bn(\underline{\bf p},\underline{j})\delta_{\sum_{r} \pp_r+2 \pi \Omega \bn,0} \prod_{r=1}^s\hat{A}_{{\bf p}_r,j_r}
\end{equation}
where $\delta$ denotes the Kronecker delta.
We define a localization operation as
\begin{equation}
	\mathcal{L} \hat B^{(h)}(\psi,A) \;:=
\frac{1}{4|\Lambda_i|^{2} } 
\sum_{\substack{\bk \in \mathcal{D}_{\balpha},  {\bf p} \in \mathcal{D}_{++},\\ \bn \in \mathbb{Z}^2,\\ j \in \{0,1\}}} \hat\bpsi^{(\leq h)}_{-\bk} \cdot \widehat{\mathcal{K}}^{2,1,h}_\bn(0,0,j)\bpsi^{(\leq h)}_{\bk-{\bf p}
-2\pi\Omega \bn} \hat{A}_{{\bf p},j} \, . 	
\end{equation}
Note that, in contrast with the analysis in Section \ref{sec:Critical}, the localization acts also on the terms $\bn\neq0$.
We get therefore
\begin{equation}
\begin{split}
e^{\mathcal{W}_\balpha(A)}&=e^{N_h(A)}\int \bar P^{(\leq h)}(\ud \psi^{(\leq h)}) 
e^{
 \frac{1}{4  |\Lambda_i|} \sum_{\bk \in \mathcal{D}_\balpha} \
\hat\bpsi^{(\le h)}_{-\bk} \cdot \gamma^{h} \nu_h \sigma_2 \hat\bpsi^{(\le h)}_{\bk}}\\
&\qquad e^{  \frac{1}{4|\Lambda_i|^{2} }\sum_{\bn \in \mathbb{Z}^2} \sum_{\bk, {\bf p},j} \hat\bpsi^{(\leq h)}_{-\bk} \cdot Z^{(j)}_{h,\bn} \s_2 
\bpsi^{(\leq h)}_{\bk-{\bf p}
-2\pi\Omega \bn} \widehat{A}_{{\bf p},j}+
\mathcal{R}
\mathcal{V}^{(h)}(\psi^{(\leq h)})+\mathcal{R} B^{(h)}(\psi^{(\leq h)},A)}
\label{ste4}  
\end{split}
\end{equation}
with $Z^{(j)}_{h,\bn}=\widehat{\mathcal{K}}^{2,1,h}_\bn(0,0,j)$.
Note that, in writing the above expression, we have used that $\widehat{\mathcal{K}}^{2,1,h}_\bn(0,0,j)$ is proportional to $\sigma_2$. This latter fact can be checked simply using the anticommutation property of Grassmann variables.
We can write 
$
\bar P^{(\leq h)}(\ud \psi^{(\leq h)}) = P^{(\leq h-1)}(\ud \psi^{(\leq h-1)})P^{(h)}(\ud \psi^{(h)})$
and integrate $\psi^{(h)}$ so that the procedure can be iterated as in Section \ref{sec:Critical}.


{\color{black}
Let us introduce the following definitions.
\begin{definition}
The \emph{special renormalized graphs} are labeled graphs defined starting from the renormalized graphs in Definition \ref{def:RenGraph} with the following additional labels and modifications
\begin{itemize}
	\item[(1)] if $z=2$ the first and the last line are attached to a single point while if $z=0$ there are no external lines.
	\item[(2)] Each point $v$ is associated with a label $S_v$;
if $S_v=0$ (normal point) $v$ is associated with
a label $i_v\in\{\n,V\}$
and a momentum label $\bn_v \in \mathbb{Z}^2$;
if $S_v=1$ (special point) it is associated with a momentum ${\bf p}_v$, an index $j_v \in \mathcal{J}$, \textcolor{black}{a momentum label $\bn_v\in \mathbb{Z}^2$ and an index $\tilde{i}_v\in\{z,B\}$}. The normal points are $q$ and the special ones are $s$. 
	\item[(3)] $\mathcal{G}^{R,z,s,h,\mathcal{J}}_{\bn,q}$ is the set of {\it \textcolor{black}{special} renormalized graphs} $\G$ {\color{black}(here $R$ stands for \emph{renormalized}, $z\in\{0,2\}$, $s\in \{0,1,2\}$, $h$ is the scale and $\mathcal{J}$ is the collection of $j_v$ of the special points)}.
\end{itemize}
\end{definition}
Similarly to what we did in Section \ref{subsec:RenormalizedGraphAnalysis}, to a special renormalized graph we associate a set of clusters in the following way.
\begin{definition}
	Given a special renormalized graph $\Gamma$, we define clusters as in Definition \ref{def:ClustersSec4}. Then, a non-trivial cluster $T$ is associated with $S_T= 1,2$
if it contains $S_T$ special end-point and $S_T=0$ otherwise; in the first
case the cluster is called special, and is associated with a momentum $2\pi \O \bn_T+{\bf p}_T$ (where $\pp_T:=\sum_{v \in T} \pp_v$),
and in the second case is called normal, and it is associated with a momentum  $2\pi \O \bn_T$.
We call $Q_T$ the number of maximal clusters in $T$; $S_T^n=M_T^n+R_T^n$ the number of normal maximal clusters and
$S_T^{sp}$ the number of maximal special clusters;
$M_T^{sp}$  is the set of maximal special trivial clusters (\textit{i.e.}\ points) in $T$.  
The scales are such that, when $z=2$,  $h_\G=h$; when $z=0$ to each external line is 
associated a scale and $h$ is the greatest of such scales.
\end{definition}

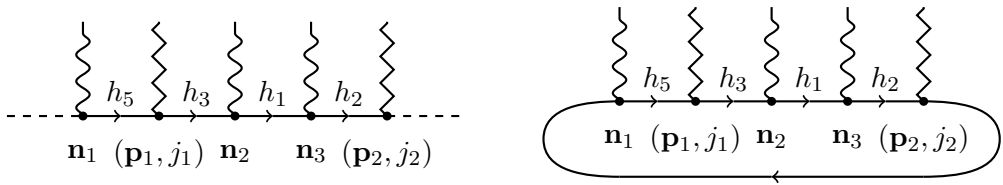
\begin{figure}[h!]
	\begin{center}
\begin{tikzpicture}[thick,scale=0.5]
	\draw[dashed] (0,0) -- (2,0);
	\draw[fill] (2,0) circle[radius=0.09]; 
	\draw[snake it] (2,0) -- (2,2.5);
	\draw[->] (2,0) -- (3,0);
	\draw[-] (3,0) -- (4,0);
	\node at(2,-1) {$\bn_1$};
	\node at(4,-1) {$\substack{(\pp_1,j_1)\\ (\bn_s)_1}$};
	\draw[fill] (4,0) circle[radius=0.09];
	\draw[decorate, decoration=zigzag] (4,0) -- (4,2.5);
	\draw[->] (4,0) -- (5,0); 
	\draw[-] (5,0) -- (6,0);
	\node at(6,-1) {$\bn_2$};
	\draw[fill] (6,0) circle[radius=0.09];
	\draw[snake it] (6,0) -- (6,2.5);
	\draw[->] (6,0) -- (7,0);
	\draw[-] (7,0) -- (8,0);
	
	\node at(8,-1) {$\bn_3$};
	\draw[fill] (8,0) circle[radius=0.09];
	\draw[snake it] (8,0) -- (8,2.5);
	\draw[->] (8,0) -- (9,0);
	\draw[-] (9,0) -- (10,0);
	
	\node at(10,-1) {$\substack{(\pp_2,j_2) \\ (\bn_s)_2}$};
	\draw[fill] (10,0) circle[radius=0.09];
	\draw[decorate, decoration=zigzag] (10,0) -- (10,2.5);
	\draw[dashed] (10,0) -- (12,0);

	\node at(3,0.6) {$h_5$};
	\node at(5,0.6) {$h_3$};
	\node at(7,0.6) {$h_1$};	
	\node at(9,0.6) {$h_2$};
\end{tikzpicture}
$\qquad$
\begin{tikzpicture}[thick, scale=0.5]
	\draw[fill] (2,0) circle[radius=0.09]; 
	\draw[snake it] (2,0) -- (2,2.5);
	\draw[->] (2,0) -- (3,0);
	\draw[-] (3,0) -- (4,0);
	\node at(2,-1) {$\bn_1$};
	\node at(4,-1) {$\substack{(\pp_1,j_1) \\ (\bn_s)_1}$};
	\draw[fill] (4,0) circle[radius=0.09];
	\draw[decorate, decoration=zigzag] (4,0) -- (4,2.5);
	\draw[->] (4,0) -- (5,0); 
	\draw[-] (5,0) -- (6,0);
	\node at(6,-1) {$\bn_2$};
	\draw[fill] (6,0) circle[radius=0.09];
	\draw[snake it] (6,0) -- (6,2.5);
	\draw[->] (6,0) -- (7,0);
	\draw[-] (7,0) -- (8,0);
	
	\node at(8,-1) {$\bn_3$};
	\draw[fill] (8,0) circle[radius=0.09];
	\draw[snake it] (8,0) -- (8,2.5);
	\draw[->] (8,0) -- (9,0);
	\draw[-] (9,0) -- (10,0);
	
	\node at(10,-1) {$\substack{(\pp_2,j_2) \\ (\bn_s)_2}$};
	\draw[fill] (10,0) circle[radius=0.09];
	\draw[decorate, decoration=zigzag] (10,0) -- (10,2.5);
	\draw (10,0) to[out=0,in=90] (12,-1);
	\draw (12,-1) to[out=-90, in=0] (10,-2);
	\draw[->] (10,-2) -- (6,-2);
	\draw (6,-2) -- (2,-2);
	\draw (2,-2) to[out=180, in=-90] (0,-1);
	\draw (0,-1) to[out=90, in=180] (2,0);

	\node at(3,0.6) {$h_5$};
	\node at(5,0.6) {$h_3$};
	\node at(7,0.6) {$h_1$};	
	\node at(9,0.6) {$h_2$};
\end{tikzpicture}
\end{center}

\caption{(left) A graph $\Gamma \in \mathcal{G}^{R,2,2,h,\{j_1,j_2\}}_{\bn,3}$. (right) a graph $\Gamma \in \mathcal{G}^{R,0,2,h,\{j_1,j_2\}}_\bn$.}
\label{fig:ClusterRepCorre}
\end{figure}

\begin{definition}
The value of graph $\G\in\mathcal{G}^{R,2,1,h,\mathcal{J}}_{\bn,q}$ with maximal clusters $\tilde{T}_w$, $w=1,\dots, Q_\G$ is defined as
\begin{equation}\label{asso111}
W_\G({\bf p})
=
\;\left[\prod_{w=1}^{Q_\G-1} \overline{W}_{\tilde T_{w}}(\bk_w) g^{(h_\G)}(\bk_{w+1}) \right] \overline{W}_{\tilde T_{Q_\G}}(\bk_{Q_\G}) 
\end{equation} 
where $\bk_w=\bk_{w-1}-2 \pi\Omega  \bn_{\tilde T_{w-1}}$ if $\tilde T_{w-1}$ is a normal cluster,
$\bk_w=\bk_{w-1}+{\bf p}_{\tilde T_{w-1}}- 2 \pi\Omega  \bn_{\tilde T_{w-1}}$ if $\tilde T_{w-1}$ is a special cluster
$\bk_1=\bk$. $\overline{W}_{\tilde T_{w}}(\bk_w)$ is defined as
\begin{equation}
	\overline{W}_{\tilde{T}_w}
	=\left\{\begin{array}{lcl} Z_{h_T,\bn_s}^{(j_w)} & & \text{if $\tilde{T}_w$ is a special $z$-point} \, ,\\
	\mathcal{R}\hat{\mathcal{K}}^{2,1,1}_{\bn_w}
	& & \text{if $\tilde{T}_w$ is a special $B$-point} \, , \\
		\gamma^{h_T} \nu_{h_T}\sigma_2 & \quad & \text{if $\tilde{T}_{w}$ is a $\nu$-point $(\bn_w=0)$,} \\
			\mathcal{R}\widehat{\mathcal{V}}_{0}
			& & \text{if $\tilde{T}_{w}$ is a $V$-point $(\bn_w= 0)$,}  \\
		\widehat{\mathcal{V}}_{\bn_w}
		& & \text{if $\tilde{T}_{w}$ is a $V$-point $(\bn_w\neq 0)$,}  \\
		\mathcal{R}W_{\tilde{T}_w}
		& & \text{if $\tilde{T}_w$ is a non-trivial cluster.}
\end{array}\right.
\end{equation}
Similarly, if the special renormalized graph is $\G\in\mathcal{G}^{R,0,2,h,\mathcal{J}}_{\bn,q}$
\begin{equation}\label{asso}
W_\G({\bf p})
={1\over 4 |\L_i|}\sum_{\bf k \in \mathcal{D}_\balpha}
\;\left[\prod_{w=1}^{Q_\Gamma-1} \overline{W}_{\tilde T_{w}}(\bk_w) g^{(h_\G)}(\bk_{w+1}) \right] \overline{W}_{\tilde T_{Q_\G}}(\bk_{Q_\G}) g^{(h_\G)}(\bk_{Q_\Gamma})
\end{equation} 
with $\bk_1=\bk$.
\end{definition}

\begin{lemma}
The kernels $\mathcal{K}$ can be written as a sum of graphs
\begin{equation}
\begin{split}
\widehat{\mathcal{K}}^{2,s,h}_\bn(\bk,\underline{\pp},\underline{j}) &=\sum_{q=0}^\infty
\sum_{\substack{\Gamma \in \mathcal{G}^{R,2,s,h,\mathcal{J}}_{\bn,q}}} 
W_\G(\bk,\underline{\pp},\underline{j}) \, , \\
\widehat{\mathcal{K}}^{0,s,h}_\bn(\underline{\bf p},\underline{j}) &=\sum_{q=0}^\infty
\sum_{\substack{\Gamma \in \mathcal{G}^{R,0,s,h,\mathcal{J}}_{\bn,q}}} 
W_\G(\underline{\pp},\underline{j})
\end{split}
\end{equation}
and the running coupling constants verify
%
\begin{equation}
Z^{(j)}_{h-1,\bn}=Z^{(j)}_{h,\bn}+\b_{z,\bn, h}^{(j)} \, ,\qquad \b_{z,\bn,h}^{(j)}=\sum_{q=1}^\infty
\sum_{\substack{\Gamma \in \mathcal{G}^{R,2,1,h-1,j}_{\bn,q} \\  h_\G=h}} 
W_\G(0,0,j) \, . \label{rc2}
\end{equation}
\end{lemma}
Also in this case, the proof follows along the lines Appendix \ref{app:Grafici} and Lemma \ref{lem:CaronDimonioConOcchiDiBragia}.
}

\subsection{Bounds}
Let us now define
\begin{equation}
	\SVert \widehat{\mathcal{K}}^{2,1,h}_\bn \SVert := \sup_{j \in \{0,1\}}\sup_{\pp \in \mathcal{D}_{++}} \sup_{\bk \in \mathcal{D}_\balpha} |\chi_h(\bk) \chi_h(\bk+\pp-2\pi \Omega \bn) \widehat{\mathcal{K}}^{2,1,h}_\bn(\bk, \pp,j)| \, .
\end{equation}
We will denote by $\prod_{v \, \text{n.s.}}=\prod_{v \in \Gamma \, , S_v=0}$.
\begin{lemma}\label{lem:Sandokan}
If $|\l|\le \l_0$ and $\n$ is chosen as in Section \ref{subsec:Counterterms}, then there exists a constant $C$ independent of $i$,$\beta$ and $h$ such that
\begin{equation}
\SVert\widehat{\mathcal{K}}^{2,1,h}_\bn\SVert \;  
\leq   C e^{-\frac{\eta}{4} |\bn|} \, ,\qquad  \sup_{\pp_1,\pp_2\in \mathcal{D}_{++}}|\widehat{\mathcal{K}}^{0,2,h}_\bn(\pp_1,\pp_2,j_1,j_2)| \;  
\leq C e^{-\frac{\eta}{8} |\bn|}  
\label{m1}
\end{equation}
and 
\begin{equation}
|\b_{z,\bn,h}^{(j)}|\le C|\l|\g^{h} e^{-\frac{\eta}{4} |\bn|}\label{m2} \, .
\end{equation}
\end{lemma}
\begin{proof}
Assume inductively that the statement is valid for $k\ge h+1$; then $|Z_{\bn,2}^{(j)}|\le C_1 e^{-\h |\bn|}$
by \pref{pip} and by induction
\begin{equation}
Z_{\bn,h}^{(j)}=Z_{\bn,2}^{(j)}+\sum_{r=h}^2 \b_{z,\bn,r}^{(j)}\le 
2 C_Z  e^{-\frac{\eta}{4} |\bn|}
\end{equation}
assuming $|\l|4 C (1-e^{-\frac{\eta}{4} })^2\le C_1$.

We start from the first of \pref{m1}. Considering that the operator $\mathcal{R}$ acting on a special cluster gives a factor $\gamma^{h_T^{\text{ext}}-h_T}$, one proceeds as  in the proof of Lemma \ref{prop:EstimateResonantClusters} to get (instead of \eqref{asso1QQQQQ})
\begin{equation}\label{asso1TER}
\begin{split}
\SVert \chi_\G W_{\G,l} \SVert	
 \leq& \bar C_1 (c C_1 C_2)^q 
|\l|^q\ \Bigg(\prod_{\substack{T\, \text{n.t.} \\ S_T=1 \\ T\neq\G}} \gamma^{h_{T}^{\text{ext}}-h_{T}}\Bigg)e^{-\frac{\eta}{4}|\bn_s| }\left(\prod_{v \, \text{n.s.}} e^{-\frac{\eta}{2}|\bn_v| }\right)  \\
&\times \left( \prod_{T\,\text{n.t.}} \g^{-h_T (M^n_T+R^n_T+S_T^{sp}  -1)}
\right) 
\Bigg(\prod_{\substack{T\, \text{n.t.}\\ \bn_T=0 \\ S_T=0}} \gamma^{2(h_{T}^{\text{ext}}-h_{T})}\Bigg)\prod_{T\,\text{n.t.}} \gamma^{h_{ T} M^\n_{T}}
\end{split}
\end{equation}
where $c=18$ (up to 2 derivatives to points and vertex, with $j=0,1$),
$
\bn_s$ is the momentum label of the special point,
$S_T^{sp}$ is the number of special end-points contained in $T$.
We can write
\begin{equation}
1=\g^{-h_\G} \Bigg(\prod_{\substack{T\, \text{n.t.} \\ S_T=1 \\ T\neq\G}} \gamma^{h_{T}^{\text{ext}}-h_{T}} \Bigg)
\prod_{T\,\text{n.t.}} \gamma^{h_{ T} M_T^{sp}} \, .
\end{equation}
We get therefore
\begin{equation}\label{asso1TopoGigio}
\begin{split}
\SVert \chi_\G W_{\G,l} \SVert	
 \leq& (c C_1 C_2)^q 
|\l|^q \g^{-h_\G} \Bigg(\prod_{\substack{T\, \text{n.t.} \\ S_T=1 \\ T\neq\G}} \gamma^{2(h_{T}^{\text{ext}}-h_{T})}\Bigg)\left(\prod_{v \, \text{n.s.}} e^{-\frac{\eta}{2}|\bn_v| }\right)
e^{-\frac{\eta}{4}|\bn_s| }  \\
&
\hspace{-2cm}\times\left( \prod_{T\,\text{n.t.}} \g^{-h_T (M^n_T+R^n_T+S_T^{sp}  -1)}
\right) 
\Bigg(\prod_{\substack{T\, \text{n.t.} \\ \bn_T=0 \\ S_T=0}} \gamma^{2(h_{T}^{\text{ext}}-h_{T})}\Bigg)\left(\prod_{T\,\text{n.t.}} \gamma^{h_{ T} M^\n_{T}}\right)
\prod_{T\,\text{n.t.}} \gamma^{h_{ T} M_T^{sp}}  \, .
\end{split}
\end{equation}
We now use that
\begin{equation}\label{eq:Calimero}
\begin{split}
&\left(\prod_{T\,\text{n.t.}} \g^{-h_T (R^n_T+S_T^{sp}-1)}\right)\Bigg(\prod_{\substack{T\, \text{n.t.} \\ S_T=1 \\ T\neq\G}} 
\gamma^{h_{T}^{\text{ext}}-h_{T}}\Bigg) \\
&\qquad \times \Bigg(\prod_{\substack{T\, \text{n.t.} \\ \bn_T=0 \\ S_T=0}} \gamma^{h_{T}^{\text{ext}}-h_{T}}
\Bigg)\left(\prod_{T\,\text{n.t.}} \gamma^{h_{ T} M^\n_{T}}\right)\prod_{T\,\text{n.t.}} \gamma^{h_{ T} M^{sp}_{T}} \le \g^{h_\G}
\end{split}
\end{equation}
and following the same argument of Lemma \ref{prop:EstimateResonantClusters} from \eqref{eq:436} to \eqref{eq:Intermediate_Proof} we get rid of all $\gamma^{-h_T M_T}$'s and,
since $\G$ is a special cluster, we finally obtain
\begin{equation}\label{asso1df}
\begin{split}
\SVert\chi_\G W_{\G,l} \SVert	
 \leq& \bar C_1 (c C_1 C_2 C_3)^q 
|\l|^q\Bigg(\prod_{\substack{T\,\text{n.t.} \\ S_T=1\\ T\neq\G}} 
\gamma^{h_{T}^{\text{ext}}-h_{T}}\Bigg) \Bigg(\prod_{\substack{T\,\text{n.t.} \\ S_T=0}} \gamma^{h_{T}^{\text{ext}}-h_{T}} \Bigg) \\
&\times e^{-\frac{\eta}{4}|\bn_s| }\left(
\prod_{v \, \text{n.s.}} e^{-\frac{\eta}{8}|\bn_v| }\right)  \prod_{v \, \text{n.s.}} e^{-\frac{\eta}{4}|\bn_v| } \, .
\end{split}
\end{equation}
To handle the sum over $\G$, we perform the sum over the scales as in Lemma \ref{prop:EstimateResonantClusters},
while in the sum over $\bn_v$'s one uses that $\bn$ is fixed, and
$\sum_{v=1}^q \bn_v+\bn_s=\bn$: the sum over $\bn_v$'s and $\bn_s$ can be performed only on $\bn_1,\dots,\bn_q$. Thus, using triangular inequality $\sum_{v=1}^q|\bn_v|+|\bn_s| \geq |\bn|$ in the last product of \eqref{asso1df}, one has
\begin{equation}
e^{-\frac{\eta}{4}|\bn_s| }\left(
\prod_{v \, \text{n.s.}} e^{-\frac{\eta}{8}|\bn_v| }\right)  \prod_{v \, \text{n.s.}} e^{-\frac{\eta}{4}|\bn_v| } \leq e^{-\frac{\eta}{4}|\bn|}\left(
\prod_{v \, \text{n.s.}} e^{-\frac{\eta}{8}|\bn_v| }\right)
\end{equation}
and then one can sum over $\bn_1,\dots,\bn_q$ as in Lemma \ref{prop:EstimateResonantClusters}.
Therefore
\begin{equation}
\sum_{\substack{\Gamma \in \mathcal{G}^{R,2,s,h,\mathcal{J}}_{\bn,q}}} 
\SVert \chi_\G W_{\G,l} \SVert
\leq  e^{-{\eta\over 4} |\bn|}|\lambda|^q \bar C_1 \bar C^q
\end{equation}
with $\bar C=(c\times 3) C_1C_2 C_3 {(1-e^{- {\eta\over 8}})^{ -2}\over (1-1/\g))}$;
by summing over $q$ we get, for $|\l|\le \bar C/2$ we get
\begin{equation}\label{eq:DaRiferirsiQua}
\sum_{q\ge 0} \sum_{\substack{\Gamma \in \mathcal{G}^{R,2,s,h,\mathcal{J}}_{\bn,q}}} 
\SVert\chi_\G W_{\G,l} \SVert
\leq  e^{-{\eta\over 4} |\bn|} \bar C_1 2\bar C \, .
\end{equation}

We choose $|\l| \le \min \{C/2, C_1/(4 C (1-e^{-\frac{\eta}{4} })^2\}$,
with $C=\bar C_1 4 C_1 C_2$.

In order to prove \pref{m2}, we note that we have to bound $\widehat{\mathcal{K}}^{2,1,h}_\bn(0,0,j)$. This is exactly the same argument used to prove \eqref{eq:ResonantClustersW1}.
%
%

Finally we have to prove the second of \pref{m1}.
Using that $\G\in\mathcal{G}^{R,0,2,h}_{\bn,q}$, the analogue of \eqref{asso1QQQQQ} becomes
\begin{equation}\label{asso1LAT}
\begin{split}
|W_{\G,l}(\pp_1,\pp_2,j_1,j_2)|	
 \leq& (4 C_1 C_2)^q \g^{2 h_\G} e^{-\frac{\eta}{4} |\bn_{s_1}|} e^{-\frac{\eta}{4} |\bn_{s_2}|}
|\l|^q\ \Bigg(\prod_{\substack{T\, \text{n.t.}\\ S_T=1 \\ T\neq\G}} \gamma^{h_{T}^{\text{ext}}-h_{T}}\Bigg)\\
&\times\left(\prod_{v \, \text{n.s.}} e^{-\frac{\eta}{2}|\bn_v| }\right)
\left( \prod_{T\,\text{n.t.}} \g^{-h_T (M_T^n+R^n_T+S_T^{sp}  -1+\d_T)}
\right) 
 \\&\times
\Bigg(\prod_{\substack{T\, \text{n.t.} \\ \bn_T=0 \\ S_T=0}} \gamma^{2(h_{T}^{\text{ext}}-h_{T})}\Bigg)\prod_{T\,\text{n.t.}} \gamma^{h_{ T} M^\n_{T}}
\end{split}
\end{equation}
where the extra $\g^{ 2 h_\G}$ comes from the integration over $\bk$ and the compact support properties of the propagators at scale $h_{\G}$; moreover $\d_\G=1$ and $\d_T=0$ if $T\neq\G$. 
Note that, since $\Gamma$ has two special points, we have

\begin{equation}
1=\g^{-2h_\G} \Bigg(\prod_{\substack{T\,\text{n.t.}, \\ T\neq\G}} \gamma^{S_T(  h_{T}^{\text{ext}}-h_{T})}\Bigg)
\prod_{T\,\text{n.t.}} \gamma^{h_{ T} M_T^{sp}} \, .
\end{equation}
To prove the second of \eqref{m1}, one repeats the argument used to prove the first of \eqref{m1} with, instead of \eqref{eq:Calimero}, the following
\begin{equation}\label{eq:CastoriSelvaggi}
\begin{split}
&\Bigg(
\prod_{\substack{T\,\text{n.t.} }} \g^{-h_T (R^n_T+S_T^{sp}+\d_T  -1)}\Bigg) \Bigg(\prod_{\substack{T\,\text{n.t.} \\ S_T=1,2 \\ T\neq\G}} 
\gamma^{h_{T}^{\text{ext}}-h_{T}}\Bigg) \\
&\qquad\qquad \times \Bigg(\prod_{\substack{T\,\text{n.t.} \\ \bn_T=0 \\ S_T=0}} \gamma^{h_{T}^{\text{ext}}-h_{T}}\Bigg)\left(
\prod_{T\,\text{n.t.}} \gamma^{h_{ T} M^\n_{T}}\right)\prod_{T\,\text{n.t.}} \gamma^{h_{ T} M_T^{sp}} \le 1\, .
\end{split}
\end{equation}
Then, to perform the sum over $\bn_v$'s one now can not repeat the previous argument to isolate the $\bn_s$ as one has to sum to at least one of them. Thus, one has
\begin{equation}\label{eq:SummationOverNKappaStorti}
	e^{-\frac{\eta}{4}\sum_{v \text{ special}}|\bn_v|}e^{-\frac{\eta}{4}\sum_{v=1}^q |\bn_v|} \leq e^{-\frac{\eta}{8}|\bn|} e^{-\frac{\eta}{8}\sum_{v \text{ special}}|\bn_v|} \prod_{v \, \text{n.s.}}e^{-\frac{\eta}{8} |\bn_v|} \, .
\end{equation}
The rest of the proof proceeds as in the cases before.
\end{proof}
We now denote by $\widetilde{\mathcal{K}}_{\bn}^{0,2,h}(\pp_1,\pp_2,j_1,j_2)$ the contribution to $\hat{\mathcal{K}}_\bn^{0,2,h}(\pp_1,\pp_2,j_1,j_2)$ given by the graphs with at least one $\nu$ or $V$ point:
\begin{equation}\label{eq:Kappatildone}
	\widetilde{\mathcal{K}}^{0,s,h}_\bn({\bf p}_1, {\bf p}_2,j_1,j_2) =\sum_{q=1}^\infty
\sum_{\substack{\Gamma \in \mathcal{G}^{R,0,s,h,\mathcal{J}}_{\bn,q}}} 
W_\G(\pp_1,\pp_2,j_1,j_2) \, .
\end{equation}
\begin{corollary}\label{cor:Ktildone}
Let $|\l| < \l_0$ and let $\nu$ be chosen as in Section \ref{subsec:Counterterms}.
Then,
\begin{equation}\label{Colosseo}
	\sup_{\pp_1,\pp_2\in \mathcal{D}_{++}}|\widetilde{\mathcal{K}}^{0,2,h}_\bn(\pp_1,\pp_2,j_1,j_2)| \;  
\leq C \gamma^{\frac{h}{4}} e^{-\frac{\eta}{8} |\bn|}  \, .
\end{equation}
\end{corollary}
\begin{proof}
	Repeating the argument of Lemma \ref{lem:Sandokan}, one has to estimate $W_{\Gamma,l}$ for graphs that have $q\geq 1$. One gets
a bound identical to 
\eqref{asso1LAT} with $\left(\prod_{T\,\text{n.t.}} \gamma^{\frac{5}{4}h_{ T} M^\n_{T}}\right)$ replacing $\left(\prod_{T\,\text{n.t.}} \gamma^{h_{ T} M^\n_{T}}\right)$ (we used Remark \ref{rem:GainNu} to estimate the $\nu$ vertices, \textit{i.e.}\ $|\nu_h| \leq \gamma^{\frac{h}{4}} |\lambda|$ and Remarks \ref{rem:CostumeDaBagno} and \ref{rem:Caligola} to estimate resonant $V$ vertices as $|\mathcal{R} \hat{V}_0| \leq \gamma^{\frac{5}{4}h_T}$). We also decompose the exponential as in \eqref{Samurai} and from \eqref{asxq} we keep the factor $\prod_{T \text{ n.t.}} \gamma^{h_T M_T^I} \leq \prod_{T \text{ n.t.}} \gamma^{\frac{1}{4}h_T M_T^I} $. Using now \eqref{eq:SummationOverNKappaStorti}, and \eqref{eq:CastoriSelvaggi}, 
we get
\begin{equation}
\begin{split}
	| \chi_\G W_{\G,l}(\pp_1,\pp_2,j_1,j_2) |	
 &\leq (c C_1 C_2 \tilde{C}_3)^q 
|\l|^q e^{-\frac{\eta}{8}|\bn|}\Bigg(\prod_{v \, \text{n.s.}} e^{-\frac{\eta}{8} |\bn_v|} \Bigg)\Bigg(\prod_{v\, \text{special}} e^{-\frac{\eta}{8} |\bn_v|} \Bigg)\\ 
&\times\Bigg(\prod_{T \, \text{n.t.}} \gamma^{h_{T}^{\text{ext}}-h_{T}}\Bigg) 
\prod_{T\,\text{n.t.}} \gamma^{\frac{h_T}{4} (M^\n_{T}+M_T^I)}
  \, .
 \end{split}
\end{equation}
We now split 
\begin{equation}
	\prod_{T \, \text{n.t.}} \gamma^{h_{T}^{\text{ext}}-h_{T}}=\Bigg(\prod_{T \, \text{n.t.}} \gamma^{\frac{3}{4}(h_{T}^{\text{ext}}-h_{T})} \Bigg) \prod_{T \, \text{n.t.}} \gamma^{\frac{1}{4}(h_{T}^{\text{ext}}-h_{T})}  \, ,
\end{equation}
and since, by hypothesis, $q\geq 1$ for at least one cluster, one has $M_T^\nu+M_T^I \geq 1$. Therefore, using the telescopic sum, we can bound
\begin{equation}
	\Bigg(\prod_{T \, \text{n.t.}} \gamma^{\frac{1}{4}(h_{T}^{\text{ext}}-h_{T})} \Bigg) \prod_{T \, \text{n.t.}} \gamma^{\frac{1}{4}h_T(M_T^\nu+M_T^I)} \leq \gamma^{\frac{h_\Gamma}{4}}=\gamma^{\frac{h}{4}} \, .
\end{equation}
At this point, one has
\begin{equation}
	\SVert \chi_\G W_{\G,l} \SVert	
 \leq (c C_1 C_2 \tilde{C}_3)^q 
|\l|^q e^{-\frac{\eta}{8}|\bn|} \gamma^{\frac{h}{4}}\Bigg(\prod_{v \, \text{n.s.}} e^{-\frac{\eta}{8} |\bn_v|} \Bigg)\Bigg(\prod_{v\, \text{special} } e^{-\frac{\eta}{8} |\bn_v|} \Bigg)\prod_{T \, \text{n.t.}} \gamma^{\frac{3}{4}(h_{T}^{\text{ext}}-h_{T})}
  \, .
\end{equation}
and one can sum over all $\bn_v$'s and one proceeds as in the proofs of Lemma \ref{prop:EstimateResonantClusters} to sum over scales to get \eqref{Colosseo}.
\end{proof}

\subsection{The decay of the energy correlations}

Before starting the analysis, let us recall that $\mu=O(\beta-\beta_c)$ and, in particular, $\mu=0$ identifies the critical temperature (see Remark \ref{rem:GainNu1}).

We have now to consider the energy correlation 
$
S(\xx_1,j_1;\xx_2,j_2)$ given by \pref{LeEquazioniTutteUguali}.  
We consider first the infinite volume limit ${\partial^2\over \partial A_{\xx_1,j_1} \partial A_{\xx_2,j_2}}
W_{{\boldsymbol\a },l}(A)$:
\begin{equation}\label{ter1}
{ \partial^2\over \partial A_{\xx_1,j_1} \partial A_{\xx_2,j_2}}
W_{{\boldsymbol\a},l}
(A)|_{A=0} =\sum_{h=h^*}^2 \sum_{\bn{\in \mathbb{Z}^2}} 
e^{-2 \pi \ii \Omega \bn \cdot \xx_2}
\mathcal{K}^{0,2,h}_\bn({\bf x}_1-\bx_2,j_1,j_2 ) \, .
\end{equation}
where $h^*=\log_\g \mu$ and
\begin{equation}
	\mathcal{K}^{0,2,h}_\bn(\bx_1-\bx_2,j_1,j_2):= \, {1\over |\L_i|}\sum_{\pp} e^{-\ii \pp\cdot (\bx_1-\bx_2)}
\widehat{\mathcal{K}}^{0,2,h}_\bn(\pp,\pp+2\pi \Omega \bn,j_1,j_2).
\end{equation}

Indeed, one can write $\mathcal{K}^{0,2,h}_\bn(\bx_1-\bx_2,j_1,j_2)$ as 
the sum over graphs in coordinate space. On each graph, the constraint between the labels $\bn_v$ and the scales $h_\ell$ remains unchanged, due to the presence of the $\chi_{\G}$ function.

Due to the Gevrey regularity of the cutoff function $\chi$ defined in \eqref{ChiLHaVisto}, there exist constants $C, \kappa >0$ such that, for $k>h^*$, the propagator obeys to the bounds, see e.g.\ Appendix A of \cite{GMR}
\begin{equation}\label{eq:PropagAttore}
 |g^{(k)}(\bx)|\le C\g^k e^{-\k (\g^k |\bx|)^{1\over 2}}
\end{equation}
and
\begin{equation}
|g^{(\le h^*)}(\bx)|\le C\g^{h^*} e^{-\k (\g^{h^*} |\bx|)^{1\over 2}} \, .
\end{equation}
Note that $\g^{h^*}=O(|\m|)$ for small $\m$.
In the analysis of the graphs in coordinate space, we use \eqref{eq:PropagAttore} to bound each propagator. Fixing $\bx_1$, the $L_1$ norm is therefore bounded exactly as in the proof of Lemma \ref{lem:Sandokan}.

Regarding the bound on the point-wise norm (i.e.\ when both $\bx_1$ and $\bx_2$ are fixed), 
we can write, if $\bar h$ is the scale of the smallest cluster $T\subset \G$ such that $S_{T}=2$, for $k \geq \bar{h}$, 
$ e^{-\k (\g^k |\bx|)^{1\over 2}}
\le e^{-\k/2 (\g^{\bar h} |\bx|)^{1\over 2}}e^{-\k/2 (\g^k |\bx|)^{1\over 2}}
$
so that
we can extract a factor 
$e^{-\k/2 (\g^{\bar h} |\bx|)^{1\over 2}}$ from each propagator.
Moreover there is an extra $\g^{2\bar h}$ in the bound due to the lack 
of sum over the coordinates so that
\begin{equation}
\left|{ \partial^2\over \partial A_{\xx_1,j_1} \partial A_{\xx_2,j_2}}
W_{{\boldsymbol\a},l}
(A)|_{A=0} \right|\le \sum_{\bar h=h^*}^2 C \g^{2 \bar h}  e^{-\k/2 (\g^{ \bar h} |\bx_1-\bx_2|)^{1\over 2}}\le  C_1 e^{-\k_1 (| \m| |\bx_1-\bx_2|)^{1\over 2}} \, ,
\end{equation}
for some constant $\k_1>0$. In deriving the above expression we have used that the sum over all the scales can be done fixing $\bar h$ instead of $h$.

To get a sharper estimate in the case $\m=0$, we can split ${ \partial^2\over \partial A_{\xx_1,j_1} \partial A_{\xx_2,j_2}}
W_{{\boldsymbol\a},l}
(A)|_{A=0} $
in the contribution with $q\ge 1$ and in the contribution with $q=0$.
The term with $q \ge 1$, according to Corollary \ref{cor:Ktildone}, has an extra $\g^{\bar h\over 4}$.  The term with $q=0$ contains two special vertices, each one of which is associated with a $Z_{k,\bn}^{(j)}$, with $k\ge h$.

In the term with $q=0$ we replace the velocities $a_j^{(k)}$ appearing in the propagators $g^{(k)}$ with $a_j^{(\infty)}$ since the difference $a_j^{(-\infty)}-a_j^{(k)}$ is bounded by $\gamma^{k/4}$ by \eqref{eq:ResonantClustersW1}.

In the same way
we can replace $Z_{k,\bn}^{(j)}$ with 
$Z_{-\io, \bn}^{(j)}$ and the difference is bounded by 
$\gamma^{k/4}$.  Moreover,
$Z^{(j)}_{-\io,0}=Z^{(j)}_{-\io,0,\l=0}+\l F_0(\l)$
and $Z^{(j)}_{-\io,\bn}=\l F_\bn(\l)$ for $\bn\neq0$, with $F_0, F_\bn$ bounded. 
Therefore, we can write
\begin{equation}\label{551}
{\partial^2\over \partial A_{\xx_1,j_1} \partial A_{\xx_2,j_2}}W_{\boldsymbol{\a},l}(A)|_{A=0}=
S_a(\xx_1,j_1;\xx_2,j_2) +S_b(\xx_1,j_1;\xx_2,j_2) \, ,
\end{equation}
where we have included in 
$S_b(\xx_1,j_1;\xx_2,j_2) $ the contributions with $q\ge 1$ and the terms with $q=0$ and containing $a_j^{(-\infty)}-a_j^{(k)}$
or $Z_{k,\bn}^{(j)}-Z_{-\io,\bn}^{(j)}$ so that 
\begin{equation}\label{552}
|S_b(\xx_1,j_1;\xx_2,j_2)|\le C_1
\sum_{\bar h=-\io}^2 \g^{2 \bar{h}+\bar{h}/4} e^{-\k (\g^{\bar h} |\bx_1-\bx_2|)^{1\over 2}}\le {C_2\over |\xx_1-\xx_2|^{2+1/4}} \, .
\end{equation}
In $S_{a}(\xx_1,j_1;\xx_2,j_2)$ are collected the terms with $q=0$ and 
$a_j^{(k)}$, $Z_{\bn,k}^{(j)}$ replaced by their limiting values, so that,
calling $\bar{g}(\bx_1,\xx_2)$ the propagator with velocities $a^{(-\infty)}_j$, we have
\begin{equation}\label{553}
S_a(\bx_1,j_1;\bx_2,j_2)=\sum_{\bn_1,\bn_1 \in \mathbb{Z}^2}Z^{(j_1)}_{-\io,\bn_1}Z^{(j_2)}_{-\io,\bn_2}e^{2\pi \ii \O \bn_1 \cdot \bx_1}e^{2\pi \ii\O \bn_2 \cdot \bx_2}
\sum_{\o \in \pm}
\bar g_{\o,\o}(\xx_1,\xx_2) \bar g_{-\o,-\o}(\xx_2,\xx_1) \, .
\end{equation}
Finally, we have to perform the sum over $\balpha$ in \eqref{eq:CinquePuntoUno}.
First note that $Z$ is non-vanishing; we write
\begin{equation}
Z=\hat Z_{--} Z^0+ \hat Z_{--} \sum_{\boldsymbol\a \in \{\pm\}^2} \t_{\boldsymbol\a} Z^0_{\boldsymbol\a} \left({\hat Z_{\boldsymbol\a}\over \hat Z_{--}}-1\right)
\end{equation}
where $Z^0=Z|_{\l=0}$ denotes the partition function of the Ising model for $\lambda=0$, $\hat{Z}_\balpha=Z_\balpha/Z^0_\balpha$ and $Z^0_{\boldsymbol\a}=Z_{\boldsymbol\a}
|_{\l=0}$. In the limit $i \to \infty$, ${1\over |\L_i|}\log{|Z^0_{\boldsymbol\a}|}$
is independent of boundary conditions if $\b\neq\b_c$, see \textit{e.g.}\ chapter 
IV in \cite{MW2},
{\color{black}
and the limit is reached as $O(e^{-c |\mu| \bar L_i})$
if $\bar L_i:=\min \{L_{i,0},L_{i,1}\}$. Moreover, $Z^0$ is non vanishing for $\beta \neq \beta_c$: indeed, for $\mu<0$, $Z^0_\balpha$ is positive for all $\balpha$; for 
$\mu>0$, $Z^0_\balpha$ is negative for $\balpha=++$ and positive for all other $\balpha$'s.}

We consider now ${\hat Z_{\boldsymbol\a}\over \hat Z_{--}}$;
note that $\o_i$ is the same in 
$\hat Z_{\boldsymbol\a}$ for any $\boldsymbol\a$.
$\log{\hat Z_{\boldsymbol\a}\over \hat Z_{--}}$
is sum of graphs containing at least a difference of 
propagators with different boundary conditions. We choose a point \textcolor{black}{$\bar \bx \in \Lambda_i$} and we decompose the graphs
in a term in which all the sums are in a \textcolor{black}{rectangle} around
$\bar \bx$ of side $L_{i,0}/4$ and  $L_{i,1}/4$ and a \textcolor{black}{remainder}. In the \textcolor{black}{remainder}
there is a product of propagators connecting $\bar \bx$ to a point distant
$O(\bar L_i)$, $\bar L_i:=\min \{L_{i,0},L_{i,1}\}$,
 hence such term is $O(|\l||\L_i|e^{-c |\mu| \bar L_i}  )$.
{\color{black}In the first term
we use 
Poisson summation allowing us to write the propagator as  
the infinite volume limit plus a term depending on boundary
conditions and exponentially decaying in $\bx_1-\bx_2$ when both $\bx_1$ and $\bx_2$ 
are in the  rectangle around
$\bar \bx$ of side $L_{i,0}/4$ and  $L_{i,1}/4$, hence
again we get for it
a bound $O(|\l| |\L_i|e^{-c |\mu| \bar L_i}  )$.}
Therefore,
\begin{equation}\label{eq:Carismatico}
 \left|{\hat Z_{\boldsymbol\a}\over \hat Z_{--}}-1\right|\le C|\l| 
|\L_i|e^{-c |\mu| \bar L_i}  
\end{equation}
 by using the uniform convergence, see Lemma \ref{prop:EstimateResonantClusters}.
This says that
\begin{equation}\label{eq:555AAA}
 c_1 |\hat Z_{--} Z^0| \le  |Z|\le c_2 |\hat Z_{--} Z^0|
\end{equation}
where $c_1, c_2=1+O(\l)$ constants. 

Using that  $2 Z = \sum_\balpha \tau_\balpha Z_\balpha$, we 
can write  \eqref{eq:CinquePuntoUno} as
\begin{equation}\label{eq:CinqueCinqueQUattro}
\begin{split}
S(\xx_1,&j_1;\xx_2,j_2)={\partial^2\over \partial A_{\xx_1,j_1} \partial A_{\xx_2,j_2}}W_{--}(A)|_{A=0} \\
&+
\sum_{\boldsymbol\a}\frac{\t_{\boldsymbol\a}Z_{\boldsymbol\a}}{2Z}
\Bigg[{\partial^2\over \partial A_{\xx_1,j_1} \partial A_{\xx_2,j_2}}W_{\balpha}(A)|_{A=0}
-{\partial^2\over \partial A_{\xx_1,j_1} \partial A_{\xx_2,j_2}}W_{--}(A)|_{A=0}\Bigg] \, .
\end{split}
\end{equation}
where in the first term $Z$
cancels out by \pref{eq:555AAA}.

The graphs contributing to ${\partial^2\over \partial A_{\xx_1,j_1} \partial A_{\xx_2,j_2}}W_\balpha(A)|_{A=0}$
can be also decomposed  as the limit $i\to\io$, 
independent from $\balpha$
and a difference which is vanishing. Indeed the difference
contains a difference of propagators, whose contribution is vanishing
at $|\m|>0$, and a difference of oscillating
factors $e^{\ii 2 \pi \O \bn \cdot \bx}$ which is bounded by 
$|\bx| |\bn| |\o-\o_i|$; note that $|\bx|$ produces an extra $\max_j\{L_{j,i}\}$
and $|\o-\o_i|\le C/\bar L_i^2$ \textcolor{black}{(see Section IV.7 in \cite{D})} while for the sum over $\bn$
one uses the exponential decay of the Fourier coefficients of the potential. Hence the difference vanishes in the limit because we take the limit on sequences of $L_{j,i}$ such that $\lim_{i \to +\infty} L_{1,i}/L_{0,i}=c>0$.
Moreover, if $\mu \neq 0$, as a consequence of \eqref{eq:555AAA} and \eqref{eq:Carismatico} we have that $Z^0_{\balpha'}/Z^0_\balpha=1+O(|\L_i|e^{-c |\mu| \bar L_i} )$
and $\hat Z_{\balpha'}/\hat{Z}_\balpha=1+O(|\l| |\L_i|e^{-c |\mu| \bar L_i})$.
Therefore the second term in \eqref{eq:CinqueCinqueQUattro} vanishes in the limit $i\to\infty$.

{\color{black}The first term in \eqref{eq:CinqueCinqueQUattro} can be decomposed according to \eqref{551} with $S_a$ given by \eqref{553} and $S_b$ satisfying \eqref{552}. Therefore, the first term in the r.h.s.\ of \eqref{eq:SdelTeorema} is given by $S_a$ and the decay in $|\bx_1-\bx_2|$ of $R_{j_1,j_2}(\bx_1,\bx_2)$ is given by \eqref{552}. This concludes the proof of Theorem \ref{mainthm}.}

\vspace{1.2pt}

\appendix

{\color{black}
\section{Proof of Lemma \ref{lem:LemmaImpestato}}\label{app:Grafici}
{\color{black}We begin by recalling that if $f(\widehat{\xi})$ is a polynomial in the Grassmann variables $\widehat{\xi}$, $\mathbb{E}_{{\xi}}(f)=\int P(\ud {{\xi}}) f({\widehat{\xi}})$. We then recall that}, for  $m \in \mathbb{N}$, Wick's theorem states that
\begin{equation}\label{eq:A-Wick}
	\mathbb{E}_{{\xi}}\big(\widehat{\xi}_{\bk_1,\sigma_1} \widehat{\xi}_{\pp_1,\rho_1} \cdots \widehat{\xi}_{\bk_{m},\sigma_m} \widehat{\xi}_{\pp_m,\rho_m}\big)=\sum_{p \in \mathfrak{S}_m} (-1)^{\text{sgn}(p)} \prod_{j  = 1}^m \mathbb{E}_{\xi}\big(\widehat{\xi}_{\bk_j,\sigma_j} \widehat{\xi}_{\pp_{p(j)},\rho_{p(j)}}\big)
\end{equation}
where for all $j$, $(\pp_j)_0<0$, $(\bk_j)_0>0$ and $\sigma_j,\rho_j \in \{\pm\}$ and $\mathfrak{S}_m$ denotes the set of permutations of $m$ elements. From the definition of propagator of a Grassmann Gaussian measure {\color{black}(see e.g.\ \cite[eq.\ (4.11)]{GM})}, one has
\begin{equation}\label{eq:PastaFagioliETrippaSarda}
	\mathbb{E}_\xi(\widehat{\xi}_{\bk,\sigma} \widehat{\xi}_{\pp,\rho})=[g^{(\xi)}(\bk)]_{\sigma,\rho} \delta_{\bk,-\pp} |\Lambda_i|\, .
\end{equation}

For $q \in \mathbb{N}$, let us compute $\mathbb{E}_\xi(V^q)$. Using linearity of the expectation and the explicit form of $V(\psi,\xi)$ (given in \eqref{eq:2.40}, \eqref{eq:2.40VERA}, \eqref{eq:2.41} and \eqref{eq:2.42}), we can write
\begin{equation}\label{eq:Bandumba}
	\mathbb{E}_\xi(V^q)=\sum^*\mathbb{E}_\xi\Big(\prod_{r=1}^q M_{\#_r}(\bk_r,\bn_r,\sigma_r,\rho_r,j_r)\Big)
\end{equation}
where $\sum^*$ is the sum over $\bk_r$, $\bn_r$, $\sigma_r$, $\rho_r$, $\#_r\in\{\psi,\xi,(Q,L),(Q,R)\}$, $j_r \in \{0,1\}$. The monomials $M_\sharp$ are defined as
\begin{equation}\label{eq:Abba}
	M_\psi(\bk,\bn,\sigma,\rho,j) = \frac{1}{|\Lambda_i|} \widehat{\psi}_{-\bk,\sigma} A_\bn^{(j)}[P_\psi^{(j)}(\bk,\bn)]_{\sigma,\rho} \widehat{\psi}_{\bk-2 \pi \Omega \bn,\rho} \, ,
\end{equation}
\begin{equation}\label{eq:Acca}
	M_\xi(\bk,\bn,\sigma,\rho,j) = \frac{1}{|\Lambda_i|}\widehat{\xi}_{-\bk,\sigma} A_\bn^{(j)} [P^{(j)}(\bk,\bn)]_{\sigma,\rho} \widehat{\xi}_{\bk-2 \pi \Omega \bn,\rho} \, ,
\end{equation}
\begin{equation}\label{eq:Adda}
	\begin{split}
	M_{Q,L}(\bk,\bn,\sigma,\rho,j) &= \frac{1}{|\Lambda_i|}\widehat{\psi}_{-\bk,\sigma} A_\bn^{(j)} [Q_\psi^{(j)}(\bk,\bn)]_{\sigma,\rho} \widehat{\xi}_{\bk-2 \pi \Omega \bn,\rho} \, , \\ M_{Q,R}(\bk,\bn,\sigma,\rho,j)&=\frac{1}{|\Lambda_i|}\widehat{\xi}_{-\bk,\sigma} A_\bn^{(j)} [Q_\psi^{(j)}(\bk,\bn)]_{\sigma,\rho} \widehat{\psi}_{\bk-2 \pi \Omega \bn,\rho} \, ,
	\end{split}
\end{equation}
and each one of them can be represented as in Fig.\ \ref{fig:AiutoAiuto} by associating a dashed line to each $\psi$ variable and a solid line to each $\xi$ variable.
\begin{figure}[h!]
	\begin{center}
		\begin{tikzpicture}[thick,scale=0.7]
			
			\draw[dashed] (0,0) -- (2,0);
			\draw[fill] (2,0) circle[radius=0.09];
			\draw[snake it] (2,0) -- (2,2);
			\draw[dashed] (2,0) -- (4,0);
			
			\draw[-] (5,0) -- (6,0);
			\draw[-] (6,0) -- (7,0);
			\draw[fill] (7,0) circle[radius=0.09];
			\draw[-] (7,0) -- (8,0);
			\draw[-] (8,0) -- (9,0);
			\draw[snake it] (7,0) -- (7,2);
			
			\draw[dashed] (10,0) -- (12,0);
			\draw[fill] (12,0) circle[radius=0.09];
			\draw[-] (12,0) -- (13,0);
			\draw[-] (13,0) -- (14,0);
			\draw[snake it] (12,0) -- (12,2);
						
			\draw[-] (15,0) -- (16,0);
			\draw[-] (16,0) -- (17,0);
			\draw[fill] (17,0) circle[radius=0.09];
			\draw[dashed] (17,0) -- (19,0);
			\draw[snake it] (17,0) -- (17,2);
			
			\node at (2,-0.7) {$M_\psi$};
			\node at (7, -0.7) {$M_{\xi}$};
			\node at (12,-0.7) {$M_{Q,L}$};
			\node at (17, -0.7) {$M_{Q,R}$};			
		\end{tikzpicture}
	\end{center}
	\caption{Vertices associated to monomials in \eqref{eq:Abba}, \eqref{eq:Acca} and \eqref{eq:Adda}.} \label{fig:AiutoAiuto}
\end{figure}
To compute each of the expectations on the r.h.s.\ of \eqref{eq:Bandumba} we use Wick's theorem \eqref{eq:A-Wick} and therefore $\mathbb{E}_\xi(M_{\#_1} \cdots M_{\#_q})$ reduces to the sum over permutations of products of expectations of pairs of $\xi$ variables. Each of such summand has a graphical interpretation obtained as follows. First, one draws the vertices in Fig.\ \ref{fig:AiutoAiuto} corresponding to the monomials $M_{\#_1},\dots,M_{\#_q}$. Second, one connects the solid lines corresponding to a pair $(\widehat{\xi}_{\bk_j, \sigma_j}, \widehat{\xi}_{\pp_{i}, \rho_i})$ whenever the expectation $\mathbb{E}_\xi(\widehat{\xi}_{\bk_j, \sigma_j} \widehat{\xi}_{\pp_{i}, \rho_i})$ appears in the product of expectations of the summand.
In this way one obtains a graphical object that we define as \emph{unordered graph}. As a consequence of this correspondence, one writes $\mathbb{E}_\xi(M_{\#_1}\cdots M_{\#_q})$ as the sum over all unordered graphs obtained by contracting all solid lines of the graphical elements associated to $M_{\#_1},\dots,M_{\#_q}$.
Note that there are two types of unordered graphs: \emph{connected} and \emph{disconnected}. 

We now write truncated expectations in terms of expectations. 
Let us consider $S=\{1,\cdots,s\}$ and let us denote by $\mathscr{P}_p$ the set of all possible partitions of $S$ into $p$ pairwise disjoint subsets. Define for any $I \subseteq S$
\begin{equation}
	\mathbb{E}^T(M_I)=\mathbb{E}^T(M_{i_1};\dots;M_{i_r}) \, , \qquad I=\{i_1, i_2, \dots, i_r\} \, ,
\end{equation} 
where each of the $M_i$'s is one of the monomials in \eqref{eq:Abba}, \eqref{eq:Acca} and \eqref{eq:Adda}.

One can prove (see eq.\ 2.100 in \cite{MastropietroNonPerturbative}) that the following formula connects expectations and truncated expectations:
\begin{equation}\label{eq:GenericalTruncated}
	\mathbb{E}^T_\xi(M_{\#_1};\dots;M_{\#_s})= \mathbb{E}_\xi(M_{\#_1} \cdots M_{\#_s})- \sum_{p=2}^{s} \sum_{\{I_1,\cdots,I_p\} \in \mathscr{P}_p} \prod_{j=1}^p \mathbb{E}^T_\xi(M_{I_j}) \, .
\end{equation}
Using this formula, one can prove inductively that the truncated expectations are obtained as the sum over \emph{connected} unordered graphs only. 

Using multilinearity of the truncated expectations, one has
\begin{equation}\label{eq:AbbassoICanotti}
	\mathbb{E}_\xi^T(V;q)=\sum^* \mathbb{E}_\xi^T(M_{\#_1}; \cdots; M_{\#_q}) \, ,
\end{equation}
where the $\sum^*$ is the same as in \eqref{eq:Bandumba} and the dependence on all parameters is understood. From \eqref{eq:AbbassoICanotti} and the observation after \eqref{eq:GenericalTruncated} we obtain a representation of $\mathbb{E}^T(V;q)$ in terms of connected unordered graphs. 

When $q=1$, $\mathbb{E}_\xi^T(V(\psi,\xi);q)=S_{\text{int}}^{(\psi)}(\psi)$ which gives the first term in the sum \eqref{eq:Wstortoq1} with $W_\Gamma$ given by the definition \eqref{eq} in the case $q=1$.
	
For the case $q \geq 2$, one notices that with the vertices in Fig.\ \ref{fig:AiutoAiuto} one can make only two types of connected graphs: either one picks $q$ vertices of type $M_\xi$, or one picks two vertices of type $M_{Q,L}$, $M_{Q,R}$ and $q-2$ vertices of type $M_\xi$.
	
	In the first case, the value of the associated truncated expectation does not depend on $\psi$ and it contributes to $E^\xi$. 
	
	In the second case, one first notices that each of the entries of the truncated expectation on the r.h.s.\ of \eqref{eq:AbbassoICanotti} is a quadratic monomial in the Grassmann variables, and then it commutes with the monomials on each other entry. Fixing an order of the entries in the truncated expectation produces a $q!$ in front and restricts the sum over the graphs of Definition \ref{def:Grafici111}. Last, the relation between $\mathbb{E}_\xi^T(M_{\#_1}; \cdots; M_{\#_q})$ and \eqref{eq:Wstortoq1} (with $W_\Gamma$ defined as in \eqref{eq}) follows from \eqref{eq:Abba}, \eqref{eq:Acca}, \eqref{eq:Adda}, \eqref{eq:A-Wick} and \eqref{eq:PastaFagioliETrippaSarda} after noting that each unordered connected graph has the same sign. Indeed, it is sufficient to take the sign appearing in Wick's theorem \eqref{eq:A-Wick} and to note that to keep the graph connected one must always exchange an even number of Grassmann variables.

\section{Derivation of \pref{gra}}
\label{app:Sym}

\noindent
\textcolor{black}{The value of a graph $W_\G$ is product of complex valued scalar functions $A_\bn$} and matrices of the form
$\begin{pmatrix}
		a_\bn(\bk) & b_\bn(\bk) \\
		b_\bn^*(\bk) & -a_\bn^*(\bk)
	\end{pmatrix}$
such that $a_\bn(\bk)=-a_{-\bn}(-\bk) \in \mathbb{C}$
and $b_{\bn}(\bk)=b_{-\bn}(-\bk) \in \ii \mathbb{R}$.  This follows by induction.
For the graph with three vertices we have to consider the product of three matrices $A_{\bn_A}(\bk),g_{\chi}(\bk-2 \pi \Omega \bn_A),B_{\bn_B}(\bk-2 \pi \Omega \bn_A)$. Here $A$ and $B$ can be either $P^{(j)}$ or $Q^{(j)}$. The following explicit computation yields
\[
	\begin{split}
		A_{\bn_A}(\bk) g_\chi(\bk-2 \pi \Omega \bn_A) &= \begin{pmatrix}
			a_{\bn_A}^{(A)}(\bk) & b_{\bn_A}^{(A)}(\bk) \\
			(b_{\bn_A}^{(A)}(\bk))^* & -(a_{\bn_A}^{(A)})^*
		\end{pmatrix}
	 \begin{pmatrix}
			a^{(\xi)}(\bk-2 \pi \Omega \bn_A) & {b}^{(\xi)}(\bk-2 \pi \Omega \bn_A) \\
			(b^{(\xi)}(\bk-2 \pi \Omega \bn_A))^* & -(a^{(\xi)}(\bk-2 \pi \Omega \bn_A))^* \\
			\end{pmatrix}	 \\
			&= \begin{pmatrix}
				\beta(\bk,\bn_A) & \alpha(\bk,\bn_A) \\
				-\alpha^*(\bk,\bn_A) & \beta^*(\bk, \bn_A)
			\end{pmatrix} \, .
	\end{split}
\]
where, explicitly,
\[
	\begin{split}
		\alpha(\bk,\bn_A)&:= a^{(A)}_{\bn_A}(\bk) b^{(\xi)}(\bk-2 \pi \Omega \bn_A)-b^{(A)}_{\bn_A}(\bk) (a^{(\xi)}(\bk-2 \pi \Omega \bn_A))^* \, ,\\
		\beta(\bk,\bn_B)&:= a_{\bn_A}^{(A)}(\bk) a^{(\xi)}(\bk-2 \pi \Omega \bn_A)+b_{\bn_A}^{(A)}(\bk) (b^{(\xi)}(\bk-2 \pi \Omega \bn_A))^* \, .
	\end{split}
\]
It follows now from the symmetry properties of $a$ and $b$ that $\alpha(\bk,\bn_A)=-\alpha(-\bk,-\bn_A)$ and $\beta(\bk,\bn_A)=\beta(-\bk,-\bn_A)$.

Computing the value of the graph, one has
\[
	\begin{split}
		A_{\bn_A}(\bk) &g_\chi(\bk-2 \pi \Omega \bn_A) B_{\bn_B}(\bk-2 \pi \Omega \bn_A)= \\
		&=\begin{pmatrix}
				\beta(\bk,\bn_A) & \alpha(\bk,\bn_A) \\
				-\alpha^*(\bk,\bn_A) & \beta^*(\bk, \bn_A)
			\end{pmatrix}
	 \begin{pmatrix}
			a^{(B)}_{\bn_A}(\bk-2 \pi \Omega \bn_A) & {b}^{(B)}_{\bn_B}(\bk-2 \pi \Omega \bn_A) \\
			(b^{(B)}_{\bn_B}(\bk-2 \pi \Omega \bn_A))^* & -(a^{(B)}_{\bn_B}(\bk-2 \pi \Omega \bn_A))^* \\
			\end{pmatrix} \\
			&=\begin{pmatrix}
				a_\bn(\bk) & b_\bn(\bk) \\
				b_\bn^*(\bk) & - a_\bn^*(\bk)
			\end{pmatrix}
	\end{split}
\]
with 
\[
	\begin{split}
		a_\bn(\bk)&=\beta(\bk,\bn_A) a_{\bn_A}^{(B)}(\bk-2 \pi \Omega \bn_A)+\alpha(\bk, \bn_A) (b^{(B)}_{\bn_B}(\bk-2 \pi \Omega \bn_A))^* \, , \\
		b_\bn(\bk)&=\beta(\bk,\bn_A) b_{\bn_B}^{(B)}(\bk-2 \pi \Omega \bn_A)-\alpha(\bk,\bn_A) (a_{\bn_B}^{(B)}(\bk-2 \pi \Omega \bn_A))^* \, .
	\end{split}
\]
From the symmetry properties of $a,b,\alpha$ and $\beta$ it follows that under the exchange $\{\bn_A,\bn_B, \bk\} \mapsto \{-\bn_A,-\bn_B,-\bk\}$ one has that $a_{\bn}(\bk) =-a_{-\bn}(-\bk)$ and $b_{\bn}(\bk)=b_{-\bn}(-\bk)$. This completes the proof for the graph with two vertices.
By the inductive hypothesis we assume that
the property holds for the product of the matrices associated to the sub-graph of the first $q-1$ points, and repeat the argument.

We note now that, calling $\mathrm{Val}(\Gamma_{\{\underline \bn_v\}}(\bk))$
the value $W_\G$ with $\G\in \mathcal{G}_{\bn,q}$
\begin{equation}\label{eq:RisommazioneFurba}
	\begin{split}
		\sum_{\bk \in \mathcal{D}_{\balpha}} \bpsi_{-\bk} \cdot \mathcal{V}_0(\bk) \bpsi_\bk &=\frac{1}{4}\sum_{\bk \in \mathcal{D}_\balpha} \bpsi_{-\bk} \cdot \sum_{\Gamma_{\{\underline \bn_v\}} \in \mathcal{G}_{0,q}^{ \xi}}\left[\mathrm{Val}(\Gamma_{\{\underline \bn_v\}}(\bk))-\left(\mathrm{Val}(\Gamma_{\{-
\underline \bn_v
\}}(-\bk))\right)^T+\right. \\
		&\qquad \left.+\mathrm{Val}(\Gamma_{\{-\underline \bn_v\}}(\bk))-\left(\mathrm{Val}(\Gamma_{\{\underline \bn_v\}}(-\bk))\right)^T \right] \bpsi_\bk \, .
	\end{split}
\end{equation}
 We start by computing
\begin{equation}
	\begin{split}
		\mathrm{Val} \, \Gamma_{\{\underline \bn_v\}}(\bk) - \mathrm{Val} \, \Gamma_{\{-\bn_n\}}(-\bk)^T \;&=\; \left(\prod_{v \in \Gamma} \widehat{A}_{\bn_v}^{(j_v)} \right) G_{\{\underline \bn_v\}}(\bk) -\left(\prod_{v \in \Gamma} \widehat{A}_{-\underline\bn_v}^{(j_v)} \right) \left(G_{\{-\bn_v\}}(-\bk)\right)^T \\
	\end{split}
\end{equation}
It is convenient to call $f_\bn:=\prod_{v \in \Gamma} \widehat{A}_{\bn_v}^{(j_v)} $. Then, using that $f_{-\bn}=f^*_{\bn}$, we have
\begin{equation}
	\begin{split}
		\mathrm{Val} \, \Gamma_{\{\underline \bn_v\}}(\bk) &- \left(\mathrm{Val} \, \Gamma_{\{-\underline \bn_v\}}(-\bk)\right)^T\\
		& \;=\; f_\bn \begin{pmatrix}
		a_\bn(\bk) & b_\bn(\bk) \\
		b^*_\bn(\bk) & - a^*_\bn(\bk) 
	\end{pmatrix}-f^*_{\bn} \begin{pmatrix}
		a_{-\bn}(-\bk) & b^*_{-\bn}(-\bk) \\
		{b_{-\bn}(-\bk)} & - a^*_{-\bn}(-\bk)
	\end{pmatrix} \\
	&\!\!\!\!\!\!\!\!=\; f_\bn \begin{pmatrix}
		a_\bn(\bk) & b_\bn(\bk) \\
		b^*_\bn(\bk) & - a^*_\bn(\bk) 
	\end{pmatrix}-{f_{\bn}^*} \begin{pmatrix}
		-a_{\bn}(\bk) & b^*_{\bn}(\bk) \\
		{b_{\bn}(\bk)} &  a^*_{\bn}(\bk)
	\end{pmatrix} \\
	&\;=\;\begin{pmatrix}
		(f_\bn+f^*_\bn) a_\bn(\bk) & f_\bn b_\bn(\bk)-f^*_\bn b^*_{\bn}(\bk) \\
		f_\bn b^*_\bn(\bk) - f^*_\bn b_\bn(\bk) & -(f_\bn(\bk)+f^*_\bn) a^*_\bn(\bk)
	\end{pmatrix} =\; \begin{pmatrix}
		\alpha_\bn(\bk) & \beta_\bn(\bk) \\
		\beta^*_\bn(\bk) & - \alpha^*_\bn(\bk)
	\end{pmatrix}
	\end{split}
\end{equation}
with
\begin{equation}\label{eq:AlphaBetaCoefficients}
		\alpha_\bn(\bk)=(f_\bn+f^*_\bn) a_\bn(\bk) \quad\quad
		\beta_\bn(\bk)=f_\bn b_\bn(\bk)-f^*_\bn b^*_{\bn}(\bk) \, .
\end{equation}
With those explicit expressions at hand, it is clear that $\beta_\bn(\bk) \in \ii \mathbb{R}$ and $\alpha_\bn(\bk) \in \mathbb{C}$, in general.
Finally
\begin{equation}
	\begin{split}
		\mathrm{Val} \, &\Gamma_{\{\underline \bn_v\}}(\bk)- \mathrm{Val} \, \Gamma_{\{-\underline \bn_v\}}(-\bk)^T+\mathrm{Val} \, \Gamma_{\{-
\underline \bn_v\}}(\bk)-\mathrm{Val} \, \Gamma_{\{-\underline \bn_v\}}(-\bk)^T \\
		&=\;\begin{pmatrix}
		\alpha_{\bn}(\bk) & \beta_{\bn}(\bk) \\
		 {\beta^*_{\bn }(\bk)} & -  {\alpha^*_{\bn }(\bk)}
	\end{pmatrix}+\begin{pmatrix}
		\alpha_{-\bn }(\bk) & \beta_{-\bn }(\bk) \\
		 {\beta^*_{-\bn }(\bk)} & -  {\alpha^*_{-\bn }(\bk)}
	\end{pmatrix} \\
	&=\begin{pmatrix}
		\alpha_{\bn }(\bk)+\alpha_{{-\bn }}(\bk) & \beta_{\bn }(\bk)+\beta_{-\bn }(\bk) \\
		 {\beta^*_{\bn }(\bk)}+ {\beta^*_{-\bn }(\bk)} & -  {\alpha^*_{\bn }(\bk)}- {\alpha^*_{-\bn }(\bk)}
	\end{pmatrix} \\
	&=\begin{pmatrix}
		\alpha_{\bn }(\bk)-\alpha_{\bn }(-\bk) & \beta_{\bn }(\bk)+\beta_{\bn }(-\bk) \\
		 {\beta^*_{\bn }(\bk)}+ {\beta^*_{\bn }(-\bk)} & -  {\alpha^*_{\bn }(\bk)}+ {\alpha^*_{\bn }(-\bk)}
	\end{pmatrix} \, .
	\end{split}
\end{equation}

\begin{remark} In the case of  layered disorder constant in one direction, say $\mathbf{e}_0$, the theory is translation invariant
in one direction and 
one has to the additional property that $a_\bn \in \ii \mathbb{R}$, implying that the velocities are real.
Indeed only 
$P^{(1)}$ matrices are present and, as we are interested in the case $k_0=k_1=0$ with $\omega_1 \neq 0$. This implies that also the entries of the propagator are purely imaginary, and the product of an odd number of imaginary numbers is imaginary.
\end{remark}
\textcolor{black}{
\begin{remark}
With respect to the symmetries of Section II.D in \cite{M3} we break manifestly symmetries 1)-3) (which are, in order, parity, diagonal reflection and orthogonal reflection). Moreover, note that all kernels $K_\bn(\bk)$ appearing in Section \ref{sec:Grassmann} are such that $K_\bn(\bk)=[K_{-\bn}(-\bk)]^*$ which is nothing but the symmetry by complex conjugation (i.e., the symmetry 4) in Section II.D of \cite{M3}.
\end{remark}
}

\section{Action of the $\mathcal{R}$ operation}
\label{appendix:LemmaGG1}
In this appendix, for $j\in\mathbb{N}$, we denote by $\mathbf{R}_j$ the set of resonant clusters strictly contained in $\mathbf{R}_{j-1}$ and not in any other \emph{resonant} cluster. (We also denote by $\mathbf{R}:=\bigcup_{j=1}^{+\infty} \mathbf{R}_j$.)
Denoting with $\mathbf{R}_2$ the set of maximal resonances contained in $\mathbf{R}_1$, the value of renormalized resonant cluster can now be estimated as, if $\tilde T$ is a resonance
\begin{equation}\label{eq:RHS_renormalized}
\left\Vert \mathcal{R} W_{\tilde{T}}^{(h_{\tilde{T}})}(\bk_{\tilde{T}}) \right\Vert \leq \sup_{t \in [0,1]} \left\Vert \frac{\ud^2}{\ud t^2} W^{(h_{\tilde{T}})}_{\tilde{T}}(t \bk_T) \right\Vert
\end{equation}

One has now to analyze what happens when a derivative acts on a renormalized cluster.  

If two derivatives corresponding to a resonance $\tilde{T}$ acts on the value of some renormalized resonant cluster $\tilde{T}' \subset \tilde{T}$, recalling that $\bk_{\tilde{T}'}=t \bk+\bq$ for suitable $\bq$, one has
\begin{equation}
	\begin{split}
		\frac{\ud^2}{\ud t^2} \mathcal{R}W^{(h_{\tilde{T}'})}_{\tilde{T}'}(t\bk+\bq) \;&=\; \frac{\ud^2}{\ud t^2} \left[W^{(h_{\tilde{T}'})}_{\tilde{T}'}(t\bk+\bq)-W^{(h_{\tilde{T}'})}_{\tilde{T}'}(0)-(t \bk+\bq)\cdot \partial_\bk W^{(h_{\tilde{T}'})}_{\tilde{T}'}(0) \right] \\
		&=\;\frac{\ud^2}{\ud t^2} W^{(h_{\tilde{T}'})}_{\tilde{T}'}(t \bk+\bq) \, .
	\end{split}
\end{equation}
If one derivative acts on a renormalized cluster, we have instead
\begin{equation}
		\frac{\ud}{\ud t} \mathcal{R}W^{(h_{\tilde{T}'})}_{\tilde{T}'}(t \bk+\bq) \;=\;  \int_0^1 \frac{\ud}{\ud t} \frac{\ud}{\ud s} W^{(h_{\tilde{T}'})}_{\tilde{T}'}(s(t \bk+\bq)) \, \ud s \, .
\end{equation}
Whence we get the two bounds
\begin{equation}
	\left\Vert \frac{\ud^2}{\ud t^2} \mathcal{R}W^{(h_{\tilde{T}'})}_{\tilde{T}'}(t \bk+\bq) \right\Vert = \left\Vert \frac{\ud^2}{\ud t^2} W^{(h_{\tilde{T}'})}_{\tilde{T}'}(t \bk+\bq) \right\Vert  \, ,
\end{equation}
\begin{equation}
	\left\Vert \frac{\ud}{\ud t} \mathcal{R} W^{(h_{\tilde{T}'})}_{\tilde{T}'}(t \bk+\bq) \right\Vert  \leq \sup_{s,t \in [0,1]} \left \Vert \frac{\ud}{\ud t} \frac{\ud }{\ud s}  W^{(h_{\tilde{T}'})}_{\tilde{T}'}( s(t\bk+\bq)) \right\Vert \, .
\end{equation}
So, summarizing, for the estimate we have the following:
\begin{itemize}
\item if two derivatives corresponding to a resonance $\tilde{T}$ act on the value of some resonance $\tilde{T}' \subset \tilde{T}$, one can replace with $\mathbf{1}$ the $\mathcal{R}$ operator; 
\item if one derivative corresponding to a resonance ${\tilde{T}}$ acts on the value of some resonance $\tilde{T}' \subset \tilde{T}$, one can replace with $\frac{\ud}{\ud s}$ the $\mathcal{R}$ operator and take the supremum over $s \in [0,1]$; 
\item if no derivatives act on a resonance, one can replace $\mathcal{R}$ with $\frac{\ud^2}{\ud s^2}$ and take the supremum over $s \in [0,1]$.
\end{itemize}

These remarks permit us to iterate this procedure considering the action of derivatives on resonances inside resonances. Proceeding in this way, we see that the R.H.S. of \eqref{eq:RHS_renormalized}  can be bounded in the following way. We denote by $f$ either a line or a vertex and with $T_R\in \mathbf{R}$ a resonant cluster.
\begin{itemize}
	\item There is one term for each ordered pair $(f_1,f_2)$, with $f_1, f_2 \in T_R$, not necessarily different (\textit{i.e.} it may happen that $f_1=f_2$).
	\item If $f_1 \in {\tilde{T}}_0$ and ${\tilde{T}}$ is a cluster contained in $T_R$, then ${\tilde{T}}=T^{(r)} \subset T^{(r-1)} \subset \dots \subset T^{(1)}=T_R$ is the chain of clusters associated to $f_1$ containing ${\tilde{T}}$ and contained in $T_R$. Similarly, if $f_2 \in \hat{T}_0$ and $\hat{T}$ is a cluster contained in $T_R$, one constructs the chain of clusters associated to $f_2$ containing $\hat{T}$ and contained in $T_R$.
	\item { At this point we replaced the $\mathcal{R}$ operator acting on the cluster $T_R$ with two derivatives.}
	\item { If a resonant cluster belongs to both the chain of clusters (the one associated with $f_1$ and the one associated with $f_2$), then its $\mathcal{R}$ operator is removed. }
	\item { If instead there is a cluster (say, $T_V$) belonging to only one of the chain of clusters, then there is one term for any $f_3 \in T_V$. If $f_3 \in (T_V')_0 \subset T_V$, then one considers the chain of cluster associated to $f_3$, containing $T_V'$ and contained in $T_V$. One replaced the $\mathcal{R}$ operator acting on $T_V$.}
	\item This construction is repeated until all $\mathcal{R}$ operators are { replaced}. At this point each cluster inside a resonance belongs to two chains of vertices.
	\item From their explicit expression, it is also obvious that one can estimate the action of a derivative on a vertex with the action of a derivative on a propagator on the same scale.
	\item Last, the number of terms that are generated in this procedure is estimated by $9^q$ (that is the number of terms generated when each vertex or each line can be derived zero, one or two times without any constraint).
\end{itemize}

Note that, an adaptation of this argument permits to treat the terms $\partial_\bk^s$ appearing in \eqref{eq:BoundOnEffectivePotential1} and \eqref{eq:ResonantClustersW}.

\vspace{0.2cm}

\noindent
\textbf{Acknowledgments.} We thank Rafael L.\ Greenblatt and Marcello Porta for fruitful discussions and the anonymous referees for their fruitful criticism. We acknowledge financial support of the MIUR-PRIN 2017 project MaQuMA cod.\ 2017ASFLJR, the European Research Council (ERC) under the European Union’s Horizon 2020 research and innovation program ERC
StG MaMBoQ, n.802901.
We also thank GNFM, the Italian National Group of Mathematical Physics. M.G.\ acknowledges Dipartimento di Matematica ``F. Enriques'', University of Milan, where part of this work was carried out. V.M.\ acknowledges Institute for Advanced Studies (Princeton) where part of this work was carried out.

\vspace{0.2cm}

\noindent
\textbf{Data availability statement.} This manuscript has no associated data.

\end{document}